\documentclass[11pt]{article}
\usepackage[letterpaper,margin=1.00in]{geometry}
\usepackage{amsmath, amssymb, amsthm, thmtools, amsfonts, dsfont}
\usepackage{bbm}

\usepackage{ifthen}

\usepackage{tikz}
\usetikzlibrary{positioning,decorations.pathreplacing,calligraphy}

\usepackage{cite}
\usepackage{appendix}
\usepackage{graphicx}
\usepackage{color}
\usepackage{algorithm}
\usepackage[noend]{algpseudocode}
\usepackage{epstopdf}
\usepackage[textsize=tiny]{todonotes}

\usepackage{framed}
\usepackage[framemethod=tikz]{mdframed}
\usepackage[bottom]{footmisc}
\usepackage[shortlabels]{enumitem}
\setitemize{noitemsep,topsep=3pt,parsep=3pt,partopsep=3pt}
\usepackage[font=small]{caption}
\usepackage{xspace}

\usepackage{mathtools}

\newtheorem{theorem}{Theorem}[section]
\newtheorem{lemma}[theorem]{Lemma}
\newtheorem{meta-theorem}[theorem]{Meta-Theorem}

\newtheorem{definition}[theorem]{Definition}

\newcommand{\FullOrShort}{full}

\ifthenelse{\equal{\FullOrShort}{full}}{
    
  \newcommand{\fullOnly}[1]{#1}
  \newcommand{\shortOnly}[1]{}
  }{
    \newcommand{\fullOnly}[1]{}
    \newcommand{\shortOnly}[1]{#1}
  }

\definecolor{darkgreen}{rgb}{0,0.5,0}
\definecolor{dkblue}{rgb}{0,0,0.6}
\usepackage{hyperref}
\hypersetup{
    unicode=false,          
    colorlinks=true,        
    linkcolor=dkblue,          
    citecolor=darkgreen,        
    filecolor=magenta,      
    urlcolor=cyan           
}
\usepackage[capitalize, nameinlink]{cleveref}
\crefname{theorem}{Theorem}{Theorems}
\Crefname{lemma}{Lemma}{Lemmas}
\Crefname{observation}{Observation}{Observations}
\Crefname{equation}{}{}

\algnewcommand\algorithmicswitch{\textbf{switch}}
\algnewcommand\algorithmiccase{\textbf{case}}

\algdef{SE}[SWITCH]{Switch}{EndSwitch}[1]{\algorithmicswitch\ #1\ \algorithmicdo}{\algorithmicend\ \algorithmicswitch}%
\algdef{SE}[CASE]{Case}{EndCase}[1]{\algorithmiccase\ #1}{\algorithmicend\ \algorithmiccase}%
\algtext*{EndSwitch}%
\algtext*{EndCase}%

\usepackage{thmtools} 
\usepackage{thm-restate}

\usepackage[textsize=tiny]{todonotes}

\newcommand{\eps}{\varepsilon}

\newcommand{\F}{{\mathcal{F}}}

\newcommand{\R}{\mathbb{R}}

\newcommand{\C}{\mathcal{C}}

\renewcommand{\S}{\mathcal{S}}

\newcommand{\M}{\mathcal{M}}
\newcommand{\U}{\mathcal{U}}
\renewcommand{\H}{\mathcal{H}}

\newcommand{\poly}{{\rm poly}}

\newcommand{\dist}{\mathop{\mbox{\rm dist}}\nolimits}
\newcommand{\cost}{\mathop{\mbox{\rm cost}}\nolimits}

\providecommand{\minl}{\ensuremath{\text{\textrm{min-$\ell$}}}}
\providecommand{\poly}{{\rm poly}}

\providecommand{\polyloglog}{\poly\log\log}

\newcommand{\set}[1]{\left\{#1\right\}}

\newcommand{\balpha}{\bar{\alpha}}
\newcommand{\bbeta}{\bar{\beta}}

\renewcommand{\paragraph}[1]{\medskip\noindent {\bf #1}}

\usepackage{setspace}

\usepackage{mathtools}
\DeclarePairedDelimiter{\abs}{\lvert}{\rvert}

\makeatletter
\let\oldabs\abs
\def\abs{\@ifstar{\oldabs}{\oldabs*}}
\makeatother

\date{}

\title{\huge An Efficient Massively Parallel Constant-Factor Approximation Algorithm for the $k$-Means Problem}

\author{%
  \thanks{, \texttt{}}%
  \and
  Fabian Kuhn\thanks{University of Freiburg, \texttt{bob@insty.edu}}%
  \and
  Carol C. Creator\thanks{Dept.\ of Mathematics, University Z, \texttt{carol@univz.edu}}%
}

\author{
    \begin{minipage}{0.25\textwidth}
        \centering
        Vincent Cohen-Addad\\
        \small Google Research\\
        \small \href{mailto:cohenaddad@google.com}{\color{black}cohenaddad@google.com}
    \end{minipage}
    \and
    \begin{minipage}{0.23\textwidth}
        \centering
        Fabian Kuhn\\
        \small University of Freiburg\\
        \small \href{mailto:kuhn@cs.uni-freiburg.de}{\color{black}kuhn@cs.uni-freiburg.de}
    \end{minipage}
    \and
    \begin{minipage}{0.33\textwidth}
        \centering
        Zahra Parsaeian\\
        \small University of Freiburg\\
        \small \href{mailto:zahra.parsaeian@cs.uni-freiburg.de}{\color{black}zahra.parsaeian@cs.uni-freiburg.de}
    \end{minipage}
}

\date{}

\begin{document}
\maketitle

\setcounter{page}{0}
\thispagestyle{empty}
\begin{abstract}
    In this paper, we present an efficient massively parallel approximation algorithm for the $k$-means problem. Specifically, we provide an MPC algorithm that computes a constant-factor approximation to an arbitrary $k$-means instance in $O(\log\log n \cdot \log\log\log n)$ rounds. The algorithm uses $O(n^\sigma)$ bits of memory per machine, where $\sigma > 0$ is a constant that can be made arbitrarily small. The global memory usage is $O(n^{1+\eps})$ bits for an arbitrarily small constant $\eps > 0$, and is thus only slightly superlinear. Recently, Czumaj, Gao, Jiang, Krauthgamer, and Vesel\'{y} showed that a constant-factor bicriteria approximation can be computed in $O(1)$ rounds in the MPC model. However, our algorithm is the first constant-factor approximation for the general $k$-means problem that runs in $o(\log n)$ rounds in the MPC model.

    Our approach builds upon the foundational framework of Jain and Vazirani. The core component of our algorithm is a constant-factor approximation for the related facility location problem. While such an approximation was already achieved in constant time in the work of Czumaj et al.\ mentioned above, our version additionally satisfies the so-called Lagrangian Multiplier Preserving (LMP) property. This property enables the transformation of a facility location approximation into a comparably good $k$-means approximation.
\end{abstract}

\newpage

\section{Introduction}\label{sec:intro}
The classic $k$-means objective is widely used to model clustering problems in data mining and machine learning applications. Given a set of points in a metric space 
the $k$-means problem asks to identify $k$ points from the metric, called \emph{centers}, 
so as to minimize the sum of the squared distances from the input points to their closest
center. Since its introduction in the late 50s by Lloyd~\cite{lloyd1982least} and 
Max~\cite{max1960quantizing} the $k$-means problem has received a tremendous amount of 
attention in a variety of communities to, e.g., model compression problems or uncover underlying structure in datasets. 

In the overwhelming majority of these applications, the data lies in Euclidean space. In this context, the $k$-means problem has a very natural clustering formulation: The optimum
centers are the means of the clusters they represent and the goal is to minimize some 
notion of \emph{dispersion} within clusters. 
For this reason, a long line of work on approximating $k$-means on massive datasets has emerged. Since the early 2000s, people have studied the $k$-means problem in streaming~\cite{SYZ18,CharikarOP03,AilonJM09}, dynamic~\cite{Meyerson01,Cohen-AddadGKR21}, and distributed/parallel settings \cite{Cohen-AddadMZ22,CzumajGJK024,Cohen-AddadEMNZ22,Cohen-AddadLNSS21,BlellochT10,BlellochGT12,EIM11,BLK18,BEL13}. 

Arguably, the parallel and distributed setting remains the model for which upper and lower bounds 
for the $k$-means problem remain frustratingly loose. Hence, the last few years have witnessed several works aiming
at filling this gap in our understanding of the complexity of the $k$-means problem in this context. 
Focusing on the Massively Parallel Computation (MPC) model~\cite{KarloffSV10,MPC_algandapplications}, people have studied the 
following question.
\begin{mdframed}[hidealllines=true, backgroundcolor=gray!15]
\begin{center}
\textbf{How well can the $k$-means problem be approximated in the MPC model?}\ \\
\end{center}
\end{mdframed}

The question is particularly interesting if we ask for algorithms with sublogarithmic time complexities and if we ask for fully scalable algorithms for which the memory per parallel machine can be made $O(n^\sigma)$ for an arbitrarily small constant $\sigma>0$.
If $k$ is sufficiently small, then some progress was made in \cite{EIM11}. The paper shows that if the memory per machine is at least $O(k^2\cdot n^\eps)$ for some constant $\eps>0$, the related $k$-median and $k$-center problems can be approximated within a constant factor in $O(1)$ MPC rounds.\footnote{The $k$-median problems asks to minimize the sum of the distances as opposed to 
the sum of the squared distances and the $k$-center problem asks to minimize the maximum distance. The $k$-median and $k$-center problems tend to be easier to tackle as not all techniques that apply to $k$-median or $k$-center also work for the $k$-means problem.} The authors of \cite{EIM11} conjecture that using similar techniques, a corresponding result can also be proved for the $k$-means problem. Even if we assume this to be true, the general problem where $k$ can be large remains open.
Unfortunately, modern data mining and machine 
learning applications focus on massive datasets and high-dimensional inputs, see e.g.~\cite{pmlr-v80-bhaskara18a,CzumajGJK024,Cohen-AddadMZ22,Cohen-AddadLNSS21}.
Additionally, widely-used privacy-preserving techniques such as $s$-anonymity often require to find a clustering of the input users 
into small size clusters (of size $s < 100$), which translates in a desired number of clusters $k$ larger than $n/100$~\cite{byun2007efficient}.

Therefore recent work has considered the scenario where $k$ may much be larger than
the memory of each machine and also where the dimension $d$ of the Euclidean space is $\Omega(\log n)$, where $n$ is the number of input points. First,  Cohen{-}Addad, Lattanzi, Norouzi{-}Fard, Sohler, and 
Svensson~\cite{Cohen-AddadLNSS21} showed how to 
obtain an $O(\log \Delta\cdot \log n)$-approximate solution for the $k$-median problem in $O(1)$ parallel rounds
in the MPC model for any value of $k$ and $d$, where $\Delta$ is the 
ratio of the maximum distance to the minimum distance in the 
input. However, the method
proposed in their paper heavily relies on quad-tree embeddings and seems tailored to the $k$-median objective. This was later 
followed by the work of Cohen-Addad, Mirrokni, and 
Zhong~\cite{Cohen-AddadMZ22} who showed how to 
get a $(1+\eps)$-approximation for so-called \emph{perturbation resilient} $k$-means inputs, a notion of beyond worst-case inputs, in $o(\log n)$ parallel rounds. Other works that studied massively parallel algorithms for the $k$-means problem include \cite{pmlr-v80-bhaskara18a,Cohen-AddadEMNZ22}.

A stronger result was later obtained by Czumaj, Gao, Jiang, Krauthgamer, and   Vesel{\'{y}}~\cite{CzumajGJK024} who obtained
an $O(1)$-approximation in $O(1)$ parallel rounds for the related facility location
problem that can then be used to obtain a bicriteria 
$O_{\eps}(1)$-approximate solution for $k$-means using $(1+\eps)k$ centers also in a constant number of MPC rounds. Recently, there has also been some progress for the $k$-center problem. In \cite{CzumajG0J25}, the authors give fully scalable $O(1)$-round MPC algorithms to obtain a $(2+\eps)$-approximation and a bicriteria $(1+\eps)$-approximation.

In the present paper, we make significant progress toward answering the question about the complexity of the general $k$-means problem in the MPC model. We show the following result. 

\begin{mdframed}[hidealllines=true, backgroundcolor=gray!15]
\textbf{\boldmath Main Theorem. For any constants $\sigma,\eps>0$ and any $k,d\geq 1$, an $O(1)$-approximate solution to the $k$-means problem in $\mathbb{R}^d$ can be computed in the MPC model in $O(\log\log n\cdot\log\log\log n)$ rounds with $O(n^\sigma)$ bits of memory per machine and with $O(n^{1+\eps})$ bits of global memory.}\
\end{mdframed}

Note that in all of the above constant factor approximations and also in the present work, the total global memory is $n^{1+\eps}$ for some fixed constant $\eps>0$, i.e., the total global memory is required to be superlinear by a small polynomial factor. 

Our main technical contribution is a fully scalable algorithm for the facility location problem, which is closely related to the k‑means problem. Our facility location algorithm is based on the classic primal-dual framework of Jain and Vazirani~\cite{JV01}. Since our facility location satisfies a property known as \emph{Lagrangian Multiplier Preserving} (LMP), we can (almost) directly use the framework from \cite{JV01} to apply our facility location algorithm to solve the $k$-means problem and prove our main theorem.

\subsection{Further Related Work}
We briefly review additional related literature on two techniques that turn out to be important for our work: locality-sensitive hashing (LSH) and ruling set algorithms in parallel and distributed models.

\paragraph{Locality-Sensitive Hashing.}
Locality-Sensitive Hashing (LSH) is a key technique for efficiently handling approximate nearest neighbor queries in high-dimensional spaces, which is central to many clustering algorithms in parallel and distributed settings. The foundational work by Andoni and Indyk~\cite{andoni2006near} introduced near-optimal LSH-based algorithms for high-dimensional nearest neighbor search, which have since been widely adopted in scalable clustering approaches. Additionally, Indyk~\cite{Indyk04} provided a broader survey of data structures for nearest neighbor search in general metric spaces. The construction of~\cite{andoni2006near} was extended to the MPC model in~\cite{Cohen-AddadEMNZ22}. In general, LSH and closely related techniques have been used in several recent works on clustering algorithms in space-restricted environments (e.g.,~\cite{pmlr-v80-bhaskara18a,Cohen-AddadEMNZ22,CzumajGJK024,Czumaj2022}).

\paragraph{Distributed and Parallel Ruling Set Algorithms.}  
An $(a,b)$-ruling set of a graph (as introduced in~\cite{awerbuch89}) is a set of nodes such that any two nodes in the set are at distance at least $a$, and every node not in the set is within distance $b$ of some node in the set. Such sets are useful primitives in distributed computing, particularly in the context of symmetry breaking and locality-based computations. Ruling sets have therefore been studied extensively in the distributed setting (e.g.,~\cite{ghaffari2016improved,pai2017symmetry,kothapalli2012super,KuhnMW2018,BalliuGKO23}).

To date, it remains unknown whether there exists a constant- or $\poly\log\log n$-round algorithm for computing an $(O(1), O(1))$-ruling set in the distributed or sublinear-memory MPC setting. The fastest known algorithm in the low-memory MPC model is the $O(\log\log n \cdot \log\log\log n)$-round algorithm of~\cite{KothapalliPP20}, which computes a $(2, O(\log\log\log n))$-ruling set. A sublogarithmic-round algorithm for computing a maximal independent set (which is a $(2,1)$-ruling set) was given by Ghaffari and Uitto~\cite{GhaffariU19}. Faster algorithms are known in the setting where each machine has memory linear in the number of graph nodes. In this case, it has been shown that a $(2,2)$-ruling set can be computed in $O(1)$ MPC rounds. A randomized algorithm for this was presented in~\cite{CambusKPU23}, followed by a deterministic algorithm by Giliberti and Parsaeian~\cite{GP24}.

\subsection{Organization of the Paper}
The remainder of the paper is organized as follows. In \Cref{sec:prelim}, we introduce the necessary mathematical notation and we formally define the $k$-means problem. We also introduce the closely related facility location problem. Subsequently, in \Cref{sec:techoverview}, we provide an informal overview over the most important challenges that we face and the main ideas that we use to overcome those challenges. In \Cref{sec:highlevelalg}, we provide and analyze a high-level version of our $k$-means algorithm, which is tailored towards parallel implementation, but still independent of the concrete computational model we use. The main part of \Cref{sec:highlevelalg} is a constant-factor LMP-approximation of the facility location problem. In \Cref{sec:mpc}, we then show how the algorithm of \Cref{sec:highlevelalg} can be implemented efficiently in the sublinear memory MPC model. Finally, in \Cref{sec:mainproof}, we combine everything to prove our main theorem.

\section{Problem Definition and Preliminaries}\label{sec:prelim}

\subsection{Mathematical Definitions and Notation}
When arguing about asymptotic complexities, we sometimes make use of the $\tilde{O}(\cdot)$-notation. We use $\tilde{O}(x)$ to denote terms that are upper bounded by $x\cdot\poly\log x$, i.e., $\tilde{O}(\cdot)$ hides factors that are polylogarithmic in the argument.

Our problems are defined for point sets in Euclidean space. For $x,y\in \R^d$, we use $\dist(x,y):=\|x-y\|_2$ to denote the Euclidean distance between $x$ and $y$. For $x\in \R$ and $Y\subset \R$, we define the distance between $x$ and $Y$ as $\dist(x,Y):=\inf_{y\in Y}\dist(x,y)$. Further, for $x\in \R^d$, $Y\subset \R^d$, and a radius $r\geq 0$, we use $B_Y(x,r):=\set{y\in Y\,:\,\dist(x,y)\leq r}$ to denote the ball of radius $r$ around $x$ restricted to the points in $Y$. 

For graphs, we need the concept of ruling sets~\cite{awerbuch89}. Given an undirected graph $G=(V,E)$ and two integers $\alpha \geq 1$ and $\beta \geq \alpha - 1$, an $(\alpha,\beta)$-ruling set $S \subseteq V$ of $G$ is a set of nodes such that for any two nodes $u, v \in S$, $d_G(u,v) \geq \alpha$, and for any $u \notin S$, there exists $v \in S$ such that $d_G(u,v) \leq \beta$. The special case where $\alpha = 2$ asks for an independent set $S$, where every node of the graph has an independent set node within hop distance at most $\beta$. Furthermore, a set $S$ is an $(\alpha, \alpha \cdot \beta)$-ruling set if and only if $S$ is a $(2,\beta)$-ruling set of the graph $G^{\alpha - 1}$. Here, for any $t \geq 1$, $G^t$ is defined as the graph on node set $V$ with an edge between any two nodes at distance at most $t$ in $G$.

Finally, we use the following notation throughout the paper:
\[
\forall x\in \R\,:\,[x]^+ := \max\set{0, x}.
\]

\subsection{The {\it k}-Means Problem}

Let $d \geq 1$ be an integer. In the $k$-means problem in $\mathbb{R}^d$, the input consists of a set of $n$ points $P \subset \mathbb{R}^d$ and a positive integer $k \leq n$. The goal is to identify a set of $k$ centers $Z \subset \mathbb{R}^d$ that minimizes the total squared Euclidean distance from each point in $P$ to its closest center in $Z$. Formally, the objective is to compute a set $Z \subseteq P$ with $|Z| = k$ such that the following cost function is minimized:
\begin{align*}
    \cost(Z) := \sum_{p \in P} \dist^2(p, Z).
\end{align*}
Note that we do not restrict the dimension $d$ as a function of $n$. However, at the cost of a $(1\pm\eps)$-factor in the approximation ratio (for an arbitrarily small constant $\eps > 0$), we can apply a data-oblivious random projection to reduce the dimension of the input point set $P$ to $\R^{O(\log k)}$ while approximately preserving the $k$-means cost. Therefore, we can assume without loss of generality that $d = O(\log k)=O(\log n)$~\cite{BBCGS19}. Further, by scaling and sacrificing another $(1+\eps)$-factor in the approximation ratio, we can assume that the minimum distance between any two points is at least $1$, and the maximum distance is bounded polynomially in $n$ (Lemma 4.1 in~\cite{ANSW20}). Throughout the paper, we therefore assume that
\begin{equation}\label{eq:pointset}
  P\subset \R^d\text{ for }d=O(\log k)\quad\text{and}\quad
  \forall x,y\in P\,:\,x\neq y \Rightarrow \dist(x,y)\in\big[1,n^{O(1)}\big].
\end{equation}

A standard approach to approximate the $k$-means problem is by reducing it to a variant of the facility location problem. As we also follow this approach, we next formally define the facility location problem and we sketch the classic approximation algorithm by Jain and Vazirani~\cite{JV01}.

\subsection{The Facility Location Problem}

The facility location problem is defined as follows: Given a set of facilities $\F$ and a set of clients $\C$, the objective is to select a subset $\F' \subseteq \F$ of facilities to open and to assign each client $c \in \C$ to an open facility $f \in \F'$ such that the total cost is minimized. In the variant of the problem we study, opening a facility $f \in \F$ incurs a fixed opening cost $\lambda > 0$, and connecting a client $c \in \C$ to a facility $f$ incurs a connection cost $\cost(c,f) \geq 0$. For convenience, we use $\cost(c,f)$ and $\cost(f,c)$ interchangeably (with $\cost(c,f) = \cost(f,c)$). In the context of this paper, we assume that the connection cost satisfies a relaxed version of the triangle inequality. More precisely, for any integer $\ell \geq 2$ and any sequence $x_0,x_1,\dots,x_{\ell}$ of clients and facilities that alternates between clients and facilities (i.e., either $x_0,x_2,\dots$ are clients and $x_1,x_3,\dots$ are facilities or vice versa), we have
\begin{equation}\label{eq:triangleineq}
    \sum_{i=1}^{\ell} \cost(x_{i-1}, x_i) \leq \ell\cdot \cost(x_0,x_s).
\end{equation}
We note that this relaxation of the triangle inequality holds if the cost $\cost(c,f)$ refers to the squared distance between $c$ and $f$ in some underlying metric space.

\paragraph{Linear Programming Formulation.} The facility location problem can be formulated as an integer linear program as follows:

\begin{align} \label{eq:fl-primal}
\begin{split}
\min & \quad \sum_{f \in \F} \lambda \cdot y_f + \sum_{f \in \F} \sum_{c \in \C} \cost(c,f) \cdot x_{cf} \\ 
\text{s.t.} & \quad \sum_{f \in \F} x_{cf} \ge 1, \quad \forall c \in \C, \\
&\quad x_{cf} \leq y_f, \quad \forall f \in \F, \forall c \in \C, \\
&\quad x_{cf} \in \{0,1\}, \quad \forall f \in \F, \forall c \in \C, \\
& \quad y_f \in \{0,1\}, \quad \forall f \in \F.
\end{split}
\end{align}
Here, $y_f$ indicates whether facility $f$ is open, and $x_{cf}$ represents whether client $c$ is connected to facility $f$. The first constraint ensures that each client is connected to at least one facility and the second constraint ensures that this facility has to be open. Relaxing $y_f$ and $x_{cf}$ to be in $[0,1]$ yields the LP relaxation.

The \emph{dual problem} of the LP relaxation of \eqref{eq:fl-primal} can be stated as follows.
\begin{equation} \label{eq:fl-dual}
\begin{split}
\max & \quad \sum_{c \in \C} \alpha_c \\ 
\text{s.t.} &\quad \alpha_c - \beta_{cf} \leq \cost(c,f), \quad \forall f \in \F, \forall c \in \C, \\ 
&  \quad \sum_{c \in \C} \beta_{cf} \leq \lambda, \quad \forall f \in \F, \\
&  \quad \alpha_c \geq 0, \quad \forall c \in \C, \\
& \quad \beta_{cf} \geq 0, \quad \forall f \in \F, \forall c \in \C.
\end{split}
\end{equation}
In this dual LP, $\alpha_c$ represents the total price paid by a client $c$, and $\beta_{cf}$ represents the contribution of client $c$ towards opening facility $f$. If a client $c$ is connected to a facility $f$, then $\alpha_c - \beta_{cf} = \cost(c,f)$; otherwise, it is zero. A facility is opened when $\sum_{c \in \C}\beta_{cf} = \lambda$.

\paragraph{Facility Location Algorithm by Jain and Vazirani}
In \cite{JV01}, Jain and Vazirani presented a centralized primal-dual algorithm to compute a $3$-approximation for facility location problem if the cost function satisfies the (strict) triangle inequality. The algorithm achieves a $9$-approximation in our case with the relaxed triangle inequality \eqref{eq:triangleineq}. In the following, we briefly sketch their algorithm.

The input is modeled as a (complete) bipartite graph, where one set of vertices represents the clients $\C$ and the other set represents the facilities $\F$. The algorithm initializes all dual variables $\alpha_c$ and $\beta_{cf}$ to $0$, and initially classifies all clients as unconnected.

The first phase of the algorithm initially defines all clients as being \emph{active} and it proceeds by uniformly increasing the $\alpha_c$ values of all active clients $c$ at a fixed rate. Whenever $\alpha_c = \cost(c,f)$ for some edge $\{c, f\}$ (i.e., for some client $c$ and some facility $f$), the edge $\{c, f\}$ is declared as \emph{tight}. At this point, client $c$ has fully paid its connection cost to facility $f$ and begins contributing toward the facility's opening cost. As a result, from now on, the dual variable $\beta_{cf}$ is increased at the same rate as $\alpha_c$.

A facility $f$ is considered \emph{temporarily open} once the total contribution from its adjacent clients is equal to the opening cost $\lambda$, i.e., if $\sum_{c \in \C} \beta_{cf} = \lambda$. Moreover, all unconnected clients that have a tight edge to $f$ are then immediately connected to $f$ and become \emph{inactive} (ties about where to connect a client are broken arbitrarily in case several facilities become temporarily open at the same time). Throughout the remainder of the first phase of the algorithm, whenever a client $c$ obtains tight edge with a facility $f$ that is already temporarily open, $c$ is connected to $f$ and $c$ becomes inactive (note that $\beta_{cf}$ in this case remains $0$). The first phase terminates when all clients are connected to a temporarily open facility, i.e., when all clients are inactive. Note that the above construction ensures that the variables $\alpha_c$ and $\beta_{cf}$ form a feasible solution of the dual LP \eqref{eq:fl-dual}.

At the end of the first phase, a client may have contributed to opening multiple facilities. To ensure each client contributes to at most one facility, we define a conflict graph $H=(\F_T, E_H)$, where $\F_T$ is the set of temporarily open facilities. There is an edge $\set{f,f'}\in E_H$ between two temporarily open facilities $f$ and $f'$ if and only if there exists a client $c$ such that $\beta_{cf}>0$ and $\beta_{cf'}>0$, i.e., $c$ contributes to opening both $f$ and $f'$. The final set of open facilities $\F'$ is chosen as a maximal independent set of $H$. In this way, each client contributes to at most one open facility. By using the relaxed triangle inequality \eqref{eq:triangleineq}, one can show that
\begin{equation}\label{eq:JVapprox}
  \sum_{c\in \C} \cost(c, \F') \leq 
  9\cdot \left(\sum_{c\in \C} \alpha_c -|\F'|\cdot\lambda\right) \leq
  9\cdot \left(\mathsf{OPT} -|\F'|\cdot\lambda\right),
\end{equation}
where $\mathsf{OPT}$ denotes the objective value of an optimal solution to the given facility location problem. The above inequality implies that the objective value $\sum_{c\in \C}\cost(c,\F')+|\F'|\cdot\lambda$ achieved by the algorithm is within a factor at most $9$ of $\mathsf{OPT}$. Moreover the solution has the additional property known as \emph{Langrangian multiplier preserving (LMP)}~\cite{JV01,BreachingLMP}. It was already shown by Jain and Vazirani~\cite{JV01} that a constant-factor LMP approximation algorithm for the facility location problem can be used to develop a constant approximation for the $k$-means problem.

\subsection{The Massively Parallel Computation Model}

The \emph{massively parallel computation} (MPC) model was introduced in \cite{KarloffSV10}. It is an abstract model that captures essential aspects of coarse-grained parallelism in large-scale data processing systems. In an MPC algorithm for an input of size $n$, we have a number of machines, each with $n^\sigma$ bits of local memory for some constant $\sigma>0$. An algorithm is called \emph{fully scalable} if the constant $\sigma$ can be made arbitrarily small. The machines can communicate with each other in synchronous rounds. In each round, every machine may send and receive up to $O(S)$ bits of data and perform arbitrary local computation. The total number of machines and thus the global memory can also be bounded. Ideally, one usually aims for a global memory of $\tilde{O}(n)$ bits. Sometimes, this is not possible. Another popular assumption that we also make in our paper is that the global memory is restricted to $O(n^{1+\eps})$ bits, where $\eps>0$ is a constant that can be chosen arbitrarily small, possibly at some cost regarding the guarantees of an algorithm.

\section{Technical Overview}\label{sec:techoverview}

In the following, we discuss the main technical challenges in devising an efficient MPC algorithm for the $k$-means problem and sketch the approach we used to resolve these challenges. As discussed, we assume that the input to the $k$-means problem is given by $n$ points $P \subset \R^d$ for $d = O(\log k)$, and that for any $x, y \in P$ with $x \neq y$, their Euclidean distance satisfies $1 \leq \dist(x,y) \leq \poly(n)$ (cf.\ Eq.~\eqref{eq:pointset}). We solve the given $k$-means instance by reducing it to several instances of the facility location problem with some fixed opening cost $\lambda \geq 1$ per facility. In this facility location instance, the facilities $\F$ and clients $\C$ are also points in $\R^d$ (in fact, in our reduction, we will have $\F = P$ and $\C = P$). For any $x, y \in \F \cup \C$, we therefore also have a well-defined Euclidean distance $\dist(x,y)$. It is convenient to define the cost function not only between a client and a facility, but between any pair of clients and facilities. The cost between any two $x,y \in \C\cup\F$ is defined as
\begin{equation}\label{eq:cost}
  \cost(x,y) := (\dist(x,y))^2.
\end{equation}
Note that by applying the triangle inequality for the Euclidean distances and the Cauchy-Schwarz inequality, this implies that the cost function satisfies the relaxed triangle inequality given by Equation \eqref{eq:triangleineq}.

\paragraph{Challenge 1: Sparse Representation of Pairwise Distances.} Since we assume that the distances are in $[1,n^{O(1)}]$, we can (approximately) represent each point in $P$ by $d=O(\log k)=O(\log n)$ coordinates of $O(\log n)$ bits each. The whole input can therefore be stored using $O(n\log^2 n)=\tilde{O}(n)$ bits. However, we need to be able to make efficient non-trivial queries. In particular, we need to be able to find (approximate) nearest neighbors and (approximately) aggregate over different neighborhoods. Note that we cannot simply store all the distances as a weighted graph, as this would require $\Omega(n^2)$ memory. To approximately store the distances in a structured way, we use \emph{locality-sensitive hashing (LSH)}, a technique that has been used in different contexts, in particular to perform approximate nearest neighbor queries~\cite{CzumajGJK024,andoni2006near,DIIS04}. The technique has also been used in a recent MPC algorithm that achieves a constant bicriteria $k$-means approximation in $O(1)$ time~\cite{CzumajGJK024}. In this paper, we make use of a construction from \cite{andoni2006near}, which allows storing the $n$ points as a weighted graph $G$ with the following properties. The nodes of the graph are the $n$ points. For an arbitrary constant $\eps>0$, the graph has at most $n^{1+\eps}$ edges and 
\[
\forall x,y\in P\,:\ d_G^{(2)}(x,y) \leq \dist(x,y)\leq O\big(\sqrt{1/\eps}\big)\cdot d_G^{(2)}(x,y),
\]
where $d_G^{(2)}(x,y)$ is the length of a shortest path in $G$ with hop-length at most $2$ between $x$ and $y$. The graph can be computed in $O(1)$ time in the MPC model, and it can then be used to perform efficient approximate aggregation over local neighborhoods. The details appear in \Cref{sec:mpc_graphapprox,sec:mpc_graphalg}.

\paragraph{Challenge 2: Parallel LMP Approximation of Facility Location.} As our $k$-means algorithm is based on the Jain/Vazirani framework, the main step of the algorithm is to compute a constant-factor LMP-approximation of the facility location problem. Our main challenge therefore is to design a parallel variant of the primal-dual algorithm that is both efficient and satisfies the LMP property. We start by describing a natural (but too slow) parallel algorithm for an idealized setting, where we can make exact queries about the underlying distances. 

In the primal-dual algorithm, we initially set the dual variables $\alpha_c$ for clients $c$ to $0$ and we gradually increase the $\alpha_c$-values at a fixed rate until some facility becomes paid. We can initially set $\alpha_c=\alpha_0>0$ for some sufficiently small $\alpha_0$. In order to be fast in a parallel setting, one can then try to be more aggressive than in the sequential algorithm and increase $\alpha_c$ in $O(\log n)$ discrete steps, where in each step, the current value is multiplied by some constant (for simplicity, say we double $\alpha_c$ for each active client $c$ in each step). If we freeze each $\alpha_c$ value as soon as $c$ contributes to the opening cost of some facility $f$ that becomes paid or overpaid, we obtain an approximately feasible dual solution. Note that a facility $f$ is called paid w.r.t.\ to dual values $\alpha_c$ if $\sum_{c\in \C} [\alpha_c - \cost(c,f)]^+\geq \lambda$ and it is called overpaid if this sum is strictly larger than $\lambda$.

The approximately feasible dual solution produced in this way can indeed be used to obtain a constant-factor approximation for the facility location problem. Unfortunately, however, it does not yield a solution that satisfies the LMP property. Slightly simplified, we can guarantee the LMP property if we ensure that the final dual solution is feasible and that the facilities we open are exactly paid for—but not overpaid—by the adjacent clients. This, however, suggests that increasing the $\alpha_c$ values in parallel in discrete steps will not work.

We can make the approach work by using the following observation about the sequential primal-dual algorithm. At the cost of losing a constant factor in the approximation quality, one can increase the $\alpha_c$-variables of different active clients at different rates. One merely has to make sure that at all times, the $\alpha_c$-variables of all the active clients are within a constant factor of each other. Consider some intermediate state in the algorithm, where the current solution is still dual feasible and for all active clients $c$, we have $\alpha_c=\alpha$ for some value $\alpha>0$. We would now like to increase the dual variables of the active clients to $2\alpha$. If we did this, some subset $S$ of the facilities would become paid or overpaid. Instead of directly increasing all the dual values of active clients to $2\alpha$, we proceed as follows. We say that two facilities $f, f'\in S$ are in conflict if there exists an active client $c$ for which $2\alpha\geq \cost(c,f)$ and $2\alpha\geq \cost(c,f')$. This relation defines a conflict graph $H$ on the potentially paid facilities $S$. We now select a subset $S_0\subseteq S$ in such a way that any two $f,f'\in S_0$ are at distance at least $4$ in $H$ and for any $f'\in S\setminus S_0$, there (ideally) exists an $f\in S_0$ at distance $O(1)$ from $f'$ in $H$. Note that such a set $S_0$ is known as a $(4,O(1))$-ruling set of the $H$ (cf.~\Cref{sec:prelim}). We then build clusters of facilities and clients around each facility in $\S_0$ and we will eventually open exactly one facility in each cluster. For each $f\in S_0$, let $C_f$ be the set of active clients for which $2\alpha> \cost(c,f)$. By gradually increasing only the $\alpha_c$-values for $c\in C_f$, we make sure that some facility $\bar{f}$ (either $f$ or one of its neighboring facilities in the conflict graph) becomes paid before the $\alpha_c$-values for $c\in C_f$ exceed $2\alpha$. We can open $\bar{f}$ and assign active clients $c$ for which $\cost(c,\bar{f})=O(\alpha)$ to $\bar{f}$. Because any two facilities in $S_0$ have distance at least $4$ in the conflict graph, the client sets $C_f$ and $C_{f'}$ for $f,f'\in S_0$ are disjoint and they also contribute to a disjoint set of facilities in $S$. The increase of the dual values in $C_f$ for different $f\in S_0$ can therefore be done completely independently and therefore in a `local' manner. Moreover, since for each $f\in S_0$, exactly one facility will be opened, one can show that even for proving the LMP property, it is sufficient to just open a facility $\bar{f}$ that approximately minimizes $\sum_{c\in C_f}\cost(c,\bar{f})$.

The sketched parallel facility location algorithm has two main problems. First, because our sparse representation of the points (and thus of facilities and clients) only allows us to make approximate queries about neighborhoods, we cannot exactly implement the algorithm in our setting. By adding sufficient slack in all steps, it is, however, possible to adapt the algorithm to the setting where all steps can only be carried out in an approximate manner. More importantly, even though we increase the dual values quite aggressively, the algorithm is still too slow. In the end, the ratio between the largest and the smallest dual value $\alpha_c$ can be polynomial in $n$. We therefore need $\Omega(\log n)$ doubling steps, which, of course, implies that the round complexity of the algorithm will also be at least $\Omega(\log n)$. However, we aim for an algorithm with a round complexity that is at most $\polyloglog n$ and thus almost exponentially smaller than what we can achieve with the above algorithm. We next discuss how we can use the same basic idea without going through the logarithmically many doubling steps of the dual variables.

\paragraph{Challenge 3: Fast Parallel Primal-Dual Algorithm.}
As even exponentially increasing the dual client values is far too slow, we need to find a way to fix the dual variables much more directly. Also note that a reasonable approximate solution to the dual LP might require different clients $c$ to use up to $\Theta(\log n)$ distinct values for $\alpha_c$. We therefore also cannot increase the $\alpha_c$ variables together in a synchronized way. Instead, we aim to define a value $\alpha_c^*$ for every client, where $\alpha_c^*$ is essentially chosen as the maximum value such that, if all clients set their dual variable to $\alpha_c^*$, no facility $f$ to which $c$ contributes (i.e., for which $\cost(c,f) > \alpha_c^*$) is overpaid. Note that $\alpha_c^*$ is a lower bound on the final value of $\alpha_c$ in the standard sequential primal-dual algorithm. Concretely, the values $\alpha_c^*$ are computed as follows.

First, every facility $f$ computes a radius $r_f$ such that
\begin{equation}
  r_f := \max 
  \set{r \in \R\,:\, \sum_{c \in \C} 
  \big[r^2 - \cost(c,f)\big]^+ \leq \lambda}.
\end{equation}
Note that $r_f$ is chosen such that, if all clients $c$ set $\alpha_c = r_f^2$, then the facility $f$ is exactly paid for. We now set $\alpha_c^*$ such that $\alpha_c^* \leq r_f^2$ for every facility $f$ for which $\alpha_c^* > \cost(c,f)$:
\begin{equation}
  \alpha_c^* := \min_{f \in \F} \max \set{r_f^2, \cost(c,f)}.
\end{equation}

Those definitions directly imply that by setting all the dual variables of clients $c$ to $\alpha_c^*$ (and setting $\beta_{cf} = \alpha_c^* - \cost(c,f)$), we obtain a dual feasible solution. Interestingly, one can also show that by setting $\alpha_c = C_A \cdot \alpha_c^*$ for a sufficiently large constant $C_A \geq 1$, every client $c$ has a facility $f$ within distance $O\big(\sqrt{\alpha_c^*}\big)$ such that $f$ is at least fully paid for by the dual variables $\alpha_c$. One could therefore try to adapt the approach of the above algorithm as follows: consider the set $S$ of facilities that are fully paid for by the dual values $C_A \cdot \alpha_c^*$, define the conflict graph $H$ between facilities in $S$ w.r.t the dual values $C_A \cdot \alpha_c^*$, select a set of centers $S_0 \subseteq S$ as a $(4, O(1))$-ruling set of $H$, and open one facility close to every node in $S_0$. However, this does not work directly, particularly because choosing the wrong centers $S_0$ might result in missing the opening of some facilities $f$ with small radius $r_f$, causing some clients to be connected to a facility that is too far away.

The problem, in particular, occurs if for a client $c$, there exists another client $c'$ within distance $O(\sqrt{\alpha_c^*})$ such that $\alpha_{c'}^* \ll \alpha_c^*$. Instead of increasing the dual variable of $c$ until some facility becomes tight, we should, in this case, simply connect $c$ to the facility that we open for client $c'$. We resolve this in the following way. We define a set $\C' \subset \C$ of \emph{problematic clients}. Essentially, a client is problematic if there exists another client $c'$ within distance $O(\sqrt{\alpha_c^*})$ such that $\alpha_{c'}^* \ll \alpha_c^*$. For every client $c$, we define two dual values $\alpha_{c,0}$ and $\alpha_{c,1}$, where $\alpha_{c,0}$ is smaller than $\alpha_c^*$ by some constant factor. For problematic clients, we set $\alpha_{c,1} = \alpha_{c,0}$, and for all other clients, we set $\alpha_{c,1} = C_A \cdot \alpha_{c,0}$, where $C_A$ is a sufficiently large constant. One can then show that, when choosing the constants appropriately, for all clients $c$, there exists a facility $f'$ within distance $\alpha_{c,0}$ such that $f'$ is paid w.r.t the dual values $\alpha_{c',1}$ (formally proven in \Cref{lem:paidfacility}). It is now also true that whenever two facilities $f, f'$ are in conflict (i.e., if there exists a client $c$ that contributes to both of them), then $r_f = \Theta(r_{f'})$ and $\alpha_{c,1} = O(r_f)$, and thus $\dist(f,f') = O(r_f)$ (cf.\ \Cref{lem:constant-rad,lemma:sameradius}). This is now sufficient for our general approach described above to work.

One of the challenges in implementing the above algorithm in the MPC model is that we only have approximate access to the underlying Euclidean metric through our sparse LSH approximation. Therefore, for example, we cannot compute $r_f$ for facilities $f$ or $\alpha_c^*$ for clients $c$ exactly. However, we can efficiently compute approximations that differ from the actual values by at most a constant factor. We can also only approximately determine whether a facility is paid for by some given dual client values. Consequently, we must compute a conflict graph $H$ that includes all facilities that are approximately paid. Fortunately, the above algorithm is robust enough to work even when all quantities are computed only up to a constant-factor error.


\paragraph{Challenge 4: Computing a Ruling Set in the MPC Model.}
One of the most challenging tasks in implementing the sketched algorithm in the MPC model is the computation of a ruling set to determine cluster centers $\S_0$ among the (approximately) paid facilities. The problem can be broken down to a standard graph algorithms problem in the MPC model. Given an $n$-node and $m$-edge graph $G=(V,E)$, one needs a fully scalable algorithm to compute a $(a,b)$-ruling set for some constants $a>1$ and $b\geq a-1$. In the conflict graph $H$, we needs to compute an $(4,O(1))$-ruling set. However, our algorithm will compute a sprase graph $G$ such that $G^2$ serves as the conflict graph. We therefore have to compute an $(7, O(1))$-ruling set on $G$. Unfortunately, there currently are no fully scalable MPC algorithm to compute such sets in time at most $\poly\log\log n$. The fastest fully scalable MPC algorithm for computing a ruling set is an algorithm that computes a $(2,O(\log\log\log n))$-ruling set in time $O(\log\log n \cdot \log\log\log n)$~\cite{KothapalliPP20}. The algorithm uses $O(m+n^{1+\eps})$ bits of global memory (for $\eps$ an arbitrarily small constant), which is fine for us. However, the algorithm does not directly compute an $(a,b)$-ruling set for any $a>2$ and even more problematically, it only computes an $(a,b)$-ruling set for $b=O(\log\log\log n)$. When using such a ruling set in our algorithm, we would only achieve an approximation ratio that is exponential in $O(\log\log\log n)$ and thus $\poly\log\log n$. We still use the ruling set algorithm of \cite{KothapalliPP20}, but we have to combine it with an additional idea.

Let us first discuss, how one can adapt any $(2,b)$-ruling set algorithm to compute an $(a,(a-1)b)$-ruling set on some given graph $G$. Note that a node set $X$ is $(a,(a-1)b)$-ruling set on $G$ if and only if $X$ is a $(2,b)$-ruling set on $G^{a-1}$. We cannot directly apply an algorithm on $G^{a-1}$ because $G^{a-1}$ might be dense even if $G$ is sparse and we therefore would need too much global memory. We can however again exploit the fact that we have $n^{1+\eps}$ bits of global memory. By sampling each node with a sufficiently small probability, one can make sure that the subgraph of $G^{a-1}$ induced by the sampled nodes has at most $n^{1+\eps}$ edges. When computing a $(2,b)$-ruling set on the sampled subgraph of $G^{a-1}$ and remove all nodes that have a ruling set node within distance $(a-1)b$ in $G$, the maximum degree of the remaining subgraph of $G^{a-1}$ drops by a factor $n^{\Omega(\eps)}$. We can therefore compute an $(a, (a-1)b)$-ruling set of $G$ in $O(1/\eps)=O(1)$ phases, where in each phase, we run a $(2,b)$-ruling set algorithm on a sufficiently sparse graph that we can construct explicitly.

Let us now discuss how we cope with the fact that we do not know how to efficiently compute an $(a,b)$-ruling set for $b=O(1)$. In \cite{BalliuGKO23}, it is shown that if the basic randomized Luby algorithm~\cite{alon86,luby86} for computing a maximal independent set is adapted to compute a $(2,2)$-ruling set, then in each step of the algorithm, a constant fraction of the nodes of a graph becomes covered. In $O(\log\log n)$ iterations of the algorithm, one can therefore guarantee that a $(1-1/\poly\log n)$-fraction of the nodes has a ruling set node within distance $2$. This idea can easily be adapted to the MPC model and also to computing $(a,a)$-ruling sets instead of $(2,2)$-ruling sets. By slightly adjusting the quality of the ruling sets and running the algorithm for several weight classes, one can also get a weighted version of this. Given a node-weighted graph, in $O(\log\log n)$ rounds, one can compute a set $X$ of centers so that any two nodes in $X$ are at distance at least $a$ (for a given $a>1$) and such that the total weight of nodes that are not within distance $O(1)$ of a node in $X$ is at most at $1/\poly\log n$-fraction of the total node weight. If one uses the algorithm of \cite{KothapalliPP20} to cover the remaining nodes within distance $O(\log\log\log n)$, we can use the computed set of centers to find a constant approximate facility location solution.

\paragraph{Challenge 5: From Facility Location to {\it k}-Means.}
Our reduction from the $k$-means problem to facility location follows the framework described by Jain and Vazirani for the $k$-median problem~\cite{JV01}. Given a $k$-median instance with point set $P$, one can define a facility location instance by setting both the facilities $\F$ and the clients $\C$ equal to $P$. If there exists an opening cost $\lambda$ for which our LMP approximation algorithm opens exactly $k$ facilities, we can directly use this solution for the $k$-means instance. Of course, this may not always be possible, so we aim to approximate this scenario. As a first step, we compute facility location solutions for two opening costs, $\lambda_1$ and $\lambda_2$, such that $\lambda_2 \leq \lambda_1 \leq 2\lambda_2$, where the solution corresponding to $\lambda_1$ opens $k_1 \leq k$ facilities and the one corresponding to $\lambda_2$ opens $k_2 \geq k$ facilities.

Jain and Vazirani describe a randomized rounding procedure that interpolates between the two facility location solutions corresponding to $\lambda_1$ and $\lambda_2$, such that the resulting solution opens exactly $k$ facilities. Provided that the facility location algorithm satisfies the LMP property, the resulting $k$-means solution achieves a constant-factor approximation with respect to the original instance.  Using the LSH-based approximation of Euclidean distances, this randomized rounding procedure can be implemented in the MPC model in a relatively straightforward manner.

\section{High-Level {\it k}-Means Algorithm}\label{sec:highlevelalg}

We are now ready to formally define and analyze our algorithm for computing a constant-factor approximation to the $k$-means problem. Instead of directly presenting an MPC algorithm, we first describe a high-level version that is (mostly) independent of implementation details. In \Cref{sec:mpc}, we then show that the high-level algorithm presented in this section can be implemented efficiently and in a fully scalable way within the MPC model. The high-level algorithm consists of two parts. In \Cref{sec:fl-highlevel}, we describe the main component: a constant-factor LMP approximation algorithm for the facility location problem. In \Cref{sec:kmeans-highlevel-alg,sec:kmeans-highlevel-analysis}, we show how to leverage this facility location algorithm to compute an approximate solution for a given $k$-means instance.

\subsection{High-Level Facility Location Algorithm Description}\label{sec:fl-highlevel}

As discussed at the beginning of \Cref{sec:techoverview}, we assume that we are given a set $\F$ of facilities, a set $\C$ of clients, and an opening cost $\lambda \geq 1$. We further assume that $\F$ and $\C$ are represented by points in $\R^d$ for $d = O(\log k)$, and that for any two distinct points $x, y \in \F \cup \C$, we have $\dist(x, y) \geq 1$. The connection cost between a client $c$ and a facility $f$ is defined as $\cost(c, f) = \dist^2(c, f)$. The details of our high-level facility location algorithm are provided in the pseudo-code of \Cref{alg:fl-overview}. In the following, we also discuss each step of the algorithm in further detail.

\begin{algorithm}[p!]
\caption{High-Level Facility Location Algorithm} \label{alg:fl-overview}
\begin{algorithmic}[1]
    \vspace*{1mm}
    \State \textbf{Step I: Assign Facility Radii.} Each facility $f\in \F$ computes a radius $\hat{r}_f$ such that 
    \begin{equation}\label{eq:facilityradius}
    \frac{r_f}{C_R} \leq \hat{r}_f \leq r_f \quad\text{ for } r_f = \max\set{r\in \R\,:\,\sum_{c\in \C} \big[r^2 - \cost(c,f)]^+ \leq \lambda} \text{ and } C_R >1.
    \end{equation}
    \State \textbf{Step II: Assign Client Dual Values.} Each client $c\in \C$ computes a value $\alpha_{c,0}$ such that 
    \begin{equation}\label{eq:initialdualvalues}
    \frac{\alpha_{c}^*}{C_D^+} \leq \alpha_{c,0} \leq \frac{\alpha_{c}^*}{C_D^-} \quad\text{ for } \alpha_{c}^*=\min_{f\in \F}\max\set{r_f^2, \cost(c,f)},\ C_D^+ > C_D^- \ge 2.
    \end{equation}
    \State \textbf{Step III: Determine Problematic Clients.} Compute $\C'\subseteq \C$ such that for all $c \in \C$ and appropriate constants $\gamma_1>1$, $\gamma_2>1$, and $Q>1$,
    \begin{equation}\label{eq:problematicclients1}
    \exists c'\in B_{\C}\left(c,\gamma_1\cdot \sqrt{\alpha_c^*}\right) \text{ such that } \alpha_{c',0} \leq \frac{\alpha_{c,0}}{Q}
    \quad \Longrightarrow \quad 
    c\in \C',
    \end{equation}
    \begin{equation}\label{eq:problematicclients2}
    c\in \C' 
    \quad \Longrightarrow \quad
    \exists c'\in B_{\C}\left(c,\gamma_2\cdot \sqrt{\alpha_c^*}\right) \text{ such that } \alpha_{c',0} \leq \frac{\alpha_{c,0}}{Q}. 
    \end{equation}
    \State Define the upper dual values as:
    \begin{equation}\label{eq:updateddualvalues}
    \forall c\in\C, \quad \alpha_{c,1} := 
    \begin{cases}
        \alpha_{c,0}, & \text{if } c\in\C',\\
        C_A\cdot\alpha_{c,0}, & \text{with $C_A > 1$ otherwise}.
    \end{cases}
    \end{equation}
    
    \State \textbf{Step IV: Approximately Paid Facilities.} Compute a set $\mathcal{S}\subseteq \F$ of facilities such that for the dual values $\alpha_{c,1}$ and $\kappa > 1$, a) $\mathcal{S}$ contains all paid facilities of $\F$ and b) all facilities in $\mathcal{S}$ are $\kappa$-approximately paid.
    \vspace*{3mm}
    \State \textbf{Step V: Dependency Graph.} The dependency graph $H=(\S, E_H)$ on $\S$ is defined as 
    \[
    E_H := \set{\set{f,f'}\in \binom{\S}{2}\,:\, 
    \exists c\in \C\ \text{s.t.}\ \kappa\cdot \alpha_{c,1} > \max\set{\cost(c,f), \cost(c,f')}}.
    \] 
    \State Compute an undirected graph $H'=(\S, E_{H'})$ such that $E_H\subseteq E_{H'}$ and for every edge $\set{f,f'}\in E_{H'}$ we have $r_{f}/C_{H,1}\leq r_{f'}\leq C_{H,1}\cdot r_f$, $\dist(f,f')\leq C_{H,2}\cdot\min\set{r_f, r_{f'}}$, and appropriate constants $C_{H,1}$ and $C_{H,2}$.
    \vspace*{3mm}
    \State \textbf{Step VI: Compute Clustering.} Determine a set of cluster centers $\mathcal{S}_0\subseteq \mathcal{S}$ such that
    \begin{itemize}
        \item Any two facilities in $\mathcal{S}_0$ are at distance at least $4$ in $H$.
        \item For every client $c\in \C$, determine one facility $f_c\in \mathcal{S}$ for which $\max\{r_f^2, \cost(c,f)\}=O(\alpha_{c,0})$.
        \item Let $\C^+\subseteq \C$ be the clients for which $f_c$ is within hop dist.\ $O(1)$ of a node $f_0\in \mathcal{S}_0$ in $H'$.
        \item For all clients $c\in \C\setminus\C^+$, $f_c$ is within distance $o(\log\log n)$ of a node in $\mathcal{S}_0$ in $H'$.
        \item Define $A:=\sum_{c\in \C} \alpha_{c,0}$ and $A^+:=\sum_{c\in \C^+} \alpha_{c,0}$. We have $(A-A^+)/A\leq 1/\log n$.
    \end{itemize}
    \vspace*{1mm}
    \State Assign each facility to a nearest center in $H'$ and each client $c\in \C$ to the cluster of $f_c$.
    \vspace*{3mm}
    \State \textbf{Step VII: Opening Facilities.} In every cluster, open one facility that approximately (within a constant factor) minimizes the distance to all clients of the cluster.
    \vspace*{1mm}
\end{algorithmic}
\end{algorithm}

The algorithm relies on several constants that determine the accuracy with which each individual step can be efficiently executed. We specify all of these constants as a function of a sufficiently large constant $\Gamma \geq 1$, which, in turn, depends on the quality of our sparse LSH graph in approximating the distances in $\R^d$. We use the following definition to characterize when a facility is paid for or approximately paid for.

\begin{definition}[Paid and Approximately Paid Facilities]\label{def:paidfacilities} 
Consider a facility location instance with facilities $\F$ and clients $\C$, and assume that each client $c \in \C$ is assigned a dual value $\alpha_c > 0$. We say that:
\begin{itemize} 
    \item A facility $f$ is \emph{paid} w.r.t the dual values $\alpha_c$ if
    \[ 
        \sum_{c \in \C} \big[\alpha_c - \cost(c,f)\big]^+ \geq \lambda. 
    \] 
    \item A facility $f$ is \emph{$\kappa$-approximately paid} w.r.t the dual values $\alpha_c$ and a constant $\kappa > 1$ if it is paid w.r.t the scaled dual values $\kappa \cdot \alpha_c$.
\end{itemize} 
\end{definition}

We further use the following terminology to describe when a client contributes to paying for a facility. We say that a client $c$ \emph{contributes to a facility} $f$ w.r.t the dual value $\alpha_c$ if $\alpha_c > \cost(c,f)$. For $\kappa \geq 1$, we say that a client $c$ \emph{$\kappa$-approximately contributes to a facility} $f$ w.r.t the dual value $\alpha_c$ if $\kappa \cdot \alpha_c > \cost(c,f)$.


\paragraph{Step I: Facility Radius Assignment.}
As described in \Cref{sec:techoverview}, we define a radius $r_f$ for each facility $f$ such that $f$ is exactly paid for if all clients set $\alpha_c = r_f^2$. Since in our MPC implementation the radii $r_f$ cannot be computed exactly, the algorithm instead computes an approximate value $\hat{r}_f$ for each facility $f$. The approximate radius $\hat{r}_f$ is chosen to satisfy
\begin{align*}
    \frac{r_f}{C_R} \leq \hat{r}_f \leq r_f,
\end{align*}
for some constant $C_R > 1$.


\paragraph{Step II: Assign Client Dual Values.}
For each client $c \in \C$, we initialize an approximate dual value $\alpha_{c,0}$ that approximates the minimal value $\alpha_c^* := \min_{f \in \F} \max\{r_f^2, \cost(c,f)\}$ required to connect $c$ to some facility in $\F$ (cf.\ the description in \Cref{sec:techoverview}). The assigned dual value $\alpha_{c,0}$ is required to satisfy
\begin{align*}
\frac{\alpha_c^*}{C_D^+} \leq \alpha_{c,0} \leq \frac{\alpha_c^*}{C_D^-},
\end{align*}
where $C_D^+$ and $C_D^-$ are constants such that $C_D^+ > C_D^- > 1$. This range ensures that $\alpha_{c,0}$ remains within a constant factor of $\alpha_c^*$. We will later show that initializing all dual values $\alpha_{c,0}$ in this way guarantees that no facility is fully paid, i.e.,
\begin{align*}
\forall f \in \F,\quad \sum_{c \in \C} \big[\alpha_{c,0} - \cost(c,f)\big]^+ < \lambda. 
\end{align*}


\paragraph{Step III: Identification of Problematic Clients} 
In this step, we identify clients in $\C$ that may become \emph{problematic} if their dual values are increased beyond their initial assignment. We define a subset $\C' \subseteq \C$ consisting of clients whose scaled dual values could lead to inconsistencies in the facility-opening process.

Informally, a client $c \in \C$ is classified as problematic if, upon increasing its dual value by a sufficiently large constant factor, it contributes to two facilities $f$ and $f'$ such that the radius of $f'$ is significantly smaller than the radius of $f$. Formally, a client $c$ must be classified as problematic if
\[
\exists c' \in B_{\C}\left(c, \gamma_1 \cdot \sqrt{\alpha_c^*}\right)\ 
\text{such that} \quad \alpha_{c',0} \leq \frac{\alpha_{c,0}}{Q},
\]
where $\gamma_1$ and $Q$ are sufficiently large constants. Since determining this condition exactly is not efficiently possible, the algorithm is allowed to include some additional clients as problematic. In particular, a client $c$ may be marked as problematic whenever
\[
\exists c' \in B_{\C}\left(c, \gamma_2 \cdot \sqrt{\alpha_c^*}\right)\ 
\text{such that} \quad \alpha_{c',0} \leq \frac{\alpha_{c,0}}{Q},
\]
where $\gamma_2 \gg \gamma_1$ is another sufficiently large constant.
For each client $c\in \C$, we then define an upper dual value $\alpha_{c,1}$ that guarantees that for every client $c$, there is a facility $f$ that is paid w.r.t.\ the dual values $\alpha_{c,1}$ and that is within distance $O\big(\sqrt{\alpha_{c,0}}\big)=O\big(\sqrt{\alpha_{c,1}}\big)$ of $c$. To guarantee this, for all clients $c\in \C'$, i.e., for all clients that are classified as problematic, we set $\alpha_{c,1}:=\alpha_{c,0}$ and for all other clients $c\in \C\setminus\C'$, we set $\alpha_{c,1}:=C_A\cdot\alpha_{c,0}$, where $C_A$ is a sufficiently large constant. 


\paragraph{Step IV: Approximately Paid Facilities.} 
In this step, the algorithm selects a set of facilities $\S$ that are candidates to be opened because they are at least approximately paid for by the computed dual client values. Formally, $\S$ must contain all facilities that are paid w.r.t the dual values $\alpha_{c,1}$, and it may include any facility that is $\kappa$-approximately paid w.r.t the dual values $\alpha_{c,1}$, for some sufficiently large constant $\kappa \geq 1$.


\paragraph{Step V: Dependency Graph.} The \emph{dependency graph} $H = (\S, E_H)$ between the paid and approximately paid facilities $\S$ is defined as follows. There is an edge $\{f, f'\} \in E_H$ between two facilities $f, f' \in \S$ if and only if there exists a client $c$ such that $c$ $\kappa$-approximately contributes to both $f$ and $f'$. Here, $\kappa$ is the same constant as in Step IV. In this way, as long as we do not increase the dual values beyond $\kappa \cdot \alpha_{c,1}$, any facilities that receive contributions from the same client are neighbors in $H$.

In the algorithm, we cannot compute $H$ exactly and therefore compute a supergraph $H' = (\S, E_{H'})$. The graph $H'$ contains all edges of $H$, i.e., $E_H \subseteq E_{H'}$, and it may include additional edges. In our analysis, we show that if a client $c$ $\kappa$-approximately contributes to two facilities $f$ and $f'$, then it must hold that $r_f = \Theta(r_{f'})$ and $\alpha_{c,1} = O(r_f^2)$ (cf.\ \Cref{lem:constant-rad,lemma:sameradius}). In the approximation analysis, this is the only property we require for neighboring nodes in the extended dependency graph $H'$. We therefore allow the algorithm to include any edge $\{f, f'\}$ in $E_{H'}$ for which the radii satisfy $r_f/C_H \leq r_{f'} \leq C_H \cdot r_f$, and the distance condition $\dist(f,f') \leq 2C_H^2 \cdot \max\{r_f, r_{f'}\}$ holds, where $C_H$ is a constant defined as a fixed multiple of $\Gamma$.


\paragraph{Step VI: Compute Clustering.} 
In this step, we utilize the extended dependency graph $H'$ from Step V to define clusters of facilities. The clustering process begins with the selection of a subset of facilities, denoted by $\S_0 \subseteq \S$, which serve as cluster centers. The selection criterion ensures that any two facilities in $\S_0$ are at least four hops apart in the graph $H'$ (and thus also in $H$). This separation condition guarantees that the clusters formed around these centers are well-separated. In particular, it implies that if, for each $f \in \S_0$, either $f$ or one of its neighboring facilities $f'$ in $H'$ is opened, then—as long as the dual values are set to at most $\kappa \cdot \alpha_{c,1}$—each client contributes to at most one open facility.

Ideally, we would like to select the centers $\S_0$ as a $(4, O(1))$-ruling set of $H'$, i.e., in such a way that every facility in $\S$ has a facility in $\S_0$ within constant distance in $H'$. Unfortunately, there is no sufficiently efficient MPC algorithm known that computes such a set $\S_0$, and we therefore have to compute a set $\S_0$ that is sufficiently close to being a $(4, O(1))$-ruling set. We assign every client $c$ to some facility $f_c \in \S$ at distance at most $\cost(c, f) = O(\alpha_{c,0})$ (such a facility exists by \Cref{lem:paidfacility}). We then assign a weight to each facility $f \in \S$, which is equal to the sum of the dual values $\alpha_{c,0}$ of all clients for which $f_c = f$. Given this, we ensure that a $(1 - 1/\log n)$-fraction of the total weight of the facilities in $\S$ lies within a constant hop distance in $H'$ of a facility in $\S_0$. We also ensure that all remaining facilities in $\S$ are within distance $o(\log\log n)$ of $\S_0$.

Given the set $\S_0$, we now define clusters as follows: each facility is assigned to its nearest center in $H'$, and each client $c \in \C$ is assigned to the cluster of its selected facility $f_c$.


\paragraph{Step VII: Opening Facilities.}
Within each cluster, we now open one facility that approximately (within a constant factor) minimizes the sum of squared Euclidean distances to the clients in that cluster. This ensures that the selected facility serves as a cost-effective center for the clients in its assigned area, which turns out to be sufficient to guarantee a constant-factor approximation that satisfies the LMP property.

\subsection{Analysis of the High-Level Facility Location Algorithm}\label{sec:fl-analysis}

In the following, we formally establish the correctness of \Cref{alg:fl-overview}. We begin by proving that the upper dual client values computed in Step III of \Cref{alg:fl-overview} ensure that each client contributes toward at least one paid facility (\Cref{lem:paidfacility}). Next, in \Cref{lem:constant-rad}, we show that if a client contributes to a facility, there exists a well-defined relationship between the client's dual value and the facility's radius. \Cref{lemma:sameradius} then establishes that if a client contributes to multiple facilities, those facilities have asymptotically equal radii. Following this, in \Cref{lemma:dualfeasible}, we demonstrate that the facilities that are actually opened are sparse enough that each client contributes toward the opening cost of at most one opened facility. Finally, \Cref{thm:approx_fl_highlevel} analyzes the approximation quality of \Cref{alg:fl-overview}, proving that the algorithm achieves a constant-factor approximation and satisfies the LMP property.

\Cref{alg:fl-overview} uses a number of constants, many of which are interdependent. Specifically, the pseudocode references the constants $C_R$, $C_D^-$, $C_D^+$, $\gamma_1$, $\gamma_2$, $Q$, $C_A$, $\kappa$, and $C_{H,1}$, $C_{H,2}$. In the analysis, we additionally introduce constants $\eta$, $\zeta$, and $\rho$. All of these constants are deterministically defined as functions of a single parameter $\Gamma \geq 5$, which is the approximation factor guaranteed by the graph construction in \Cref{lemma:graphapprox}. The precise definitions of these constants are provided in \Cref{eq:constants}.

\begin{align}\label{eq:constants}
  C_R &= 9\Gamma, &
  C_D^- &= 2\Gamma^2, &
  C_D^+ &= 8\Gamma^4, &
  \gamma_1 &= 4\Gamma^4, \notag \\
  \gamma_2 &= 9\Gamma^4, &
  Q &= 8\Gamma^4, &
  C_A &= 8\Gamma^8, &
  \kappa &= \Gamma^2, \notag \\
  \eta &= 8000\Gamma^{12}, &
  \zeta &= 1 + 16\sqrt{2}\Gamma^7, &
  \rho &= 128\Gamma^{12}, & C_{H,1} &= 4\cdot C_R^2\cdot\zeta^2,\notag \\
  C_{H,2} &= 12\sqrt{\kappa\cdot\rho}\cdot C_R^3\cdot\zeta^2
\end{align}

If the constants are fixed as described above, the following analysis holds for any $\Gamma \geq 5$. The exact value of $\Gamma$ depends on the implementation of the algorithm in the MPC model, which requires $\Gamma$ to be sufficiently large. For the remainder of \Cref{sec:fl-analysis}, we assume that the constants are set according to Equation~\eqref{eq:constants}.

\begin{lemma}\label{lem:paidfacility}
    For every client $c \in \C$, there exists a facility $f \in \F$ that is paid w.r.t the dual values $\alpha_{c,1}$ and that satisfies
    \begin{align*}
        \max\set{r_f^2, \cost(c,f)} \leq \eta \cdot \alpha_{c,0}.
    \end{align*}
\end{lemma}
\begin{proof}
    Recall that a facility $f$ is considered paid w.r.t the dual values $\alpha_{c,1}$ if
    \begin{equation}\label{eq:lem:paidfacility}
        \sum_{c \in \C} \left[\alpha_{c,1} - \cost(c,f)\right]^+ \geq \lambda.    
    \end{equation}
    
    Suppose, for the sake of contradiction, that there exists a client $c \in \C$ such that no facility $f \in \F$ is both paid and satisfies \eqref{eq:lem:paidfacility}. Without loss of generality, assume that $c$ is a client with minimum $\alpha_{c,0}$-value for which this is the case.

    We distinguish two cases: 1) when $c$ is a non-problematic client ($c \in \C \setminus \C'$), and 2) when $c$ is a problematic client ($c \in \C'$).

    \paragraph{Case 1 (\boldmath$c \in \C \setminus \C'$):}
    Let $f \in \F$ be a facility that minimizes $\max\{r_f^2, \cost(c,f)\}$. By the definition of $\alpha_c^*$ (cf.\  \eqref{eq:initialdualvalues}), we have
    \begin{align*}
        \alpha_c^* = \max \{r_f^2, \cost(c,f)\}.
    \end{align*}
    This in particular implies that $\max\set{r_f^2, \cost(c,f)} \leq \alpha_c^*$. For each client $c' \in B_{\C}(f, r_f)$, the triangle inequality therefore yields
    \begin{equation}\label{eq:d_c_cp_bound}
      d(c, c') \leq d(c,f) + d(f,c') \leq \sqrt{\cost(c,f)} + r_f \leq 2\sqrt{\alpha_c^*}.
    \end{equation}
    Since $c$ is not classified as problematic, every client $c' \in B_{\C}(c, 2\sqrt{\alpha_c^*})\subseteq B_{\C}(c, \gamma_1\cdot\sqrt{\alpha_c^*})$ satisfies $\alpha_{c', 0} > \frac{\alpha_{c, 0}}{Q}$. For the sake of contradiction, assume that there are no problematic clients within the ball $B_{\C}(f, r_f)$. Under this assumption, we have:
    \begin{eqnarray*}
    \sum_{c' \in \C} \left[ \alpha_{c',1} - \cost(c',f) \right]^+ 
    & \stackrel{(B_{\C}(f, r_f)\subseteq \C,\, [x]^+\geq x)}{\geq} & \sum_{c' \in B_{\C}(f, r_f)} \left( \alpha_{c',1} - \cost(c',f) \right) \\
    & \stackrel{(c'\not\in\C')}{=} & \sum_{c' \in B_{\C}(f, r_f)} \left( C_A \cdot \alpha_{c',0} - \cost(c',f) \right)  \\
    & \stackrel{(c\not\in\C')}{\geq} & \sum_{c' \in B_{\C}(f, r_f)} \left( \frac{C_A}{Q} \cdot\alpha_{c,0} - \cost(c',f) \right)  \\
    & \stackrel{(\alpha_{c,0}\geq\alpha_c^*/C_D^+)}{\geq} & \sum_{c' \in B_{\C}(f, r_f)} \left( \frac{C_A}{Q \cdot C_D^+} \cdot\alpha_c^* - \cost(c',f) \right) \\
    & \stackrel{(r_f^2\leq \alpha_c^*)}{\geq} & \sum_{c' \in B_{\C}(f, r_f)} \left( \frac{C_A}{Q \cdot C_D^+} \cdot r_f^2 - \cost(c',f) \right) \\
    & \stackrel{(C_A\geq Q\cdot C_D^+,\ \text{Eq.}\, \eqref{eq:facilityradius})}{\geq} & \lambda.
    \end{eqnarray*}
    Therefore, $f$ is a paid facility w.r.t the dual values $\alpha_{c,1}$. Because $\alpha_{c,0} \geq \alpha_c^*/C_D^+ = \alpha_c^*/(8\Gamma^4)$ and $\eta = 8000\Gamma^{12}$, we also have $\cost(c,f) \leq \alpha_c^* \leq \eta \cdot \alpha_{c,0}$. Since we assumed that there is no facility $f$ that is paid and satisfies~\eqref{eq:lem:paidfacility}, this contradicts the assumption that there is no problematic client $c' \in B_{\C}(f, r_f)$. Consequently, there must exist some $c' \in \C' \cap B_{\C}(f, r_f)$. Since $c' \in B_{\C}(f, r_f)$, it follows that $\alpha_{c'}^* \leq r_f^2 \leq \alpha_c^*$. This further implies that
    \begin{equation*}
        C_D^-\cdot \alpha_{c',0} \le \alpha_{c'}^* \leq \alpha_c^* \leq C_D^+ \cdot \alpha_{c, 0}.
    \end{equation*}
    Since $c'$ is problematic, there exists a client $c'' \in B_{\C}(c', \gamma_2 \cdot \sqrt{\alpha_{c'}^*})$ such that $\alpha_{c'', 0} \leq \frac{\alpha_{c', 0}}{Q}$. Consequently, we obtain
    \begin{equation}\label{eq:alpha_cpp_bound}
        \alpha_{c'', 0} \leq \frac{C_D^+}{C_D^-\cdot Q} \cdot \alpha_{c, 0}.
    \end{equation}
    Because for $\Gamma\geq 5$, we have $Q = 8\Gamma^4 > C_D^+/C_D^- = 8\Gamma^4/(2\Gamma^2)$, it follows that $\alpha_{c'', 0} < \alpha_{c,0}$. Since we assumed that $c$ has the smallest $\alpha_{c, 0}$ among all clients $c$ for which there is no paid facility $f$ with $\cost(c,f) \leq \eta \cdot \alpha_{c, 0}$, there must exist a paid facility $f$ such that $\cost(c'',f) \leq \eta \cdot \alpha_{c'', 0}$. Consequently, we derive the following bound:
    \begin{eqnarray*}
        \cost(c,f) & \stackrel{(\text{Eq.}\,\eqref{eq:triangleineq})}{\le} & 
        3\cdot \left(\cost(c, c') + \cost(c', c'') + \cost(c'', f)\right)\\
        & \le & 12\alpha_c^* + 12\gamma_2^2\alpha_{c'}^* + 3\eta\cdot \alpha_{c'', 0}\\
        & \stackrel{(\alpha_{c'}^*\leq \alpha_c^*,\ \text{Eq.}\,\eqref{eq:alpha_cpp_bound})}{\le} & 12 C_D^+\cdot \alpha_{c, 0} + 12\gamma^2_2\cdot C_D^+\cdot \alpha_{c, 0} + 3\eta\cdot \frac{C_D^+}{C_D^-\cdot Q}\alpha_{c, 0}\\
        & \leq & \eta\cdot\alpha_{c, 0}
    \end{eqnarray*}
    The second inequality follows from \eqref{eq:d_c_cp_bound}, from $c'' \in B_{\C}(c', \gamma_2 \cdot \sqrt{\alpha_{c'}^*})$, and from $\cost(c'',f) \leq \eta \cdot \alpha_{c'',0}$. The last inequality follows from the definitions of the constants $C_D^+$, $C_D^-$, $\gamma_2$, $Q$, and $\eta$, and from $\Gamma \geq 5$ (cf.\ Eq.~\eqref{eq:constants}). This contradicts the assumption that there is no paid facility $f$ for which \eqref{eq:lem:paidfacility} holds, and it thus concludes Case 1.

    \medskip

    \paragraph{Case 2 (\boldmath$c \in \C'$):} 
    Suppose $c$ is a problematic client. By definition, there then exists a client $c' \in B_{\C}(c, \gamma_2 \cdot \sqrt{\alpha_c^*})$ for which $\alpha_{c',0} \le \frac{\alpha_{c,0}}{Q} < \alpha_{c,0}$. Since we assumed that $c$ has the smallest $\alpha_{c,0}$ among all clients for which there is no paid facility such that \eqref{eq:lem:paidfacility} holds, this implies that there exists a paid facility $f$ such that $\cost(c',f) \le \eta \cdot \alpha_{c',0}$. Using this, we obtain
    \begin{eqnarray*}
        \cost(c,f) & \stackrel{\text{Eq.}\,\eqref{eq:triangleineq}}{\leq} &
        2\cdot\left(\cost(c, c') + \cost(c', f)\right) \\
        & \leq & 8\gamma_2^2 \cdot \alpha_c^* + 2\eta \cdot \alpha_{c',0} \\
        & \stackrel{(\alpha_{c',0} \leq \alpha_{c,0}/Q)}{\leq} & 8\gamma_2^2 \cdot C_D^+ \cdot \alpha_{c,0} + 2\eta \cdot \frac{\alpha_{c,0}}{Q} \\
        & \leq & \eta \cdot \alpha_{c,0}.
    \end{eqnarray*}
    The second inequality follows from $c' \in B_{\C}(c, \gamma_2 \cdot \sqrt{\alpha_c^*})$ and from 
    $\cost(c',f) \le \eta \cdot \alpha_{c',0}$. The last inequality follows from the definitions of the constants $C_D^+$, $\gamma_2$, $Q$, and $\eta$, and from $\Gamma \geq 5$ (cf.\ Eq.~\eqref{eq:constants}). This again contradicts the assumption that there is no paid facility $f$ for which \eqref{eq:lem:paidfacility} holds, and it thus concludes Case 2 and therefore the proof of the lemma.
\end{proof}

\begin{lemma}\label{lem:constant-rad}
    Let $c \in \C$ be a client and $f \in \F$ a facility. If client $c$ $\kappa$-approximately contributes to facility $f$ w.r.t the dual values $\alpha_{c,1}$, then
    \begin{align*}
        \alpha_{c,1} \leq \rho\cdot r_f^2,\quad \text{where }\rho:=\frac{C_A\cdot C_D^+\cdot Q}{(C_D^-)^2} = 128\Gamma^{12}.
    \end{align*}
\end{lemma}

\begin{proof}
    We again distinguish two cases: 1) when $c$ is a non-problematic client ($c \in \C \setminus \C'$) and 2) when $c$ is a problematic client ($c \in \C'$).
    
    \paragraph{Case 1 (\boldmath$c \in \C \setminus \C'$):}
    Assume, for the sake of contradiction, that $\alpha_{c,1} > \rho\cdot r_f^2$. This directly implies that
    \begin{equation}\label{eq:alphacstar_lower}
        \alpha_c^*\ \geq\ C_D^-\cdot \alpha_{c,0}\ \stackrel{(c\in\C)}{=}\ \frac{C_D^-}{C_A}\cdot \alpha_{c,1}\ >\ \frac{C_D^-}{C_A}\cdot \rho\cdot r_f^2.
    \end{equation}
    By the definition of $\alpha_c^*$, we know that $\alpha_c^*\leq \max\big\{r_f^2, \cost(c,f)\big\}$. This implies that
    \begin{equation}\label{eq:cost_c_f_lower}
    \cost(c,f) \geq \alpha_c^* > \frac{C_D^-}{C_A}\cdot \rho \cdot r_f^2 \geq 9\cdot r_f^2.
    \end{equation}
    The last inequality follows because $\rho\geq 9C_A/C_D^-$.
    Since $B_{\C}(f, r_f)$ is non-empty by the definition of $r_f$, there exists a client $c' \in B_{\C}(f, r_f)$. Using this, we can derive the following bound:
    \begin{eqnarray}
        \dist(c,c') & \leq &
        \dist(c,f) + \dist(c',f)\nonumber\\
        & \leq & \dist(c,f) + r_f\nonumber\\
        & \stackrel{(\text{Eq.}\,\eqref{eq:cost_c_f_lower})}{\leq} &  
        \dist(c,f) + \frac{\dist(c,f)}{3}\nonumber \\
        & \leq & \frac{4}{3} \cdot \sqrt{\kappa \cdot \frac{C_A}{C_D^-} \cdot \alpha_c^*}\nonumber\\
        & \le & \gamma_1 \cdot \sqrt{\alpha_c^*}.\label{eq:cp_in_problematicradius}
    \end{eqnarray}
    The second-to-last inequality follows because $c$ $\kappa$-approximately contributes to $f$, and thus $\kappa\cdot\alpha_{c,1} \geq \cost(c,f)$, and from the fact that for $c \in \C \setminus \C'$, $\alpha_{c,1} = C_A \cdot \alpha_{c,0} \leq \frac{C_A}{C_D^-} \cdot \alpha_c^*$. The last inequality holds because $\gamma_1 \geq \sqrt{2 \cdot \kappa \cdot C_A / C_D^-}$.

    Note that because $c' \in B_{\C}(f, r_f)$, we have $\max\set{r_f^2, \cost(c',f)}=r_f^2$ and therefore $\alpha_{c'}^*\leq r_f^2$. This implies
    \[
    \alpha_{c',0} \leq \frac{\alpha_{c'}^*}{C_D^-} \leq \frac{r_f^2}{C_D^-}
    \leq \frac{C_A}{\rho\cdot (C_D^-)^2}\cdot \alpha_c^*\leq
    \frac{C_A\cdot C_D^+}{\rho\cdot (C_D^-)^2}\cdot \alpha_{c,0}
    \leq \frac{\alpha_{c,0}}{Q}.
    \]
    The last inequality follows because $\rho \geq C_A\cdot C_D^+\cdot Q/(C_D^-)^2$. Thus, client $c$ is a problematic client, contradicting the assumption that $c$ is non-problematic, i.e., $c \in \C \setminus \C'$. This concludes Case 1.
    
    \paragraph{Case 2 (\boldmath$c \in \C'$):}
    In this case, we have $\alpha_{c, 1} = \alpha_{c, 0} \le \alpha_c^*/C_D^-$. By the definition of $\alpha_c^*$ (Eq.~\eqref{eq:initialdualvalues}), we know that $\alpha_c^* \le \max\{r_f^2, \cost(c,f)\}$, which implies
    \begin{align*}
        \alpha_{c, 1} &\le \frac{\max\{r_f^2, \cost(c,f)\}}{C_D^-}.
    \end{align*}
    Since client $c$ approximately contributes to facility $f$, we also have $\cost(c,f) \le \kappa \cdot \alpha_{c, 1}$. Substituting into the above inequality, we obtain
    \begin{align*}
        \cost(c,f) &\le \frac{\kappa}{C_D^-} \cdot \max\{r_f^2, \cost(c,f)\}.
    \end{align*}
    Given that $\kappa < C_D^-$, it follows that $\max\{r_f^2, \cost(c,f)\} = r_f^2$. We can therefore conclude that
    \begin{equation*}
        \alpha_{c, 1} \le \frac{r_f^2}{C_D^-} \leq \rho\cdot r_f^2,
    \end{equation*}
    which concludes Case 2 and thus the proof.
\end{proof}

\begin{lemma}\label{lemma:sameradius}
    Let $c \in \C$ be a client, and let $f, f' \in \F$ be two facilities. If the client $c$ $\kappa$-approximately contributes to both facilities $f$ and $f'$ w.r.t the dual values $\alpha_{c,1}$, then the radii of these facilities satisfy 
    \begin{equation}\label{eq:sameradius}
        \max\set{r_f, r_{f'}} \leq \zeta \cdot \min\set{r_f, r_{f'}},
        { \quad\text{where}\ \zeta=
        1+2\cdot\frac{\sqrt{C_A\cdot C_D^+\cdot Q\cdot \kappa}}{C_D^{-}}}.
    \end{equation}
    We further also have $\dist(f,f') < \frac{\zeta}{2}\cdot \min\set{r_f, r_{f'}}$.
\end{lemma}
\begin{proof}
    Without loss of generality, assume that $r_f\leq r_{f'}$. By \Cref{lem:constant-rad}, we have
    \begin{equation}\label{eq:def_of_rho}
        \alpha_{c,1} \leq \rho\cdot r_f^2,\quad\text{where }\rho = \frac{C_A\cdot C_D^+\cdot Q}{(C_D^-)^2}.     
    \end{equation}
    Since, $c$ $\kappa$-approximately contributes to $f$ and $f'$, we further have
    \begin{equation*}
        \cost(c,f) \le \kappa \cdot \alpha_{c,1}
        \quad\text{and}\quad
        \cost(c,f') \le \kappa \cdot \alpha_{c,1}.
    \end{equation*}
    Combining the two inequalities gives
    \[
    \dist(c,f) \le \sqrt{\kappa \cdot \rho}\cdot r_f
    \quad\text{and}\quad
    \dist(c,f') \le \sqrt{\kappa \cdot \rho}\cdot r_f
    \]
    and therefore (by the triangle inequality)
    \begin{equation}\label{eq:dist_f_fp}
        \dist(f,f') \leq 2\cdot\sqrt{\kappa\cdot\rho}\cdot r_f.
    \end{equation}
    
    Note that the ball of radius $r_f$ around facility $f$ is contained within the ball of radius $r_f + \dist(f,f')$ around facility $f'$. We define a new radius $r'$ as
    \begin{align*}
        r' := r_f + 2\dist(f, f').
    \end{align*}
    Now consider the ball of radius $r'$ around facility $f'$. We derive
    \begin{align*}
        \sum_{\bar{c} \in B_{\C}(f', r')} \!\!\!\!(r'^2 - \cost(\bar{c},f')) &\geq \sum_{\bar{c} \in B_{\C}(f, r_f)} \!\!\!\!(r'^2 - \cost(\bar{c},f'))\\
        &= \sum_{\bar{c} \in B_{\C}(f, r_f)} \!\!\!\!\left(r_f^2 + 4r_f \dist(f,f') + 4\dist^2(f, f') - \cost(\bar{c},f')\right)\\
        &\geq \sum_{\bar{c} \in B_{\C}(f, r_f)} \!\!\!\!\left(r_f^2 + 4r_f \dist(f,f') + 4\dist^2(f, f') - (\dist(\bar{c},f) + \dist(f,f'))^2\right)\\
        &\geq \sum_{\bar{c} \in B_{\C}(f, r_f)} \!\!\!\!\left(r_f^2 - \cost(\bar{c},f) + 4r_f \dist(f,f') - 2r_f \dist(f,f')\right)\\
        &= \sum_{\bar{c} \in B_{\C}(f, r_f)} \!\!\!\!\left(r_f^2 - \cost(\bar{c},f) + 2r_f \dist(f,f')\right)\\
        &\geq \sum_{\bar{c} \in B_{\C}(f, r_f)} \!\!\!\!\left(r_f^2 - \cost(\bar{c},f)\right)\\
        &= \lambda.
    \end{align*}
    The first inequality follows because $B_{\C}(f, r_f)\subseteq B_{\C}(f', r')$. The equation in the second line follows from the definition of $r'$. The second inequality (in line 3) follows from the triangle inequality. The third inequality follows because $\bar{c}\in B_{\C}(f, r_f)$ and thus $\dist(\bar{c},f)\leq r_f$. By \eqref{eq:facilityradius}, the fact that $\sum_{\bar{c} \in B_{\C}(f', r')} (r'^2 - \cost(\bar{c},f'))\geq\lambda$ holds implies that $r_{f'}\leq r'$. Inequality \eqref{eq:dist_f_fp} therefore implies that
    \begin{align*}
        r_{f'} \le r_f + 2d(f, f') \le r_f + 2\sqrt{\rho\kappa} \cdot r_f = \left(1 + 2\sqrt{\rho\kappa}\right)\cdot r_f.
    \end{align*}
    Because we assumed that $r_f \le r_{f'}$, together with \eqref{eq:def_of_rho} this directly implies that statement of the lemma.
\end{proof}

Before analyzing the approximation guarantee of the exact algorithm as stated in \Cref{alg:fl-overview}, we first analyze an alternative algorithm that is closer to the sequential primal-dual algorithm, but which cannot be easily implemented in the MPC model. In the alternative algorithm, only Step VII of \Cref{alg:fl-overview} is modified, where we decide which facilities to open. We now first describe this alternative Step VII.

\paragraph{Alternative Step VII: Opening Facilities.} Let $f \in \S_0$ be a cluster center. We define $\C_f$ to be the set of clients $c$ that $\kappa$-approximately contribute to $f$ w.r.t the dual values $\alpha_{c,1}$. Note that if there exists a client that $\kappa$-approximately contributes to two facilities $f_1, f_2 \in \S$, then there is an edge between $f_1$ and $f_2$ in the dependency graph $H$ and thus also in the extended graph $H'$. Therefore, the clients in $\C_f$ can only $\kappa$-approximately contribute to $f$ and its direct neighbors in $H$ (and possibly to facilities that are not $\kappa$-approximately paid and may not be in $\S$).

The algorithm now computes a dual solution $(\alpha_c, \beta_{cf})$ to the LP \eqref{eq:fl-dual} as follows. For all clients $c$ that are not in the set $\C_f$ for any $f \in \S_0$ (i.e., clients that do not $\kappa$-approximately contribute to any cluster center), we set $\alpha_c := \alpha_{c,0}$. Furthermore, for each $f \in \S_0$, let $N_H^+(f)$ denote the set consisting of $f$ and all its neighbors in $H$.

For each $f \in \S_0$, we determine a value $\xi_f \in \mathbb{R}$, defined as the minimum $\xi \in \mathbb{R}$ for which there exists a facility $f' \in N_H^+(f)$ such that
\[
\sum_{c \in \C_f} \big[\alpha_{c,0} + \xi \cdot (\kappa \cdot \alpha_{c,1} - \alpha_{c,0}) - \cost(c,f')\big]^+ + \sum_{c \in \C \setminus \C_f} [\alpha_{c,0} - \cost(c,f')]^+ = \lambda.
\]

For each $c \in \C_f$, we then set $\alpha_c := \alpha_{c,0} + \xi_f \cdot (\kappa \cdot \alpha_{c,1} - \alpha_{c,0})$.
Further, for all $c$ and $f$, we define $\beta_{cf} := [\alpha_c - \cost(c,f)]^+$. For each $f \in \S_0$, we open exactly one facility $f' \in N_H^+(f)$ that satisfies
\[
\sum_{c \in \C_f} \big[\alpha_c - \cost(c,f')\big]^+ = \sum_{c \in \C_f} \beta_{cf} = \lambda.
\]

We denote the set of open facilities by $\F'$. By construction, we have $|\F'| = |\S_0|$. The following lemma shows that the computed solution is dual feasible and satisfies all properties required for the analysis of the approximation factor.

\begin{lemma}\label{lemma:dualfeasible}
    The dual solution computed by the above algorithm is feasible for the LP \eqref{eq:fl-dual}. Moreover, the algorithm guarantees that for every client $c \in \C$:
    \begin{itemize}
        \item[(I)] We have $\alpha_{c,0} \leq \alpha_c \leq \kappa \cdot \alpha_{c,1}$, and
        \item[(II)] There is at most one open facility $f \in \F'$ for which $\beta_{cf} = [\alpha_c - \cost(c,f)]^+ > 0$.
    \end{itemize}
\end{lemma}

\begin{proof}
    We first show that for every client $c \in \C$, we have $\alpha_{c,0} \leq \alpha_c \leq \kappa \cdot \alpha_{c,1}$. To show that $\alpha_c \geq \alpha_{c,0}$, we need to show that $\xi_f \geq 0$ for every $f \in \S_0$. This follows from the fact that if all clients choose $\alpha_{c,0}$ as their dual value, then no facility is paid. That is, we need to show that $\sum_{c \in \C} [\alpha_{c,0} - \cost(c,f)]^+ < \lambda$ for every $f \in \F$. If for all $c \in \C$, $\alpha_{c,0} \leq \cost(c,f)$, we have $\sum_{c \in \C} [\alpha_{c,0} - \cost(c,f)]^+ = 0$, which is clearly less than $\lambda$. Let us therefore consider a facility $f \in \F$ for which there exists at least one client $c$ with $\alpha_{c,0} > \cost(c,f)$:
    \begin{eqnarray*}
        \sum_{c \in \C} \big[\alpha_{c,0} - \cost(c,f)\big]^+
        & < &
        \sum_{c \in \C} \big[\alpha_c^* - \cost(c,f)\big]^+\\
        & \stackrel{\text{Eq.}\,\eqref{eq:initialdualvalues}}{\leq} &
        \sum_{c \in \C} \left[\max\set{r_f^2, \cost(c,f)} - \cost(c,f)\right]^+\\
        & = & \sum_{c \in \C} \big[r_f^2 - \cost(c,f)\big]^+\\
        & \stackrel{\text{Eq.}\,\eqref{eq:facilityradius}}{=} & \lambda.
    \end{eqnarray*}
    The first inequality follows because for all clients $c$, $\alpha_{c,0} < \alpha_c^*$, and because there exists a client $c$ for which $\alpha_{c,0} > \cost(c,f)$. In order to have a paid facility within $N_H^+(f)$ for some $f \in \S_0$, one therefore needs to choose $\xi_f > 0$.

    To show that $\xi_f \leq 1$, observe that every facility $f \in \S$—and therefore also every facility $f \in \S_0$—is $\kappa$-approximately paid w.r.t the dual values $\alpha_{c,1}$. This directly implies that if all clients $c \in \C_f$ choose $\kappa \cdot \alpha_{c,1}$ as their dual value, then $f$ is paid. The value of $\xi_f$ can therefore not exceed $1$.

    It remains to show that the computed solution is feasible for the dual LP \eqref{eq:fl-dual} and that Property (II) holds. To prove dual feasibility, first note that we clearly have $\alpha_c \geq 0$ and $\beta_{cf} \geq 0$ for all $c \in \C$ and $f \in \F$. The inequalities $\alpha_c - \beta_{cf} \leq \cost(c,f)$ also directly hold for all $c \in \C$ and $f \in \F$ because we set the $\beta_{cf}$-variables to $\beta_{cf} = [\alpha_c - \cost(c,f)]^+$. It therefore remains to show that for all $f \in \F$, we have $\sum_{c \in \C} \beta_{cf} \leq \lambda$. If we set all dual client variables to $\alpha_{c,0}$, this directly follows from the above calculations. Note that we set $\alpha_c > \alpha_{c,0}$ only for clients $c \in \C_f$ for some $f \in \S_0$. Because for each $f \in \S_0$ and each client $c \in \C_f$, the set $N_H^+(f)$ contains all the facilities for which $\kappa \cdot \alpha_{c,1} \geq \cost(c,f)$, increasing the values of the clients in $\C_f$ can only affect the value of $\sum_{c \in \C} \beta_{cf'}$ for facilities $f' \in N_H^+(f)$. Further, because facilities in $\S_0$ are at distance at least $4$ in $H$, the sets $N_H^+(f)$ for different $f \in \S_0$ are disjoint. For a specific $f \in \S_0$, the choice of the value $\xi_f$ guarantees that no facility in $N_H^+(f)$ becomes overpaid, and thus the computed solution is feasible for \eqref{eq:fl-dual}.

    Let us now prove that Property (II) holds, i.e., that for each client $c \in \C$, there is at most one open facility $f$ for which $\alpha_c > \cost(c,f)$. Note that we can only have $\alpha_c > \cost(c,f)$ if $\kappa \cdot \alpha_{c,1} > \cost(c,f)$, and thus only if $c$ $\kappa$-approximately contributes to $f$ w.r.t the dual values $\alpha_{c,1}$. Therefore, if $\alpha_c > \cost(c,f_1) > 0$ and $\alpha_c > \cost(c,f_2) > 0$ for two facilities $f_1$ and $f_2$, then $f_1$ and $f_2$ must be neighbors in $H$. Because the cluster centers $\S_0$ are at distance at least $4$ from each other, and all the open facilities are either cluster centers or direct neighbors of cluster centers, the open facilities of two different clusters are at least two hops apart in $H$. It is therefore not possible that a client contributes toward the opening cost of two open facilities (i.e., Property (II) holds).
\end{proof}

We are now ready to prove that \Cref{alg:fl-overview}, with the alternative Step VII, computes a constant-factor approximate solution to the facility location problem that satisfies the LMP property.

\begin{lemma}\label{lemma:approx_alternativealg}
    \Cref{alg:fl-overview}, with the alternative Step VII, opens a set of facilities $\F'$ and connects each client $c \in \C$ to an open facility $f'_c \in \F'$ such that
    \[
    \sum_{c \in \C} \cost(c, f'_c)\ \leq\ \Lambda \cdot \left(
    \mathsf{OPT} - |\F'| \cdot \lambda \right)
    \]
    for some constant $\Lambda \geq 1$, where $\mathsf{OPT}$ denotes the objective value of an optimal solution to the given facility location instance.
\end{lemma}
\begin{proof}
    We partition the clients $\C$ into three sets: $\C_O$, $\C_N$, and $\C_F$. The clients $c \in \C_O$ are those for which there exists an open facility $f \in \F'$ such that $\alpha_c > \cost(c,f)$. The clients in $\C_N$ are those not in $\C_O$ and are labeled $\C^+$ in Step VI of \Cref{alg:fl-overview}. That is, $\C_N$ consists of clients $c$ for which the facility $f_c$ (as determined in Step VI of the algorithm) is within constant distance in $H'$ from its cluster center. For concreteness, assume that \Cref{alg:fl-overview} guarantees that for clients $c \in \C_N$, the facility $f_c$ is within at most $d_0 = O(1)$ hops in $H'$ from its cluster center. Finally, the clients in $\C_F$ are those that are far from their cluster center. According to \Cref{alg:fl-overview}, for clients $c \in \C_F$, the facility $f_c$ is guaranteed to be within $o(\log\log n)$ hops from its cluster center. We analyze the connection costs of the three sets of clients separately.
   
    Let us first focus on the clients in $\C_O$, which contribute to the opening cost of some open facility. By \Cref{lemma:dualfeasible}, for every $c \in \C_O$, there is exactly one open facility $f \in \F'$ for which $\alpha_c > \cost(c, f)$. We have:
    \begin{eqnarray*}
        \sum_{c \in \C_O} \cost(c, f_c') & = &
        \sum_{c \in \C_O} \big(\alpha_c - (\alpha_c - \cost(c, f_c'))\big)\\
        & = & 
        \sum_{c \in \C_O} \big(\alpha_c - [\alpha_c - \cost(c, f_c')]^+\big)\\
        & = & \sum_{c \in \C_O} \alpha_c 
        - \sum_{f \in \F'} \sum_{c \in \C} [\alpha_c - \cost(c, f)]^+\\
        & = & \sum_{c \in \C_O} \alpha_c - |\F'| \cdot \lambda.
    \end{eqnarray*}
    
    For clients $c \in \C \setminus \C_O$, the connection cost depends in particular on how far the facility $f_c \in \S$ is from its cluster center. Assume that $f_c$ is $d_c$ hops away (in $H'$) from the cluster center. If the cluster center is $d_c$ hops away from $f_c$, then the opened facility $f_c'$ is at most $d_c + 1$ hops away from $f_c$.
    
    Following the definition of the dependency graph $H'$, there exists an edge between two facilities $f$ and $f'$ if and only if 
    \[
    \dist(f, f') \leq O(1)\cdot \max\{r_f, r_{f'}\}.
    \] 
    Consequently, along any path in $H'$, facility radii can change (increase or decrease) by at most a constant factor per hop. Therefore, we obtain:
    \[
    \dist(f_c, f_c') \leq \sum_{i = 1}^{d_c} O(r_{f_i}) \leq r_{f_c}\cdot O(1)^{d_c} = 2^{O(d_c)}\cdot r_{f_c}.
    \]
    Since client $c$ approximately contributes to facility $f_c$, we have $\dist(c, f_c) = O(r_{f_c}) = O(\sqrt{\alpha_c})$ by \Cref{lem:constant-rad}. Combining these results, we derive:
    \[
    \dist(c, f_c') \leq \dist(c, f_c) + \dist(f_c, f_c') \leq O(\sqrt{\alpha_c}) + 2^{O(d_c)}\cdot \sqrt{\alpha_c} = 2^{O(d_c)}\cdot \sqrt{\alpha_c}.
    \]
    Squaring both sides yields
    \[
    \cost(c, f_c') = \dist^2(c, f_c') \leq 2^{O(d_c)}\cdot \alpha_c.
    \]
    For each client $c\in \C_N$, we thus have $\cost(c,f_c') = O(1)\cdot \alpha_c$, and therefore,
    \[
        \sum_{c\in \C_N} \cost(c,f'_c) = O(1)\cdot \sum_{c\in \C_N} \alpha_c.
    \]
    For clients $c\in \C_F$, we have
    \[
        \sum_{c\in \C_F} \cost(c,f'_c) = 2^{o(\log\log n)}\cdot \sum_{c\in \C_F} \alpha_c = o(\log n)\cdot\sum_{c\in \C_F}\alpha_c.
    \] 
    Since \Cref{alg:fl-overview} guarantees that $\sum_{c\in \C_F}\alpha_{c,0} \leq \frac{1}{\log n}\cdot\sum_{c\in \C}\alpha_{c,0}$, it follows from the definition of $\alpha_c$ that:
    \[
        \sum_{c\in \C_F} \cost(c,f'_c) = O(1)\cdot \sum_{c\in \C}\alpha_c.
    \]
    Combining the bounds obtained for $\sum_{c\in \C_O} \cost(c,f'_c)$, $\sum_{c\in \C_N} \cost(c,f'_c)$, and $\sum_{c\in \C_F} \cost(c,f'_c)$, we conclude that the lemma holds for a sufficiently large constant $\Lambda \geq 1$.
\end{proof}

We can now prove that the original high-level algorithm described in \Cref{alg:fl-overview} (i.e., with Step VII as stated in the pseudocode) also computes a constant-approximate solution to the facility location problem that satisfies the LMP property.

\begin{theorem}\label{thm:approx_fl_highlevel}
    \Cref{alg:fl-overview} opens a set of facilities $\F'$ and connects each client $c \in \C$ to an open facility $f'_c \in \F'$ such that
    \[
        \sum_{c \in \C} \cost(c, f'_c)\ \leq\ \Lambda \cdot \left(
        \mathsf{OPT} - |\F'|\cdot\lambda\right),
    \]
    for some constant $\Lambda \geq 1$, where $\mathsf{OPT}$ denotes the objective value of an optimal solution to the given facility location instance.
\end{theorem}

\begin{proof}
Let $\F'$ denote the set of facilities opened by \Cref{alg:fl-overview}. Further, let $\F''$ denote the set of facilities opened by the variant of the algorithm that uses the alternative Step VII. Both algorithms compute exactly the same clustering, and each opens exactly one facility per cluster. Hence, it clearly holds that $|\F'| = |\F''|$. Additionally, in both algorithms, each client is connected to the facility opened in its own cluster.

For each cluster center $f \in \S_0$, let $\C^{(f)}$ be the set of clients assigned to the cluster of $f$. For each $f \in \S_0$, let $f'_f \in \F'$ be the facility opened by \Cref{alg:fl-overview} in the cluster of $f$. Similarly, let $f''_f \in \F''$ be the facility opened by the variant of \Cref{alg:fl-overview} (with the alternative Step VII) in the cluster of $f$. Moreover, let $f^*_f$ be the facility in the cluster of $f$ minimizing the total cost $\sum_{c \in \C^{(f)}} \cost(c, f^*_f)$. Step VII of \Cref{alg:fl-overview} guarantees that for some constant $\psi \geq 1$:
\[
\sum_{c \in \C^{(f)}} \cost(c, f'_f) \leq \psi \cdot \sum_{c \in \C^{(f)}} \cost(c, f^*_f).
\]

By the definition of $f^*_f$, it also clearly holds that:
\[
\sum_{c \in \C^{(f)}} \cost(c, f^*_f) \leq \sum_{c \in \C^{(f)}} \cost(c, f''_f).
\]

Since every client belongs to exactly one cluster, combining these inequalities gives us:
\[
\sum_{c \in \C} \cost(c, f'_c) \leq \psi \cdot \sum_{c \in \C} \cost(c, f''_c).
\]

Using $|\F'| = |\F''|$ and applying the result from \Cref{lemma:approx_alternativealg}, we obtain:
\[
\sum_{c \in \C} \cost(c, f'_c) \leq \psi \cdot \Lambda'' \cdot (\mathsf{OPT} - |\F'| \cdot \lambda),
\]
where $\Lambda''$ is the approximation ratio provided by \Cref{lemma:approx_alternativealg} for the algorithm variant using the alternative Step VII. Thus, the claim of the theorem follows directly from \Cref{lemma:approx_alternativealg}.
\end{proof}

\subsection{High-Level {\it k}-Means to Facility Location Reduction Algorithm}\label{sec:kmeans-highlevel-alg}

    We next show how to efficiently compute a constant approximation for the $k$-means problem by utilizing the existence of a constant factor LMP approximation for the facility location problem. Our approach builds on the randomized $k$-median algorithm presented by \cite{JV01}. We first describe the algorithm and then analyze its approximation guarantees in \Cref{sec:kmeans-highlevel-analysis}. The algorithm consists of the following three steps.

    \begin{enumerate}[(I)]
        \item Given a $k$-means instance with point set $P$, $|P|=n$ and an opening cost $\lambda \geq 1$, one can define a facility location instance by setting $\F = \C = P$. In the first step, the algorithm finds two opening costs $\lambda_1$ and $\lambda_2$ and two corresponding feasible primal and dual facility location solutions. For each solution corresponding to opening cost $\lambda_i$ (for $i \in \set{1,2}$), the set of opened facilities is $\F_i$, each client $c \in \C$ is assigned to a facility $f_{c}^{(i)} \in \F_i$, the dual client value of $c$ is $\alpha_{c}^{(i)}$, and the dual values $\beta_{cf}^{(i)}$ are equal to $\beta_{cf}^{(i)}=\big[\alpha_{c}^{(i)}-\cost(c,f)\big]^+$. Both solutions are constant approximations with the LMP property, i.e., for $i \in \set{1,2}$, we have
        \begin{equation}\label{eq:kmeans_LMP}
            \sum_{c\in\C}\cost\big(c,f_c^{(i)}\big) \leq \Lambda \cdot
            \left(\sum_{c\in \C} \alpha_c^{(i)} - \lambda_i \cdot |\F_i|\right)
        \end{equation}
        for some constant $\Lambda \geq 1$. We further require $k_1 := |\F_1| \leq k$, $k_2 := |\F_2| \geq k$, and $\lambda_2 \leq \lambda_1 \leq 2\lambda_2$. If $k_1 = k$ or $k_2 = k$, we output the corresponding facility location solution as the final $k$-means solution and stop.

        \item Determine a set $\F_2'$ of size $|\F_2'| = k_1$ as follows. For every $f \in \F_1$, add a facility $f_2'(f) \in \F_2$ to $\F_2'$ such that 
        \begin{equation}\label{eq:kmeans_rounding_facility}
            \dist\big(f, f_2'(f)\big) \leq \delta \cdot \min_{f' \in \F_2} \dist(f, f')
        \end{equation}
        for some constant $\delta \geq 1$. If the resulting set $\F_2'$ is of size $< k_1$, add arbitrary additional $k_1 - |\F_2'|$ facilities from $\F_2 \setminus \F_2'$ to $\F_2'$.

        \item Define $a := \frac{k_2 - k}{k_2 - k_1}$ and $b := \frac{k - k_1}{k_2 - k_1}$. The algorithm then determines a set $Z$ of $k$ centers for the $k$-means solution as follows: 
        \begin{enumerate}[(A)]
            \item With probability $a$, it adds the $k_1$ points in $\F_1$ to $Z$, otherwise (with probability $1-a=b$) it adds the $k_1$ points in $\F_2'$ to $Z$.
            \item Additionally, it selects $k - k_1$ uniformly random facilities from $\F_2 \setminus \F_2'$ and adds them to $Z$.
        \end{enumerate} For each client $c$, we determine a facility $\phi(c)\in Z$ to which $c$ is connected to as follows: 
        \begin{itemize}
            \item If $f_c^{(2)}\in \F_2'$ and the set $\F_1$ is sampled in Step (A) above, we set $\phi(c):=f_c^{(1)}$.
            \item If $f_c^{(2)}\in \F_2'$ and the set $\F_2'$ is sampled in Step (A) above, we set $\phi(c):=f_c^{(2)}$.
            \item If $f_c^{(2)}\in \F_2\setminus\F_2'$, we set $\phi(c):=f_c^{(2)}$ if $f_c^{(2)}$ is sampled as one of the additional $k-k_1$ facilities in Step (B) above. Otherwise we set $\phi(c):=f_1^{(1)}$ if $\F_1$ is sampled in Step (A) above. Otherwise, we set $\phi(c):=f_2'\big(f_1^{(1)}\big)$.
        \end{itemize}
    \end{enumerate}

\subsection{High-Level {\it k}-Means to Facility Location Reduction Analysis}\label{sec:kmeans-highlevel-analysis}
We now analyze the $k$-means algorithm described in \Cref{sec:kmeans-highlevel-alg}. We first show that the facility location that was computed for opening cost $\lambda_1$ in Step (I) of the algorithm is also a constant-factor LMP approximation for opening cost $\lambda_2$.

\begin{lemma}\label{lemma:FL_openingcostreduction}
    Consider the facility location solution for opening cost $\lambda_1$ as defined in Step (I) above. Define updated dual values as follows:
    \[
    \balpha_c^{(1)} := \frac{\lambda_2}{\lambda_1}\cdot\alpha_c^{(1)}
    \quad\text{and}\quad
    \bbeta_c^{(1)} := \big[\balpha_c^{(1)} - \cost(c,f)\big]^+.
    \]
    This dual solution is feasible for LP \eqref{eq:fl-dual} with opening cost $\lambda_2$ and we have
    \[
        \sum_{c\in\C}\cost\big(c,f_c^{(1)}\big) \leq \frac{\lambda_1}{\lambda_2}\cdot\Lambda \cdot
        \left(\sum_{c\in \C} \balpha_c^{(1)} - \lambda_2 \cdot k_1\right)
    \]
\end{lemma}
\begin{proof}
    To prove the feasibility w.r.t.\ \eqref{eq:fl-dual}, we have to show that for all $f\in \F$, $\sum_{c\in \C}\bbeta_{cf}^{(1)}\leq \lambda_2$ and that for all $c\in \C$ and all $f\in \F$, $\balpha_c^{(1)}-\bbeta_{cf}^{(1)}\leq\cost(c,f)$. The second part follows directly from the definition of $\bbeta_{cf}^{(1)}$. For the first part, for all $c\in \C$ and $f\in \F$, we have
    \[
    \bbeta_{cf}^{(1)} = \big[\balpha_c^{(1)} - \cost(c,f)\big]^+
    = \bigg[\frac{\lambda_2}{\lambda_1}\cdot\alpha_c^{(1)} - \cost(c,f)\bigg]^+
    \leq \bigg[\frac{\lambda_2}{\lambda_1}\cdot\alpha_c^{(1)} - \frac{\lambda_2}{\lambda_1}\cdot\cost(c,f)\bigg]^+
    = \frac{\lambda_2}{\lambda_1}\cdot\beta_{cf}^{(1)}.
    \]
    The inequality follows because $\lambda_2\leq \lambda_1$. Feasibility of the dual solution now follows because we already know that $\sum_{c\in \C}\beta_{cf}^{(1)}\leq \lambda_1$

    Let us therefore now focus on showing that primal solution is a $\frac{\lambda_1}{\lambda_2}\cdot\Lambda$-approximate LMP solution for the facility location problem with opening cost $\lambda_2$. From \eqref{eq:kmeans_LMP}, we have
    \begin{eqnarray*}
        \sum_{c\in\C}\cost\big(c,f_c^{(1)}\big)  & \leq &
        \Lambda \cdot
        \left(\sum_{c\in \C} \alpha_c^{(1)} - \lambda_1 \cdot k_1\right)\\
        & = &
        \frac{\lambda_1}{\lambda_2}\cdot\Lambda \cdot
        \left( \frac{\lambda_2}{\lambda_1}\cdot\sum_{c\in \C} \alpha_c^{(1)}  -  \frac{\lambda_2}{\lambda_1}\cdot\lambda_1 \cdot k_1\right)\\
        & = &
        \frac{\lambda_1}{\lambda_2}\cdot\Lambda \cdot
        \left(\sum_{c\in \C} \balpha_c^{(1)} - \lambda_2 \cdot k_1\right),
    \end{eqnarray*}
    which concludes the proof.
\end{proof}

We now switch our attention to the $k$-means problem. We begin by formulating the primal integer linear program (ILP) and its corresponding dual for the $k$-means problem. For convenience, we define $k$-means as a variant of the facility location problem in which facilities and clients are given, there is a cost for opening facilities, and we are required to open at most $k$ facilities. This formulation is equivalent to the standard $k$-means problem when $\F$ and $\C$ are identical.

\paragraph{Primal {\it k}-Means Integer LP:}
\begin{align}\label{eq:primalKMeansLP}
\begin{split}
\min & \quad  \sum_{f \in \F} \sum_{c \in \C} \cost(c,f)\cdot  x_{cf} \\ 
\text{s.t.} & \quad \sum_{f \in \F} x_{cf} \ge 1, \quad \forall c \in \C, \\
&\quad x_{cf} \leq y_f, \quad \forall f \in \F, \forall c \in \C, \\
& \quad \sum_{f \in \F} y_f \le k,\\
&\quad x_{cf} \in \{0,1\}, \quad \forall f \in \F, \forall c \in \C, \\
& \quad y_f \in \{0,1\}, \quad \forall f \in \F.
\end{split}
\end{align}

\paragraph{Dual {\it k}-Means LP:}
\begin{align}\label{eq:dualKMeansLP}
\begin{split}
\max & \quad \sum_{c \in \C} \alpha_c - \lambda\cdot k\\ 
\text{s.t.} &\quad \alpha_c - \beta_{cf} \leq \cost(c,f), \quad \forall f \in \F, \forall c \in \C, \\ 
&  \quad \sum_{c \in \C} \beta_{cf} \leq \lambda, \quad \forall f \in \F, \\
&  \quad \alpha_c \geq 0, \quad \forall c \in \C, \\
& \quad \beta_{cf} \geq 0, \quad \forall f \in \F, \forall c \in \C.
\end{split}
\end{align}
Note that, unlike in the dual LP \eqref{eq:fl-dual} for facility location, in \eqref{eq:dualKMeansLP}, the variable $\lambda$ is part of the solution. 

By using the two facility location solutions with opening costs $\lambda_1$ and $\lambda_2$ that are computed in Step (I) of the algorithm of \Cref{sec:kmeans-highlevel-alg} and by using the adapted solution constructed in \Cref{lemma:FL_openingcostreduction}, we now construct two feasible solutions to \eqref{eq:primalKMeansLP} and \eqref{eq:dualKMeansLP} with input values $k_1$ and $k_2$. We will refer to the two solutions as solution $1$ and solution $2$ in the following

For the primal LP \eqref{eq:primalKMeansLP} with $k_i$ centers, where $i\in \set{1,2}$, the variables of solution $i$ are set as follows.
For every $c \in \C$ and $f \in \F$, define $x_{cf}^{(i)} := 1$ if $f = f_c^{(i)}$, and $x_{cf}^{(i)} := 0$ otherwise. For each $f \in \F$, we set $y_f^{(i)} := 1$ if $f \in \F_2$, and $y_f^{(i)} := 0$ otherwise.

For the dual LP \eqref{eq:dualKMeansLP}, for solution $1$, we use the modified dual solution that is constructed in \Cref{lemma:FL_openingcostreduction}. Hence, the variable $\lambda$ is set to $\lambda_2$ for solution $1$ and $2$. Further, for every $c\in \C$, variable $\alpha_c$ is set to $\balpha_c^{(1)}$ for solution $1$ and to $\alpha_c^{(1)}$ for solution $2$. Further, for all $c \in \C$ and $f \in \F$, we use the variables $\bbeta_{cf}^{(1)}$ for solution $1$ and $\beta_{cf}^{(2)}$ for solution $2$. Feasibility in both cases follows directly from the fact that the facility location solutions are feasible for \eqref{eq:fl-primal} and \eqref{eq:fl-dual}.

Moreover, Inequality~\eqref{eq:kmeans_LMP} directly implies that the objective value of the primal $k_2$-means solution $2$ is at most $\Lambda$ times the objective value of the corresponding dual solution. Similarly, \Cref{lemma:FL_openingcostreduction} implies that the primal $k_1$-means solution $1$ is at most $\frac{\lambda_1}{\lambda_2}\cdot\Lambda\leq2\Lambda$ times the objective value of the corresponding dual solution. Therefore, if $k_1 = k$ or $k_2 = k$, the respective facility location solution directly yields a $2\Lambda$-approximate solution for the $k$-means problem.

The next lemma shows that if neither $k_1$ nor $k_2$ equals $k$, it is possible to interpolate between the two solutions to obtain a fractional $\frac{\lambda_1}{\lambda_2}\cdot\Lambda$-approximate solution for the $k$-means problem, that is, a feasible solution to the fractional relaxation of \eqref{eq:primalKMeansLP}.

We define two real numbers $a,b\geq 0$ as in Step (III) of the algorithm of \Cref{sec:kmeans-highlevel-alg}, i.e., $a := \frac{k_2-k}{k_2-k_1}$ and $b := \frac{k-k_1}{k_2-k_1}$. Note that $a + b = 1$ and $a\cdot k_1 + b\cdot k_2 = k$. We define fractional solutions with variables $x_{cf}$ and $y_f$ for \eqref{eq:primalKMeansLP}, and with variables $\alpha_c$, $\beta_{cf}$, and $\lambda$ for \eqref{eq:dualKMeansLP} as follows.
\begin{eqnarray*}
    \forall c\in \C, \forall f\in \F 
    & : & 
    x_{cf} = a\cdot x_{cf}^{(1)} + b\cdot x_{cf}^{(2)},\quad
    y_f = a\cdot y_f^{(1)} + b\cdot y_f^{(2)}\\
    \forall c\in \C, \forall f\in \F
    & : &
    \alpha_c = a\cdot \balpha_c^{(1)} + b\cdot \alpha_c^{(2)},\quad
    \beta_{cf} = \big[\alpha_c -\cost(c,f)\big]^+
\end{eqnarray*}
The value of $\lambda$ is set to $\lambda_2$. The following lemma shows that this convex combination of the two solutions provides a $\frac{\lambda_1}{\lambda_2}\cdot\Lambda$-approximate solution for the fractional LP relaxation of \eqref{eq:primalKMeansLP}.

\begin{lemma}\label{lemma:fractionalkmeanssolution}
    The above fractional solutions are feasible for the fractional LP relaxation of \eqref{eq:primalKMeansLP} and for the LP \eqref{eq:dualKMeansLP}. Moreover, we have
    \begin{equation}\label{eq:kmeans_fractional_lemma}
        \sum_{c\in\C, f\in \F} \!\!\!\!\!\cost(c,f)\!\cdot\! x_{cf} =
        \sum_{c\in \C}\!\left(a\!\cdot\! \cost\big(c,f_c^{(1)}\big) + b\!\cdot\! \cost\big(c,f_c^{(2)}\big)\right) \leq
        \frac{\lambda_1}{\lambda_2}\cdot\Lambda\cdot\!\left(\sum_{c\in\C}\alpha_c - \lambda_1\!\cdot\! k\right)\!.
    \end{equation}
    Note that this implies that the primal solution is optimal up to a factor  at most $\frac{\lambda_1}{\lambda_2}\cdot\Lambda$.
\end{lemma}

\begin{proof}
    We first prove the feasibility of the solution of \eqref{eq:primalKMeansLP}. The inequalities of the form $\sum_{f\in \F}x_{cf}\geq 1$ and $x_{cf}\leq y_f$ directly follow because the variables $x_{cf}$ and $y_f$ result from a convex combination of two solutions that satisfy these inequalities. The inequality $\sum_{f\in \F}y_f\leq k$ follows because we have $\sum_{f\in \F} y_f^{(i)}\leq k_i$ for $i\in \set{1,2}$ and because $y_f=a\cdot y_f^{(1)}+b\cdot y_f^{(2)}$ and $k=a\cdot k_1 + b\cdot k_2$.

    Let us also prove the feasibility of the solution of \eqref{eq:dualKMeansLP}. The inequalities $\alpha_c-\beta_{cf}\leq \cost(c,f)$ follow directly from the choice of $\beta_{cf}$. The inequality $\sum_{c\in\C}\beta_{cf}\leq \lambda=\lambda_2$ follows because
    \begin{eqnarray*}
        \sum_{c\in\C} \beta_{cf} & = &
        \sum_{c\in\C}\big[\alpha_c-\cost(c,f)\big]^+\\
        & = &
        \sum_{c\in\C}\big[a\cdot \balpha_c^{(1)} + b\cdot \alpha_c^{(2)}-\cost(c,f)\big]^+\\
        & \stackrel{a+b=1}{=} &
        \sum_{c\in\C}\big[a\cdot \balpha_c^{(1)} - a\cdot\cost(c,f) + b\cdot \alpha_c^{(2)}-b\cdot\cost(c,f)\big]^+\\
        & \leq &
        \sum_{c\in\C}\left(\big[a\cdot \balpha_c^{(1)} - a\cdot\cost(c,f)\big]^+ + \big[b\cdot \alpha_c^{(2)}-b\cdot\cost(c,f)\big]^+\right)\\
        & = &
        \sum_{c\in\C}\left(a\cdot\big[\balpha_c^{(1)} - \cost(c,f)\big]^+ + b\cdot\big[\alpha_c^{(2)}-\cost(c,f)\big]^+\right)\\
        & \leq & \lambda_2
    \end{eqnarray*}
    The first inequality holds because for all $A,B\in\R$, we have $[A+B]^+\leq[A]^++[B]^+$. The last inequality follows from $a+b=1$ and from $\sum_{c\in\C}\bbeta_{cf}^{(1)}\leq\lambda_2$ and $\sum_{c\in\C}\beta_{cf}^{(2)}\leq\lambda_2$.

    It remains to prove \eqref{eq:kmeans_fractional_lemma}. This however follows because $a+b=1$, $k = a \cdot k_1 + b \cdot k_2$ and the inequality of \eqref{eq:kmeans_fractional_lemma} therefore just is a convex combination of \eqref{eq:kmeans_LMP} (for $i=2$) and of the inequality proven in \Cref{lemma:FL_openingcostreduction}. This concludes the proof.
\end{proof}

\begin{theorem}\label{thm:kmeansapprox}
    The $k$-means algorithm described in \Cref{sec:kmeans-highlevel-alg} returns a $(2+4\delta^2)\cdot\Lambda$-approximate solution in expectation for any given $k$-means instance.
\end{theorem}

\begin{proof}
    Consider the connection cost of some client $c\in\C$. In the algorithm of \Cref{sec:kmeans-highlevel-alg}, $c$ is connected either to $f_c^{(1)}$, to $f_c^{(2)}$, or to $f_2'(f_c^{(1)})$. Let us distinguish two cases.

    \paragraph{Case 1: \boldmath$f_c^{(2)}\in \F_2'$:} In this case, $c$ is connected to $f_c^{(1)}$ (i.e., $\phi(c)=f_c^{(1)}$) if and only if the set $\F_1$ is sampled in Step (A) of Step (III) of the algorithm and it is connected to $f_c^{(2)}$ otherwise. The set $\F_1$ is sampled with probability exactly $a$. Client $c$ is therefore connected to $f_c^{(1)}$ with probability $a$ and to $f_c^{(2)}$ with probability $b=1-a$. The expected connection cost of $c$ is therefore
    \begin{equation*}
    \mathbb{E}[\cost(c, \phi(c))] = a \cdot \cost\big(c,f_c^{(1)}\big) + b \cdot \cost\big(c,f_c^{(2)}\big).
    \end{equation*}
    
    \paragraph{Case 2: \boldmath$f_c^{(2)}\in \F_2\setminus\F_2'$:}
    In this case, we have $\phi(c)=f_c^{(2)}$ if and only if $f_c^{(2)}$ is sampled as one of the $k-k_1$ random additional centers in $\F_2\setminus\F_2'$ in Step (B). The probability for this is $(k-k_1)/|\F_2\setminus\F_2'|=(k-k_1)/(k_2-k_1)=b$. We have $\phi(c)=f_c^{(1)}$ if set $\F_1$ is sampled in Step (A) and $f_c^{(2)}$ is not sampled in Step (B). This happens with probability $a(1-b)=a^2$. In the remaining cases, i.e., with probability $1-a^2-b=ab$, $c$ is connected to $f_2'(f_c^{(1)})$. In this case, the connection cost of $c$ can be bounded as follows. By the definition of $f_2'(f_c^{(c)})$, i.e., by Eq.~\eqref{eq:kmeans_rounding_facility}, we have
    \[  
        \dist(f_c^{(1)}, f_2'(f_c^{(1)})) \le \delta \cdot \dist(f_c^{(1)}, f_c^{(2)}),
    \]
    where $\delta\geq 1$ is some constant that is given by Step (II) of the algorithm of \Cref{sec:kmeans-highlevel-alg}. We now have
    \begin{eqnarray*}
        \dist\big(c, f'_2(f_c^{(1)})\big) 
        & \leq &
        \dist\big(c, f_c^{(1)}\big) + \dist\big(f_c^{(1)}, f_2'(f_c^{(1)})\big)\\
        & \leq & \cost(c, f_c^{(1)}) + \delta\cdot \cost(f_c^{(1)}, f_c^{(2)})\\
        &\leq & \cost(c, f_c^{(1)}) + \delta\cdot\left(
        \cost(f_c^{(1)}, c) + \cost(c, f_c^{(2)})
        \right)\\
        & \leq & (\delta+1)\cdot \left(\dist(c, f_c^{(1)}) + \dist(c, f_c^{(2)})\right).
    \end{eqnarray*}
    By squaring on both sides, we obtain
    \[
    \cost\big(c, f'_2(f_c^{(1)})\big) \leq
    2(\delta+1)^2\cdot \left(\cost(c, f_c^{(1)}) + \cost(c, f_c^{(2)})\right).
    \]
    The expected connection cost in Case 2 can therefore be upper bounded as
    \begin{eqnarray*}
        \mathbb{E}[\cost(c, \phi(c))] & = &
        a^2\cdot \cost\big(c, f_c^{(1)}\big) + 
        b\cdot \cost\big(c, f_c^{(2)}\big) + 
        ab\cdot \cost\big(c, f_2'(f_c^{(1)})\big)\\
        & \leq &
        a^2\cdot \cost\big(c, f_c^{(1)}\big) + 
        b\cdot \cost\big(c, f_c^{(2)}\big) + 
        2(\delta+1)^2\cdot \left(\cost(c, f_c^{(1)}) + \cost(c, f_c^{(2)})\right)\\
        & = &
        \left[a + 2(\delta+1)^2\cdot b\right]
        \cdot a\cdot \cost\big(c, f_c^{(1)}\big) +
        \left[1 + 2(\delta+1)^2\cdot a\right]
        \cdot b\cdot \cost\big(c, f_c^{(2)}\big)\\
        & \leq &
        \left[1 + 2(\delta+1)^2\right]\cdot \left(
        a \cdot \cost\big(c,f_c^{(1)}\big) + b \cdot \cost\big(c,f_c^{(2)}\big)
        \right).
    \end{eqnarray*}
    The last inequality follows from $a,b\leq 1$.
    The upper bound on the expected connection cost in Case 2 clearly also holds in Case 1. Together with \Cref{lemma:fractionalkmeanssolution}, this implies that the expected connection cost is within a factor at most $(1+2(\delta+1)^2)\cdot\frac{\lambda_1}{\lambda_2}\cdot\Lambda \leq (2+4(\delta+1)^2)\cdot\Lambda$ of the objective value of some feasible dual solution to the $k$-means LP. This concludes the proof.
\end{proof}

\section{Implementation in the MPC Model}
\label{sec:mpc}

In this section, we discuss how to implement our high level facility location algorithm of \Cref{alg:fl-overview} and the randomized rounding algorithm to obtain an approximate $k$-means solution in a fully scalable way in the MPC model. As the technically demanding part is the implementation of the facility location algorithm, we first concentrate on the implementation of \Cref{alg:fl-overview} and only discuss the $k$-means algorithm at the very end.

As discussed in \Cref{sec:intro}, by using locality-sensitive hashing techniques, one can approximate the distances in $\mathbb{R}^d$ by a sparse graph with $n^{1+\eps}$ edges for some constant $\eps>0$. In the following, in \Cref{sec:mpc_graphapprox}, we first provide the details of this approximation. In \Cref{sec:mpc_graphalg}, we then describe the necessary MPC graph algorithms to implement our facility location and $k$-means algorithm in the MPC model. Finally, in \Cref{sec:MPC_impl}, we then show how to put the pieces together and thus how to use the graph approximation of $\mathbb{R}^d$ described in \Cref{sec:mpc_graphapprox} together with the algorithm of \Cref{sec:mpc_graphalg} to implement \Cref{alg:fl-overview}.

\subsection{Approximating the Geometric Space by a Graph}
\label{sec:mpc_graphapprox} 

In the following, we consider a set of points $P\subset \mathbb{R}^d$ of size $|P|=n$ equipped with the $\ell_2$-metric. For $x,y\in P$, we use $\dist(x,y)$ to denote the $\ell_2$-distance between $x$ and $y$. We assume that for any $x,y\in P$, if $x\neq y$, then $\dist(x,y)\geq 1$ and $\dist(x,y)\leq \poly(n)$. Our goal is to obtain a constant approximation of the metric space $(P,\dist)$ by (weighted) shortest path metric of a relatively sparse graph. This can be achieved by using locality-sensitive hashing.

\begin{definition}[Locality-Sensitive Hashing (LSH)]\label{def:lsh}
Consider a family $\H$ of hash functions mapping $\mathbb{R}^d$ into some universe $\U$. For $D, \Gamma>1$ and $p_1,p_2\in [0,1]$, $\mathcal{H}$ is $(D,\Gamma\cdot D,p_1,p_2)$-sensitive, if $\forall x,y\in \mathbb{R}^d$, it satisfies:
\begin{enumerate}
    \item  If $\dist(x,y)\leq D$, $\Pr_{h\in\mathcal{H}}[h(x)=h(y)]\geq p_1$. 
    \item If $\dist(x,y)\geq \Gamma\cdot D,\Pr_{h\in\mathcal{H}}[h(x)=h(y)]\leq p_2$.
\end{enumerate}
\end{definition}

\noindent We next describe locality-sensitive hashing that can be used to compute a graph approximation of $(P,\ell_2)$. We first give an algorithm that only takes care of point pairs $(x,y)\in P^2$ for which $\dist(x,y)\approx D$.

\paragraph{LSH-Based Spanner Algorithm for a given 
distance (a.k.a. scale) \boldmath$D$}
\begin{enumerate}
    \item \textbf{Input}: $n$ points $P \subset \mathbb{R}^d$, $\Gamma > 1$, $D > 0$.
    \item Choose a $(D, \Gamma\cdot D, p_1, p_2)$-sensitive family $\H$ and draw independent $h_1, \ldots , h_t$ from $H$ for $t = 5 \ln (n) / p_1$.
    \item Create a graph $G_D = (P, E_D)$ where each point $x \in P$ denotes a vertex in the graph.
    \item For each hash function $h_i$, $i \in [t]$ and each $u\in \U$ for which there is some $x\in P$ for which $h_i(x)=u$, let $x_u\in P$ be an arbitrary point for which $h_i(x_u)=p$ and add an edge $\set{x_u,y}$ to $E_D$ between $x$ and every other $y\in P$ for which $h_i(y)=h_i(x_u)=u$.
    \item Return $G_D$
\end{enumerate}

\noindent We have the following lemma.
\begin{lemma}\label{lem:graph-construct}
    Consider a set $P\subset\mathbb{R}^d$ of $n$ points and $\Gamma>1,D>0$.
Let $G_D=(P,E_D)$ be the output of the above algorithm. Then the number of edges is $|E_D|\leq \frac{5\ln n}{p_1}$ and 
with probability at least $1-\max\set{1/n^3,5p_2p_1^{-1}n^2\ln n}$, the following holds:
(I) $\forall x,y\in P$ with $\dist(x,y)\leq D$, either there is an edge between $x$ and $y$ in $G$ or $x$ and $y$ have a common neighbor in $G$.
(II) If there is an edge between $x,y\in P$ in $G$, $\dist(x,y)< \Gamma\cdot D$.
\end{lemma}
\begin{proof}
    The bound on $|E_D|$ follows because the number of edges added per hash function $h_i$ is at most $n-1$ and there are $5\ln(n)/p_1$ hash functions $h_i$.

    For any $x,y\in P$, there is a path of length $\leq 2$ in the graph $G_D$ if there exists an $i\in \set{1,\dots,t}$ for which $h_i(x)=h_i(y)$. By the assumption that $\H$ is a $(D, \Gamma\cdot D, p_1, p_2)$-sensitive family of hash functions, for every $x,y\in P$ with $\dist(x,y)\leq D$ and every $i$, we have $\Pr(h_i(x)=h_i(y))\geq p_1$. The probability that $x$ and $y$ are not connected by a path of length $\leq 2$ is therefore at most
    \[
    1- (1-p_1)^t \geq 1 - e^{-p_1\cdot t} =
    1 - e^{-p_1\cdot \frac{5\ln n}{p_1}} = \frac{1}{n^5}.
    \]
    By a union bound over the at most $n^2/2$ pairs of nodes $x,y$ with $\dist(x,y)\leq D$, the probability that each such pair is connected by a path of length at most $2$ and thus (I) holds is therefore at least $1-\frac{1}{2n^3}$.

    In order for $x,y\in P$ to be connected by an edge, we must have $h_i(x)=h_i(y)$ for some $i\in \set{1,\dots,t}$. If $\dist(x,y)\geq \Gamma\cdot D$, the probability that this happens for a specific $i$ is at most $p_2$. Taking a union bound over all $t=5\ln(n)/p_1$ hash functions $h_i$ and all at most $n^2/2$ pairs of nodes $x,y$, we can upper bound the probability that (II) does not hold as $p_2\cdot\frac{5\ln n}{p_1}\cdot \frac{n^2}{2}$.
\end{proof}

\begin{lemma} \label{lem:opt_lsh} \cite{andoni2006near,Cohen-AddadMZ22}
For any ``scale'' $D>0$, dimension $d>0$, $c>0$ and $s\in[1,\log n]$, there is a $(D,\Gamma\cdot D,1/n^{s/\Gamma^2+o(1)},1/n^s)$-sensitive family $\mathcal{H}$ of hashing functions for $\R^d$.
In addition, for any set of $n$ points $P\subset \mathbb{R}^d$ and any $h\in\mathcal{H}$, there is a fully scalable MPC algorithm with $O(1)$ rounds and $n^{1+o(1)}d$ total space that computes $h(x)$ for every $x\in P$.
The space to store each $h(x)$ is $O(\log^3 n)$.
\end{lemma}

\noindent We now have everything that we need to describe the graph approximation of $(\mathbb{R}^d, \dist)$ that we use.

\begin{lemma}\label{lemma:graphapprox}
    Let $P$ be a set of $|P|=n$ points in $\mathbb{R}^d$ and let $\eps>0$ be an arbitrary constant. There is a weighted graph $G=(P, E)$ with node set $P$ and $|E|\leq n^{1+\eps}$ edges and positive edge weights $w(e)\ge 1$ such that for some constant $\Gamma$, $G$ has the following properties
    \begin{enumerate}
        \item[(I)] For any two nodes $x,y\in P$, either $G$ contains an edge $e=\set{x,y}$ of length $w(e)\leq \dist(x,y)/2$ or there exists a point $z\in P$ such that $e_1=\set{x,z}\in E$, $e_2=\set{z,y}\in E$, and $w(e_1)+w(e_2)\leq \dist(x,y)$.
        \item[(II)] For any two nodes $x,y\in P$, $\dist(x,y)\leq \Gamma\cdot d_G(x,y)$, where $d_G(x,y)$ denotes the weighted shortest path distance in $G$.
    \end{enumerate}
    Further, there is a randomized MPC algorithm to compute $G$ in $O(1)$ rounds and with a global memory of size $O(n^{1+\eps})$.
\end{lemma}
\begin{proof}
	We normalize the distances such that $\dist(x, y) \in [1, O(n^6)]$ for all $x, y \in P$ (see \Cref{sec:prelim}). Let $D_i = 2^i$ for $i \in [0, O(\log n)]$ and define $G_i$ to be the graph constructed using \Cref{lem:graph-construct} with parameter $D = D_i$.
	
	In $G_i$, for every pair $x, y \in P$ with $\dist(x,y) \le D_i$, either $\{x, y\} \in E_i$ or there exists $z \in P$ such that $\{x,z\}, \{z,y\} \in E_i$. For all edges $\{x, y\} \in E_i$, assign weight $w(x, y) = D_i/4$. We take the union of all $G_i$ over all $i$ to form $G$ and for each edge, retain the minimum weight among those assigned in different $G_i$.
	
	Suppose $\dist(x,y) \in (D_{i-1}, D_i]$. Then $w(x,y) = D_i/4 < \dist(x,y)$, and for pairs connected through an intermediate point $z$, we have $w(x,z) + w(z,y) \le D_i/2 < \dist(x,y)$. This establishes Property (I).
	
	For Property (II), note that any path of edges in $G$ consists of edges of length at most $D_i$, so the total length over any shortest path is at most $\Gamma \cdot \dist(x,y)$ for some constant $\Gamma$.
	
	To analyze the complexity, observe that for each scale $D_i = 2^i$, we construct a graph $G_i$ using \Cref{lem:graph-construct}. Each such construction requires a family of $(D_i, \Gamma \cdot D_i, p_1, p_2)$-sensitive hash functions, which, by \Cref{lem:opt_lsh}, can be computed for all $x \in P$ in $O(1)$ MPC rounds using $n^{1+o(1)}d$ total space. Since we have $O(\log n)$ different scales $D_i$, we perform $O(\log n)$ such graph constructions in parallel, so the total space for all hash functions is $\tilde{O}(n^{1+o(1)})$. Each $G_i$ contains at most $p_1^{-1} \cdot 5 \log n$ edges, and taking the union across all $i$ yields at most $p_1^{-1} \cdot 5 \log^2 n$ total edges.
	
	By choosing parameters so that $p_1 = 1/n^{s/\Gamma^2+o(1)}$, we ensure that the total number of edges in $G$ is at most $n^{s/\Gamma^2 + o(1)} \log^2 n$. Choosing $\Gamma = O(1/\sqrt{\epsilon})$ implies:
	\[
	|E| \le \log^2 n \cdot n^{s \cdot \epsilon + o(1)}.
	\]
	To ensure $|E| \le n^{1 + \epsilon}$, we select
	\[
	s = \frac{1 + \epsilon}{\epsilon} - o(1).
	\]
	Since $\log^2 n = n^{o(1)}$, it follows that the number of edges is bounded by $n^{1+\epsilon}$ as required.
	
	Finally, since each edge weight is derived from a constant number of scales and updating the minimum weights can be done in parallel, the final graph $G$ can be computed in $O(1)$ MPC rounds using total space $\tilde{O}(n^{1+\epsilon})$.
\end{proof}

\subsection{MPC Graph Algorithms}
\label{sec:mpc_graphalg}

Since we can approximate the underlying Euclidean metric by a graph, some of our algorithms become standard MPC graph algorithms. We therefore next discuss some MPC algorithms that directly work on a given undirected graph $G=(V,E)$. When designing those MPC algorithms, we assume that each machine has a local memory of $n^\delta$ words, where $n=|V|$ is the number of nodes and $\delta>0$ is an arbitrary constant (and where each word typically consists of $\Theta(\log n)$ bits). Note that in the context of graph algorithms, this setting is known as the sublinear-memory MPC regime. If not specified explicitly, we will in the following assume that an MPC algorithm is allowed to use a total memory of $O(m+n)$ bits, where $m=|E|$ is the number of edges of $G$. That is, we allow a total space that is by a $\poly(\log n)$ factor larger than what is needed to store $G$.

We assume that $G$ is given as a list of nodes and edges. To be able to distinguish different nodes, each node has a unique $O(\log n)$-bit name, which we will refer to as the ID of the node. Since sorting can be done in constant time and with linear global memory in the MPC model~\cite{goodrich2011sorting}, we can assume that the edges of $G$ are sorted by one of their nodes. If for each edge $e=\set{u,v}\in E$, we store a copy of $e$ for $u$ and a copy of $e$ for $v$, we can assume that the edges of each node are stored consecutively. If the degree of a node is at most $n^\delta$, all its edges can be stored on a single machine. Otherwise, the edges of a node will be stored on consecutive machines. In this way, operations such as computing a sum or a minimum overall neighbors for each node are straightforward to implement in $O(1)$ time. For our algorithm, we will however also need to implement such basic operations and run some algorithms on the graph $G^t$ for some $t=O(1)$, where $G^t$ is the graph on nodes $V$ with an edge between any two nodes that are in hop distance $\leq t$ in $G$. Note that already for $t=2$, $G^t$ can be a much denser graph than $G$ and we can therefore not store $G^t$ explicitly without increasing the global memory too substantially. We therefore need to implement algorithms on $G^t$ as algorithms that operate on $G$. The next two lemmas provide some basic operations on $G^t$ that can be implemented efficiently by using standard techniques. In the following, we use $N_t(v):=\set{u \in V : d_G(u,v)\leq t}$ to denote the $t$-hop neighborhood of $v$ in $G$ (for convenience, including the node $v$ itself).

For convenience, we introduce the following notation to denote the $\ell$ smallest elements of a set $A\subseteq \mathbb{X}$ of some totally ordered domain $\mathbb{X}$.

\begin{equation}\label{eq:ksmallest}
    \minl(A) := A' \subseteq A \text{ such that }
    |A'| = \min\set{\ell, |A|} \land \forall (a,b)\in A' \times (A\setminus A') : a<b.
\end{equation}

\begin{lemma}\label{lemma:basicMins}
    Let $G=(V,E)$ be an $n$-node graph and let $\ell\geq 1$ be an integer parameter such that $\ell\leq n^{\delta/2}$ (where $n^\delta$ is the memory of a single machine). Assume that as input, each node $u\in V$ has a set $A_u\subseteq \mathbb{X}$ of size $|A_u|\leq \ell$, where $\mathbb{X}$ is a totally ordered set of size $|\mathbb{X}|$ at most $\poly(n)$. Let $t\geq 1$ be an integer. As output, every node $u$ must output the set $B_u=\minl\big(\bigcup_{v\in N^t(u)} A_v\big)$. If one can use $O((m + n)\cdot \ell)$ bits of global memory, the sets $B_u$ can be computed in $O(t)$ rounds in the MPC model.
\end{lemma}
\begin{proof}
    First, note that it is sufficient to show that the problem can be solved for $t=1$ in $O(1)$ rounds. To see this, note that for every $u\in V$ and every $t> 1$, we have
    \[
    \minl\left(
    \bigcup_{v\in N^{1}(u)}
        \minl\left(
        \bigcup_{w\in N^{t-1}(v)} A_w
        \right)
    \right)\ =\
    \minl\left(
        \bigcup_{v\in N^{t}(u)} A_v
    \right)
    \]
    We can therefore break the task of finding the $\ell$ smallest values in the $t$-hop neighborhood into $t$ times consecutively finding the $\ell$ smallest values in the $1$-hop neighborhood. For the remainder of the proof, we therefore restrict ourselves to the case $t=1$.

    Because we have $O((m + n)\cdot \ell)$ bits of global memory, we can explicitly store every edge $\set{u,v}$ together with the input sets $A_u$ and $A_v$ of its two nodes. As usual, we store each edge twice, once for each of its nodes and we sort the edges by the nodes so that for each node, all its edges are stored consecutively. We can do this so that a) if all edges $\set{u,v}$ of a node $u$ (including the sets $A_u$ and $A_v$) fit on one machine, they are all stored on one machine and b) otherwise, we have a set of machines that only store the edges of node $u$. For the nodes $u$ for which all edges fit on one machine, we can compute $\minl\big(\bigcup_{v\in N^1(u)} A_v\big)$ without communication. Consider a node $u$ for which the edges $\set{u,v}$ do not fit on one machine and let $\M_u$ be the set of machines that store $u$'s edges $\set{u,v}$ together with the sets $A_u$ and $A_v$. We connect the machines into a $n^{\delta/2}$-ary aggregation tree (of height at most $O(1/\delta)=O(1)$). We can aggregate the smallest $v$ values among the neighbors of $v$ on this aggregation tree in $O(1)$ time. To see that we have enough bandwidth, note that each machine on this tree needs to receive at most $\ell$ values from each of its $n^{\delta/2}$ children in the tree. Because $\ell\leq n^{\delta/2}$, the total number of values a machine has to receive in a single round is $n^\delta$. This concludes the proof. 
\end{proof}

\begin{lemma}\label{lemma:basicSums}
    Let $G=(V,E)$ be a graph, let $\eps\geq 1/\poly\log n$, let $t\geq 0$ be a non-negative integer, and assume that every node $u\in V$ is given an input value $x_u\geq 0$. For each node $u\in V$, we define $X_u^{(t)}:=\sum_{u\in N^t(v)} x_v$. There is an $O(t)$-round sublinear-memory MPC algorithm to compute a value $\tilde{X}_u^{(t)}$ such that $X_u^{(t)}\leq \tilde{X}_u^{(t)}\leq (1+\eps)X_u^{(t)}$. The algorithm is randomized and succeeds with high probability.
\end{lemma}
\begin{proof}
    To approximate the sum $X_u^{(t)}$ of the values $x_v$, we can use the technique described in \cite{Mosk-AoyamaS08}, where it is shown that the computation of a $(1+\eps)$-approximation of the sum of $n$ values can be reduced to $O(\log(n)/\eps^2)$ (independent) instances of computing the minimum of $n$ values. We can apply the algorithm of \Cref{lemma:basicMins} for $\ell=1$ to determine the minimum value in the $t$-hop neighborhood $O(\log(n)/\eps^2)$ times in parallel to compute $\tilde{X}_v^{(t)}$.
\end{proof}

As a key step of our facility location algorithm, we have to find a set of center nodes in the dependency graph of paid (or approximately paid) facilities so that two nodes in the set are sufficiently separated and ideally, all nodes are within a constant distance of those center nodes. Such sets are known as ruling sets~\cite{awerbuch89}. Formally, an $(a,b)$-ruling set of a graph $G=(V,E)$ is a subset $S\subseteq V$ of the nodes so that for all $u,v\in S$, $d_G(u,v)\geq a$ and for all $u\not\in S$, there exists $v\in S$ such that $d_G(u,v)\leq b$. There is quite extensive literature on distributed and also MPC algorithms for computing ruling sets. However, unfortunately, for the sublinear-memory regime, it is not known if an $(a,b)$-ruling set for any $a\geq 2$ and $b=O(1)$ can be computed in time $\poly(\log\log n)$ in the distributed or in the sublinear memory MPC setting. In fact, it is not even known if such an algorithm exists for $b=O(\log\log n)$ if $a> 2$. As a first step, we show how to compute an $(a,b)$-ruling set for $a>1$ and $b=O(\log\log\log n)$ by using a bit of additional global memory and an algorithm of \cite{KothapalliPP20}.

\begin{lemma}\label{lemma:logloglogRS}
    Let $G=(V,E)$ be an $n$-node graph with $m=|E|$ edges and let $U\subseteq V$. For every $t\geq 2$, there exists a sublinear-memory MPC algorithm that computes an $(2, O(\log\log\log n))$-ruling set $S\subseteq U$ of $G^t[U]$ in $O\big(t\cdot\log\log\log n + \log\log n\cdot\log\log\log n \big)$ rounds. The algorithm requires $O((m+n)\cdot n^\eps)$ bits of global memory, where $\eps>0$ is a constant that can be chosen arbitrarily small.
\end{lemma}
\begin{proof}
    We first describe the high-level idea of the algorithm. Note that a set $S\subseteq U$ is a \\ $(2, O(\log\log\log n))$-ruling set of $G^t[U]$ if and only if $S$ is an independent set of $G^t$ and for every node $u\in U\setminus S$, there exists a node $v\in S$ at distance $O(\log\log\log n)$ in $G^t$. The algorithm consists of $O(1/\eps)$ phases. At the beginning $S=\emptyset$ and the set of nodes that still needs to be covered is $U_0=U$. After phase $p$, we have already selected a set $S_p\subseteq U$ to be in $S$ and the set of uncovered nodes in $U$ is $U_p$ (i.e., $U_p$ are the nodes $u\in U$ for which there is no node $v\in S_p$ that is within distance $O(\log\log\log n)$ in $G^t$). The goal of each phase is to reduce the maximum degree of the remaining subgraph of $G^t$ by a factor $O(n^\eps)$. More precisely, we will compute $S_p\supseteq S_{p-1}$ such that $G^t[U_p]$ has maximum degree at most $n^{1-\eps p}$. The number of phases is therefore at most $O(1/\eps)=O(1)$. In the following, we describe how a single phase $p$ is implemented. At the same time, we also show by induction that for every $p\geq 0$, the maximum degree of $G^t[U_p]\leq n^{1-\eps p}$. This condition is clearly true for $p=0$.

    We assume that at the beginning of phase $p\geq 1$, we know the set $U_{p-1}$ of uncovered nodes. By using the induction hypothesis, we further assume that the maximum degree of $G^t[U_{p-1}]$ is at most $n^{1-(p-1)\eps}$. In phase $p$, we first independently mark every node in $U_{p-1}$ with probability $c\cdot \log(n) / n^{1-\eps p}$ for a sufficiently large constant $c>0$. Let $M_p$ be the set of those marked nodes in phase $p$. Note that in the graph $G^t[U_{p-1}]$, every node of degree at least $n^{1-\eps p}$ with high probability has at least one neighbor in $M_p$. Assume that we can compute a $(2,b)$-ruling set $S_p^+$ of $G^t[M_p]$ for some $b=O(\log\log\log n)$. By setting $S_p = S_{p-1} \cup S_p^+$ and including in $U_p$ all nodes from $U_{p-1}$ that are at a $G^t$-distance greater than $b+1 = O(\log\log\log n)$ from $S_p^+$, the maximum degree of $G^t[U_p]$ is reduced to at most $n^{1-\varepsilon p}$, as required. Once, $S_p^+$ is known, we can use \Cref{lemma:basicMins} with $\ell=1$ to compute $U_p$ in $O(t\log\log\log n)$ rounds. It remains to show that we can compute a $(2,O(\log\log\log n)$-ruling set of $G^t[M_p]$.

    We first show that the graph $G^t[M_p]$ can be explicitly constructed. Because the maximum degree of $G^t[U_{p-1}]$ is at most $n^{1-\eps(p-1)}$ and every node in $U_{p-1}$ is added to $M_p$ with probability $c\cdot \log(n) / n^{1-\eps p}$, w.h.p., for each node in $U_{p-1}$ there are at most $c'\cdot n^\eps\log n$ marked nodes in the $t$-hop neighborhood in $G$, where $c'>0$ is a sufficiently large constant. We choose $\eps$ such that $c'\cdot n^\eps\log n\leq n^{\delta/2}$. Then, by using \Cref{lemma:basicMins} with $\ell=O(n^\eps\log n)$, every node in $U_{p - 1}$ can obtain the IDs of all the marked nodes in its $t$-hop neighborhood in $G$ in time $O(t)$. We can therefore explicitly construct the graph $G^t[M_p]$. We can now use the algorithm of Kothapalli, Pai, and Pemmaraju~\cite{KothapalliPP20} to compute a $(2, O(\log\log\log n))$-ruling set of $G^t[M_p]$ in $O(\log\log n \cdot\log\log\log n)$ rounds. The global memory needed by this algorithm is $O(m' + n^{1+\eps})$, where $m'$ is the number of edges of $G^t[M_p]$.
\end{proof}

We are not aware of a sublinear-memory $\poly(\log\log(n))$-time MPC algorithm to compute a ruling set of $G^t$ that is better than what \Cref{lemma:logloglogRS}. Unfortunately, this does not suffice to obtain a constant approximation in our facility location algorithm (and thus in our $k$-means algorithm). For this, we would need a $(2,O(1))$-ruling set of $G^t$ for some sufficiently large constant $t$. Since such an algorithm is not known, we have to instead settle for an algorithm that computes centers that only cover most of the nodes within constant distance.

To this end, we can employ a variant of Luby's classic parallel MIS algorithm~\cite{alon86,luby86}. In Theorem 2 of \cite{BalliuGKO23}, it is shown that if we do one step of Luby's algorithm and if, after adding the nodes of the step to the independent set, instead of only the direct neighbors, we remove the $2$-hop neighborhoods of those nodes, then a constant fraction of the graph gets removed with constant probability.

\paragraph{Adapted Luby Step} Consider the following process on an undirected graph $H=(V_H, E_H)$. Assume that every node $v\in V_H$ has an estimate $\hat{d}(v)$ of its degree in $H$ such that $\deg_H(v)\leq \hat{d}(v)\leq 2\deg_H(v)$. Each node $v\in V_H$ is added to the set $S_H$ independently with probability $c\log n/\hat{d}(v)$ for a sufficiently large constant $c>0$.

\begin{lemma}\label{lemma:lubystep}
    When applying the above step on a graph $H$, with high probability, a constant fraction of the nodes in $V_H$ are within distance at most $2$ of $S_H$. Furthermore, for every node $u$, there are at most $O(\log n)$ neighbors $v \in S_H$ with $\hat{d}(v) \geq \hat{d}(u)$.
\end{lemma}
\begin{proof}
    We first prove that, with high probability, a constant fraction of the nodes in $V_H$ are within distance at most $2$ of $S_H$. To this end, we classify nodes as either \emph{good} or \emph{bad}. A node $v \in V_H$ is considered \emph{good} if the total marking probability over its inclusive 2-hop neighborhood, defined as $N_2^+(v) := \{ u \in V_H : d_H(u, v) \leq 2 \}$, satisfies $\sum_{u \in N_2^+(v)} p_u \geq 1/2$. Otherwise, $v$ is deemed \emph{bad}.

    Let $B \subseteq V_H$ denote the set of bad nodes. For each $v \in B$, we have by definition that $\sum_{u \in N_2^+(v)} p_u < 1/2$. To bound $|B|$, we distribute the \emph{badness} of each $v \in B$ to its immediate neighbors $u \in N_H(v)$ by assigning a charge of $1 / \hat{d}(v)$ to each such neighbor. Since $\sum_{u \in N_H(v)} 1 / \hat{d}(v) = 1$, each bad node contributes a total charge of 1.

    Consider now a node $u \in V_H$ that may receive charge from multiple bad neighbors. The total charge received by $u$ is bounded by the marking probabilities in its neighborhood, which by assumption is less than $1/2$. Hence, no node receives more than $1/2$ units of total charge.

    Since the total charge distributed by the bad nodes is equal to $|B|$, and no node can receive more than $1/2$ units of charge, it follows that $|B| \leq |V_H|/2$. Consequently, at least half of the nodes are good, which means that a constant fraction of nodes are within distance at most $2$ of $S_H$, as required.

    We now turn to the second part of the lemma and show that, with high probability, for any node $u \in V_H$, the number of neighbors $v \in S_H$ with $\hat{d}(v) \geq \hat{d}(u)$ is $O(\log n)$. For each neighbor $v \in N_H(u)$, the probability that $v \in S_H$ is $p_v = \frac{c \log n}{\hat{d}(v)}$. Since we are only considering neighbors with $\hat{d}(v) \geq \hat{d}(u)$, we have $p_v \leq \frac{c \log n}{\hat{d}(u)}$.

    The expected number of such neighbors is then at most $\deg_H(u) \cdot \frac{c \log n}{\hat{d}(u)} \leq c \log n$, using the fact that $\hat{d}(u) \geq \deg_H(u)$. Applying a Chernoff bound, we conclude that with high probability, the number of such neighbors does not exceed $O(\log n)$.
\end{proof}

\noindent The following lemma provides our main ruling set algorithm that we use to compute the center facilities in \Cref{alg:fl-overview}.

\begin{lemma}\label{lemma:weighted2rulingset}
    Let $G=(V,E)$ be a $n$-node graph with $m=|E|$ edges and node weights $w(v)$ such that $w(v)\geq 1$ and $w(v)\leq w_{\max}=\poly(n)$. Further, assume that we are given a parameter $\eps\geq 1/\poly(\log n)$ and an integer $t\geq 1$. There exists an $O\big(t\cdot \log\log n\big)$-round sublinear-space MPC algorithm that computes a set of nodes $S$ with the following properties.
    \begin{enumerate}
        \item[(I)] The nodes in $S$ form an independent set of $G^t$
        \item[(II)] Let $U:=\set{u\in V : \exists v\in S\text{ with }d_G(u,v)\leq 5t}$, $W_V:=\sum_{v\in V} w(v)$, and $W_U:=\sum_{v\in U} w(v)$. W.h.p., we have $W_U\geq (1-\eps)\cdot W_V$.
        \item[(III)] For all $u\in V$, there exists a $v\in S$ for which $d_G(u,v)\leq O(t\log\log\log n)$.
    \end{enumerate}
    The algorithm requires global memory $O((m+n)\cdot n^\eps)$ for a constant $\eps>0$ that can be chosen arbitrarily small.
\end{lemma}
\begin{proof}
     As a first step, for every $v\in V$, we define $w'(v):=2^{\lfloor \log_2 w(v)\rfloor}$, i.e., we round each weight $w(v)$ down to the next integer power of $2$. In this way, we only have $O(\log n)$ different weights $w'(v)\in\set{2^0,2^1,\dots,2^{h}}$, where $h=\lfloor\log_2 w_{\max}\rfloor=O(\log n)$. For each of the possible rounded node weights $2^{\ell}$, let $V_{\ell}$ be the set of nodes with $w'(v)=2^{\ell}$. 
    
    We now focus on one value of $k$. We define $H_{\ell}:=G^{2t}[V_{\ell}]$ to be the subgraph of $G^{2t}$ induced by the nodes $V_{\ell}$ of rounded weight $w'(v)=2^{\ell}$. We first show that one \emph{Adapted Luby Step} as described above can be implemented on $H_{\ell}$ in $O(t)$ rounds in the MPC model. By \Cref{lemma:basicSums} each node $v\in V_{\ell}$ can compute an estimate $\hat{d}(v)$ of its degree $\deg_{H_{\ell}}(v)$ in $H_{\ell}$ such that $\deg_{H_{\ell}}(v)\leq \hat{d}(v)\leq 2\deg_{H_{\ell}}(v)$. Each node $v\in V_{\ell}$ can then be independently added to the set $S_{H_{\ell}}$ with probability $1/\hat{d}(v)$. Finally, one can use \Cref{lemma:basicMins} to determine the set of nodes $U_{\ell}\subseteq V_{\ell}$ such that for $u\in U_{\ell}$, there is a node $v\in S_{H_{\ell}}$ within distance $4t$ in $G$. Note that $U_{\ell}$ includes all nodes $u\in V_{\ell}$ that have a node $v\in S_{H_{\ell}}$ within distance $2$ in $H_{\ell}$. By \Cref{lemma:lubystep}, we therefore know that w.h.p., $U_{\ell}$ contains a constant fraction of the nodes in $V_{\ell}$. If we think of the nodes in $U_{\ell}$ as being removed from the graph, we can repeat this process on the remaining nodes. If we repeat this $c\cdot \log (1/ \eps)$ times for a sufficiently large constant $c$, the number of remaining nodes is at most $\eps/2\cdot |V_{\ell}|$. The final set $S_{H_{\ell}}$ is defined as the union of the sets that we construct in those $O(\log 1/\eps)$ iterations.

    We next move our attention to the graph $G^t$. For a node $u\in V$, let $N_{\ell}^t(u):= S_{H_{\ell}}\cap N^t(u)$ be the set of $S_{H_{\ell}}$-neighbors of $u$ in $G^t$. Note that the nodes in $N_{\ell}^t(u)$ are forming a clique in the graph $H_{\ell}=G^{2t}[V_{\ell}]$. By \Cref{lemma:lubystep}, for every node $v\in S_{H_{\ell}}$, w.h.p., there are at most $O(\log n)$ nodes $w\in S_{H_{\ell}}$ such that $v$ and $w$ are neighbors in $H_{\ell}$ and $\hat{d}(w)\geq \hat{d}(v)$. The edges of the clique induced by the nodes $N_{\ell}^t(u)$ in $H_{\ell}$ can therefore be oriented such that each node $v\in N_{\ell}^t(u)$ has at most $O(\log n)$ outneighbors $w\in N_{\ell}^t(u)$. This directly implies that $|N_{\ell}^t(u)|=O(\log n)$. For every node $u\in V$, the number of $S_{H_{\ell}}$-neighbors in $G^t$ is therefore bounded by $O(\log n)$, w.h.p.
    
    We next define $S_H := \bigcup_{{\ell}=0}^h S_{H_{\ell}}$. Note that since the sets $S_{H_{\ell}}$ can be computed in parallel (by increasing the global space by an $O(\log n)$-factor), the set $S_H$ can be computed in $O(\log\log n)$ rounds. We compute a set $S'$ that satisfies conditions (I) and (II) of the lemma by computing a maximal independent set (MIS) on the induced subgraph $G^t[S_H]$ of $G^t$. Note that when doing this, every node $u\in V$ that is within distance at most $4t$ from a node $v\in S_H$ in $G$ is within distance at most $5t$ of a node in $S'$. Because in every set $V_{\ell}$, the number of nodes that is not within distance distance $4t$ of a node in $S_{H_{\ell}}$ (and thus within distance $5t$ of a node in $S'$) is at most $\eps/2\cdot|V_{\ell}|$, overall, the set of nodes that is not within distance distance $5t$ of a node in $S'$ has rounded weight at most an $\eps/2$-fraction of the total rounded weight of all the nodes in $V$. For the original weights, we therefore have $W_U\geq (1-\eps)\cdot W_V$ as required.
    
    We need to show that given $S_H$, the set $S'$ can be computed in  $O\big(t\cdot(\log\log n \cdot\log\log\log n)\big)$ rounds. We observed above that for every $u\in V$, the number of $S_{H_{\ell}}$-neighbors in $G^t$ is at most $O(\log n)$. Therefore, for every $u\in V$, the number of $S_{H}$-neighbors in $G^t$ is at most $O({\ell}\log n)=O(\log^2 n)$. The graph $G^t[S_H]$ therefore has maximum degree $O(\log^2 n)$. By using \Cref{lemma:basicMins}, the edges of the graph $G^t[S_H]$ can be computed explicitly in $O(t)$ rounds. We can then compute $S'$ by using the MIS algorithm of \cite{GhaffariU19}. For graphs of maximum degree $\Delta$, the algorithm of \cite{GhaffariU19} has a round complexity of $O(\sqrt{\log\Delta}\cdot\log\log\Delta + \sqrt{\log\log n})$. We can therefore compute the set $S'$ from the set $S_H$ in time $O(t+\sqrt{\log\log n}\cdot\log\log\log n)$.
    
    The set $S'$ only satisfies conditions (I) and (II) of the lemma statement. To also satisfy condition (III), we have to add some nodes to $S'$ to make sure that every node that is not covered by $S'$ is within distance $O(t\log\log\log n)$ of the final set $S$. Let $U\subseteq V$ be the set of nodes that are not within distance $O(t\log\log\log n)$ of $S'$. We can apply \Cref{lemma:logloglogRS} to compute a set $S^+\subseteq U$ that guarantees this and we can then set $S:=S' \cup S^+$. The time for computing $S^+$ is $O(t\log\log\log n + \log\log n\cdot\log\log\log n)$. The overall time for computing the set $S$ is therefore also $O(t\log\log\log n + \log\log n\cdot\log\log\log n)$ and we need $O((m+n)\cdot n^\eps)$ memory for some (arbitrarily small) constant $\eps>0$ in order to compute the set $S^+$.
\end{proof}

\subsection{Implementation of Our Algorithm in the MPC Model}
\label{sec:MPC_impl}

Before discussing the implementation of our facility location and $k$-means algorithm in the MPC model, we give a lemma that provides a standard building block that we use in several of the steps.

\begin{lemma}\label{lemma:groupwiseaggregation}
    Assume that we have $n$ elements, where each element has some label and some value. There is an $O(1)$-round fully scalable MPC algorithm to group the elements by label and aggregate (e.g., compute the sum/min/max of the values) over each group. Each machine that initially stores an element of some given label in the end learns the aggregate value for that label.
\end{lemma}
\begin{proof}
    By using the algorithm of \cite{goodrich2011sorting}, we can first sort the elements according to their label. Using one global aggregation tree, it is then straightforward to aggregate over the elements of each group. If each machine has a local memory of $n^\sigma$ bits for some constant $\sigma$, one can build an aggregation tree of fan-out $n^{\Theta{\sigma}}$ and therefore of $O(1)$ height. The result of the aggregation can be distributed to the machines that initially hold the elements by using the same tree (communicating in the reverse order).
\end{proof}

We first focus on the implementation of \Cref{alg:fl-overview}, i.e., of the facility location algorithm. Let us therefore assume that we are given a set of facilities $\F\subset\mathbb{R}^d$, a set of clients $\C\subset\mathbb{R}^d$, and a fixed opening cost $\lambda>0$ per facility. As a first step, we take the set of points $\F\cup \C\subset \mathbb{R}^d$ and apply \Cref{lemma:graphapprox} to obtain a graph $G=(\F\cup\C, E)$ with edge weights $w(e)>0$ such that $G$ satisfies the conditions of \Cref{lemma:graphapprox}. In the following, we will use $\Gamma$ as the constant that is guaranteed by \Cref{lemma:graphapprox}.

In order to implement \Cref{alg:fl-overview} in the MPC model, we have to approximate the size of balls around points in $\mathbb{R}^d$. We first introduce some notation. As already defined in \Cref{sec:prelim}, we use $B_X(x, r):=\set{y\in X\,:\,\dist(x,y)\leq r}$
to denote the ball of radius $r$ around $x\in \R^d$, restricted to the points in $X$. We further define corresponding neighborhoods in the graph $G$ that we use to approximate $(P,\dist)$. For $x, y\in P$, we use $d_G^{(2)}(x, y)$ to denote the length of the shortest path connecting $x$ and $y$ and consisting of at most $2$ hops. Recall that by the definition of $G$ (cf.\ \Cref{lemma:graphapprox}), for any two point $x,y\in X$, we have $d_G^{(2)}(x,y)\leq \dist(x,y)$. We define 
\[
B_X^+(x,r) := \set{y \in X\,:\, d_G^{(2)}(x,y)\leq r}.
\]
Note that by \Cref{lemma:graphapprox}, for all $x$, $r$, and $X$, we immediately get
\begin{equation}\label{eq:ballcontainment}
    B_X^+\left(x, \frac{r}{\Gamma}\right) \subseteq B_X(x, r) \subseteq B_X^+(x, r).
\end{equation}

We now have everything that we need to describe the implementations of the individual steps of \Cref{alg:fl-overview}.

\begin{lemma}\label{lemma:fl-stepI}
    Step I of \Cref{alg:fl-overview} can be implemented in $O(1)$ MPC rounds with $C_R=9\Gamma$.
\end{lemma}
\begin{proof}
    For convenience, we define a function $\phi(r)$ as
    \[
    \phi(r) := \sum_{c\in\C}\left[r^2 - \cost(c,f)\right]^+ = \sum_{c\in B_{\C}(f, r)} (r^2 - \cost(c,f)).
    \]
    Recall that for each facility $f\in \F$, $r_f$ is defined as
    $r_f := \min_{r\geq 0}\set{\sum_{c\in\C}\left[r^2 - \cost(c,f)\right]^+\geq \lambda}$ and thus $r_f$ is defined as the smallest $r$ for which $\phi(r)\geq\lambda$. Our goal is to compute an estimate $\hat{r}_f$ such that $r_f/C_R \leq \hat{r}_f\leq r_f$ (where $C_R=9\Gamma$).

    First note that for every $r\geq 0$, we have
    \begin{equation}\label{eq:phiupper}
        \phi(r) = \sum_{c\in B_{\C}(f, r)} (r^2 - \cost(c,f)) 
        \ \leq\  |B_{\C}(f,r)|\cdot r^2\ \leq\ |B_{\C}^+(f,r)|\cdot r^2.
    \end{equation}
    The last inequality follows from \Cref{eq:ballcontainment}. We define $r_f'$ to be the supremum over the values $r$ for which $|B_{\C}^+(f,r)|\cdot r^2\leq \lambda$. Note that by \Cref{eq:phiupper} this implies that every $\xi>0$, we have $\phi(r_f'-\xi)<\lambda$ and thus $r_f'\leq r_f$.

    For any $r\geq 0$, we can further lower bound $\phi(r)$ as follows
    \begin{equation}\label{eq:philower}
        \phi(r) \geq \sum_{c\in B_{\C}(f, r/2)} (r^2 - \cost(c,f))
        \ \geq\ 
        |B_{\C}(f, r/2)|\cdot\left(\frac{r}{2}\right)^2
        \ \geq\ 
        |B_{\C}^+(f, r/(2\Gamma))|\cdot \left(\frac{r}{2\Gamma}\right)^2.
    \end{equation} The second inequality follows because for all $c\in B_{\C}(f, r/2)$, we have $\cost(c,f)\leq (r/2)^2$ and because $r^2-(r/2)^2>(r/2)^2$. The last inequality follows from \Cref{eq:ballcontainment} and from the fact that $\Gamma\geq 1$. The definition of $r_f'$ together with the fact that $\phi(2r_f/3)<\lambda$ and \Cref{eq:philower} imply that $r_f'\geq r_f/(3\Gamma)$. We therefore know that $r_f/(3\Gamma)\leq r_f'\leq r_f$. The radius $r_f'$ thus satisfies the requirements for the estimate $\hat{r}_f$ (with $C_R=3\Gamma$). We can however not efficiently compute $r_f'$ exactly for every facility $f$.

    We can however compute a value $\hat{r}_f$ for which $r_f'/3\leq\hat{r}_f \leq r_f'$ as follows. For every integer $\ell\geq 0$ and each facility $i$, we define $b_{\ell}^+(f):=|B_{\C}^+(f, 2^\ell)|$ to be the size of the number of clients $c$ within $2$-hop distance at most $2^\ell$ in $G$. For a fixed $\ell$ and a constant $\eps>0$, we can use \Cref{lemma:basicSums} to compute a value $\tilde{b}_{\ell}^+(f)$ with $b_{\ell}^+(f)\leq \tilde{b}_{\ell}^+(f)\leq 3/2\cdot b_{\ell}^+(f)$ for every facility $f$ in $O(1)$ rounds in the sublinear-memory MPC model. And since we assume that the maximum distance between any two points is polynomial in $n$, by increasing the global memory by a factor $O(\log n)$, we can compute $\tilde{b}_{\ell}^+(f)$ for all $f$ and all integers ${\ell}$ in parallel in $O(1)$ sublinear-space MPC rounds. We define
    \[
    \hat{{\ell}}_f := \max_{{\ell}\in \mathbb{N}_0}\set{\tilde{b}_{\ell}^+(f)\cdot 4^{\ell} \leq \lambda}
    \quad\text{and}\quad
    \hat{r}_f := 2^{\hat{{\ell}}_f}.
    \]
    Note that $\hat{r}_f\leq r_f'$ because for any ${\ell}>\log_2(r_f')$, we have 
    \[
    \tilde{b}_{\ell}^+(f)\cdot 4^{\ell} \geq b_{\ell}^+(f)\cdot 4^{\ell} = |B_{\C}(f, 2^{\ell})|\cdot 4^{\ell} >\lambda.
    \]
    The last inequality follows because $2^{\ell}>r_f'$ and we defined $r_f'$ as the supremum over all $r$ for which $|B_{\C}(f, r)|\cdot r^2 \leq\lambda$ To see that $\hat{r}_f \geq r_f' / 3$, consider the largest $\ell$ for which $\tilde{b}_{\ell}^+(f) \cdot 4^{\ell} \leq \lambda$. Since $\tilde{b}_{\ell}^+(f) \geq b_{\ell}^+(f)$, we know that
    \[
    b_{\ell}^+(f) \cdot 4^{\ell} \leq \lambda.
    \]
    By the definition of $r_f'$, this implies $r_f' \leq 2^{\ell}$. 

    On the other hand, from the definition of $\tilde{b}_{\ell}^+(f)$, we also have:
    \[
    \tilde{b}_{\ell+1}^+(f) \cdot 4^{\ell+1} > \lambda.
    \]
    Using the fact that $\tilde{b}_{\ell+1}^+(f) \leq \frac{3}{2} b_{\ell+1}^+(f)$, we get:
    \[
    \frac{3}{2} b_{\ell+1}^+(f) \cdot 4^{\ell+1} > \lambda \quad \Rightarrow \quad b_{\ell+1}^+(f) \cdot 4^{\ell+1} > \frac{2}{3} \lambda.
    \]
    By the definition of $r_f'$, this implies that $r_f' \geq 2^{\ell+1} / 3$. Since $\hat{r}_f = 2^{\ell}$, we conclude:
    \[
    \hat{r}_f \geq \frac{r_f'}{3}.
    \]
    Thus, we obtain:
    \[
    \frac{r_f'}{3} \leq \hat{r}_f \leq r_f'.
    \]
    Combining this with our previous bound $r_f / (3\Gamma) \leq r_f' \leq r_f$, we conclude:
    \[
    \frac{r_f}{9\Gamma} \leq \hat{r}_f \leq r_f.
    \]
    Hence, $\hat{r}_f$ satisfies the required approximation guarantee with $C_R = 9\Gamma$.

    Finally, since $\tilde{b}_{\ell}^+(f)$ values can be computed in $O(1)$ rounds and the selection of $\hat{\ell}_i$ requires only a maximum search over $O(\log n)$ values, the entire procedure runs in $O(1)$ MPC rounds.
\end{proof}

Now, we move towards the second step of \Cref{alg:fl-overview} and prove that it can be implemented in $O(1)$ MPC rounds. Specifically, we show that the computation of the initial dual values $\alpha_{c,0}$ for each client $j$ can be performed efficiently in parallel.

\begin{lemma}\label{lemma:fl-stepII}
    Step II of \Cref{alg:fl-overview} can be implemented in $O(1)$ MPC rounds with constants $C_D^+ = 8\Gamma^4$ and $C_D^- = 2\Gamma^2$.
\end{lemma}

\begin{proof}
    Recall that for each client $c \in \C$, the value $\alpha_c^*$ is defined as $\alpha_c^* = \min_{f \in \F} \max\{r_f^2, \cost(c,f)\}$. This represents the smallest dual value for which client $c$ can contribute to the opening of a facility. Our goal is to compute an estimate $\alpha_{c,0}$ such that
    \begin{align*}
        \frac{ \min_{f \in \F} \max\{r_f^2, \cost(c,f)\}}{C_D^+}=\frac{\alpha_c^*}{C_D^+} \leq \alpha_{c,0} \leq \frac{\alpha_c^*}{C_D^-} = \frac{ \min_{f \in \F} \max\{r_f^2, \cost(c,f)\}}{C_D^-},
    \end{align*}
    where the constants $C_D^+$ and $C_D^-$ ensure the necessary approximation guarantees.

    For convenience, for a ball of radius $2^{\ell}$ around a given client $c \in \C$, we define the function $\rho(\ell)$ as
    \begin{align*}
        \rho(\ell) = \min_{f \in B_{\F}(c, 2^\ell)} \{r_f\}.
    \end{align*}
    Here, $\ell \in [\log n]$ because, as discussed in \Cref{sec:prelim}, the distance between any two input points $x$ and $y$ satisfies $1 \leq \dist(x, y) \leq \poly(n)$, implying that the number of relevant distance scales is at most $O(\log n)$.
    
    For each client $c \in \C$, we define a value $\bar{\alpha}_c$ as
    \begin{align*}
        \bar{\alpha}_c = \min_{\ell} \max\left\{\rho(\ell)^2, 4^{\ell}\right\}.
    \end{align*}
    
    We now show that this definition ensures the bound $\alpha_c^* \leq \bar{\alpha}_c \leq 4 \alpha_c^*$. To that end, observe that for each facility $f \in \F$ at distance $d(c, f)$ from client $c$, the facility first appears in the ball $B_{\F}(c, 2^\ell)$ with $\ell = \ell_f := \lceil \log_2 d(c, f) \rceil$. Since all balls with larger radius also contain $f$, and the minimization in the definition of $\bar{\alpha}_c$ is taken over all such $\ell$, the contribution of facility $f$ is accounted for at index $\ell_f$.
    
    Thus, for every facility $f$, we have:
    \begin{align*}
        \max\{r_f^2, \cost(c,f)\} \leq \max\{r_f^2, 4^{\ell_f}\},
    \end{align*}
    and therefore:
    \begin{align*}
        \alpha_c^* = \min_{f \in \F} \max\{r_f^2, \cost(c,f)\} \leq \min_{\ell} \max\{\rho(\ell)^2, 4^\ell\} = \bar{\alpha}_c.
    \end{align*}
    
    To obtain the upper bound, consider any facility $f^*$ that achieves the minimum in the definition of $\alpha_c^*$, i.e.,
    \begin{align*}
        \alpha_c^* = \max\{r_{f^*}^2, \cost(c, f^*)\}.
    \end{align*}
    Let $\ell^* = \lceil \log_2 d(c, f^*) \rceil$. Then $f^* \in B_{\F}(c, 2^{\ell^*})$ and hence:
    \begin{align*}
        \rho(\ell^*) \leq r_{f^*}, \quad \text{so} \quad \bar{\alpha}_c \leq \max\{r_{f^*}^2, 4^{\ell^*}\}.
    \end{align*}
    Using $\cost(c, f^*) = d(c, f^*)^2 \leq 4^{\ell^*}$, it follows that:
    \begin{align*}
        \bar{\alpha}_c &\leq \max\{r_{f^*}^2, 4^{\ell^*}\} \leq 4 \cdot \max\{r_{f^*}^2, d(c, f^*)^2\} \\
        &= 4 \alpha_c^*.
    \end{align*}
    
    Combining both bounds yields:
    \begin{align}\label{eq:baralpha}
        \alpha_c^* \leq \bar{\alpha}_c \leq 4 \alpha_c^*.
    \end{align}

    We further define the function $\rho^+(\ell)$ as follows:
    \begin{align*}
        \rho^+(\ell) = \min_{f \in B_{\F}^+(c, 2^\ell)} \{r_f\}.
    \end{align*}
    Moreover, we define the value $\alpha_c^+$ as
    \begin{align*}
        \alpha_c^+ = \min_{\ell} \max\left\{\rho^+(\ell)^2, 4^{\ell}\right\}.
    \end{align*}
    
    We now show that this definition ensures the bound $\alpha_c^*/\Gamma^2 \leq \alpha_c^+ \leq 4 \alpha_c^*$. Considering the definitions of $\alpha_c^+$ and $\bar{\alpha}_c$, and noting that both are defined as a minimum over $\ell$ of the maximum between the squared radius and $4^{\ell}$, we apply the containment property from \Cref{eq:ballcontainment}. Since $B_{\F}(c, 2^\ell) \subseteq B_{\F}^+(c, 2^\ell)$, and the minimization across all $\ell$ naturally discards large-radius facilities when possible, it follows that:
    \begin{align*}
        \alpha_c^+ \leq \bar{\alpha}_c.
    \end{align*}
    
    Using the previously established bound $\bar{\alpha}_c \leq 4\alpha_c^*$ from \eqref{eq:baralpha}, we obtain:
    \begin{align*}
        \alpha_c^+ \leq 4\alpha_c^*.
    \end{align*}
    
    To establish a lower bound on $\alpha_c^+$, we expand its definition as follows:
    \begin{align*}
        \alpha_c^+ &= \min_{\ell} \max\left\{\rho^+(\ell)^2, 4^{\ell}\right\} \\
                   &= \min_{\ell} \max\left\{\min_{f \in B_{\F^+}(c, 2^\ell)} r_f^2, 4^\ell\right\} \\
                   &= \min_{f} \max\left\{r_f^2, 4^{\ell^+_f}\right\},
    \end{align*}
    where $\ell^+_f := \min\{\ell \mid f \in B_{\F}^+(c, 2^\ell)\}$.
    
    From the definition of $B^+$ and the underlying graph structure, we have the guarantee:
    \begin{align*}
        \frac{\dist(c, f)}{\Gamma} \leq 2^{\ell^+_f}.
    \end{align*}
    
    Since $\cost(c,f) = \dist(c, f)^2$, we conclude:
    \begin{align*}
        \frac{\cost(c,f)}{\Gamma^2} \leq 4^{\ell^+_f}.
    \end{align*}
    
    Now, combining the above with the definition of $\alpha_c^*$, we obtain:
    \begin{align*}
        \frac{\alpha_c^*}{\Gamma^2} \leq \alpha_c^+.
    \end{align*}
    Combining both bounds yields:
    \begin{align}
        \alpha_c^*/\Gamma^2 \leq \alpha^+_c \leq 4 \alpha_c^*.
    \end{align}
    
    Finally, we set
    \begin{align*}
        \alpha_{c,0} := \frac{\alpha_c^+}{8\Gamma^2},
    \end{align*}
    which yields the desired bounds:
    \begin{align*}
        \frac{\alpha_c^*}{8\Gamma^4} \leq \alpha_{c,0} \leq \frac{\alpha_c^*}{2\Gamma^2}.
    \end{align*}

    This concludes the proof.
    
\end{proof}

Now, we move towards the third step of \Cref{alg:fl-overview} and prove that it can be implemented in $O(1)$ MPC rounds. Specifically, we show that finding the set of problematic clients can be performed efficiently in parallel.

\begin{lemma}\label{lemma:fl-stepIII}
    Step III of \Cref{alg:fl-overview} can be implemented in $O(1)$ MPC rounds with constants $\gamma_1 = 4\Gamma^4$, $\gamma_2 = 9\Gamma^4$, and $Q = 8\Gamma^4$.
\end{lemma}

\begin{proof}
    Recall that a client $c \in \C$ is marked as problematic if there exists another client $c' \in B_{\C}(c, \sqrt{\alpha_c^*})$ such that
    \begin{align*}
        \alpha_{c',0} \leq \frac{\alpha_{c,0}}{Q}.
    \end{align*}
    However, since the algorithm only operates on approximations $\alpha_{c,0}$ of the true values $\alpha_c^*$, where
    \begin{align*}
        \frac{\alpha_c^*}{8\Gamma^2} \leq \alpha_{c,0} \leq \frac{\alpha_c^*}{2},
    \end{align*}
    which means we conservatively search within $B_{\C}(c, \sqrt{\alpha_{c,0}}) \in B_{\C}(c, \sqrt{\alpha_c^*/2})$.

    Moreover, the algorithm uses the approximate neighborhood $B^+_{\C}$ (computed from the power graph) instead of the exact ball $B_{\C}$. From \Cref{eq:ballcontainment}, we know that
    \begin{align*}
        B_{\C}(j, \sqrt{\alpha_c^*/2}) \subseteq B^+_{\C}(j, \sqrt{\alpha_c^*/2}),
    \end{align*}
    which ensures that any client $c'$ that could make $c$ problematic is still included in the search region.

    By \Cref{lemma:graphapprox}, the power graph guarantees that for all clients $x, y$,
    \begin{align*}
        \dist(x, y) \leq \Gamma \cdot d^{(2)}_G(x, y),
    \end{align*}
    so any client $c'$ in $B^+_{\C}(c, \sqrt{\alpha_c^*/2})$ lies within distance at most $\Gamma \cdot \sqrt{\alpha_c^*/2}$ of $c$, which is within distance at most $\gamma_1\cdot \sqrt{\alpha_j^*}$.

    To show that the connection cost from a problematic client $c$ to a facility $f'$ remains bounded after freezing $\alpha_{c}$, we assume that $c$ previously contributed to $f'$ and is now frozen due to the presence of a nearby client $c'$ such that $\alpha_{c',0} \le \alpha_{c,0}/Q$.

    Let $c'$ be the witness for $c$ being problematic, and suppose that in the final solution, $c$ connects to facility $f'$ (which $c'$ contributes to). We aim to bound $\cost(c,f')$ in terms of $\alpha_{c,0}$.
    
    Since $c'$ lies within distance at most $\sqrt{\alpha_c^*/2}$ of $c$, and we only have access to approximate distances via \Cref{lemma:graphapprox}, the actual distance is bounded by:
    \begin{align*}
        \dist(c, c') &\le \Gamma \cdot d_G^{(2)}(c, c') \le \Gamma \cdot \sqrt{\alpha_c^*/2}.
    \end{align*}
    
    Moreover, since $c'$ contributes to facility $f'$, we have:
    \begin{align*}
        \dist(c', f') \le \sqrt{\alpha_{c',1}} \le \sqrt{C_A \cdot \alpha_{c',0}} \le \sqrt{C_A \cdot \frac{\alpha_{c,0}}{Q}}.
    \end{align*}
    
    Using the triangle inequality, we can now bound $\dist(j, i')$:
    \begin{align*}
        \dist(c, f') &\le \dist(c, c') + \dist(c', f') \\
        &\le \Gamma \cdot \sqrt{\frac{\alpha_c^*}{2}} + \sqrt{C_A \cdot \frac{\alpha_{c,0}}{Q}}.
    \end{align*}
    
    Using the bound $\alpha_c^* \le 8\Gamma^2 \cdot \alpha_{c,0}$ from Step II, we obtain:
    \begin{align*}
        \dist(c, f') &\le \Gamma \cdot \sqrt{4\Gamma^2 \cdot \alpha_{c,0}} + \sqrt{\frac{C_A}{Q}} \cdot \sqrt{\alpha_{c,0}} \\
        &= 2\Gamma^2 \cdot \sqrt{\alpha_{c,0}} + \sqrt{\frac{C_A}{Q}} \cdot \sqrt{\alpha_{c,0}} \\
        &= \left(2\Gamma^2 + \sqrt{\frac{C_A}{Q}}\right) \cdot \sqrt{\alpha_{c,0}}.
    \end{align*}
    
    Squaring both sides, we get:
    \begin{align*}
        \cost(c, f') = \dist^2(c, f') \le \left(2\Gamma^2 + \sqrt{\frac{C_A}{Q}}\right)^2 \cdot \alpha_{c,0}.
    \end{align*}
    
    Thus, the connection cost of a problematic client $c$ remains bounded by a constant factor of $\alpha_{c,0}$, where:
    \[
    \gamma_2 := \left(2\Gamma^2 + \sqrt{\frac{C_A}{Q}}\right)^2.
    \]

    Given the requirement that $Q \ge C_D^+/C_D^- = 4\Gamma^2$, we choose $C_A = \Gamma^4 \cdot Q$, which implies that $C_A / Q \le \Gamma^4$. As a result, we obtain $\gamma_2 = 9\Gamma^4$.

    Finally, note that $B^+_{\C}$ corresponds to the two-hop neighborhood in the graph $G$ and can be computed in $O(1)$ MPC rounds. Therefore, the entire implementation of Step III requires only constant rounds.
\end{proof}

\begin{lemma}\label{lemma:fl-stepIV}
    Step IV of \Cref{alg:fl-overview} can be implemented in $O(1)$ MPC rounds, and it computes a set $S \subseteq \mathcal{F}$ of facilities satisfying the following properties:
    \begin{enumerate}[(a)]
        \item $S$ includes all fully paid facilities in $\mathcal{F}$.
        \item Every facility $i \in S$ is $\kappa$-approximately paid, where $\kappa = \Gamma^2$.
    \end{enumerate}
\end{lemma}

\begin{proof}
    Recall that a facility $i$ is called \emph{paid} if the total contribution from clients in its neighborhood satisfies the condition:
    \begin{align*}
        \sum_{c \in B_{\C}(f, r_f)} \left(r_f^2 - \cost(c,f)\right) = \lambda,
    \end{align*}
    where $\lambda$ is the opening cost of facility $f$. 

    To determine whether a facility is paid, we need to compute this sum by considering all clients within the ball $B_{\C}(f, r_f)$. Specifically, we rewrite the summation as:
    \begin{align*}
        \sum_{c \in B_{\C}(f, r_f)} \left(r_f^2 - \cost(c,f)\right) = |B_{\C}(f, r_f)| \cdot r_f^2 - \sum_{c \in B_{\C}(f, r_f)} \cost(c,f).
    \end{align*}
    Thus, we only need to compute $|B_{\C}(f, r_f)|$, the number of clients in the ball, and the sum of distances $\cost(c,f)$ for all $c \in B_{\C}(f, r_f)$.

    However, we do not have direct access to $B_{\C}(f, r_f)$. Instead, following \Cref{eq:ballcontainment}, we approximate it using $B_{\C}^+(f, r_f)$, which can be accessed in $O(1)$ rounds. The issue with using $B_{\C}^+(f, r_f)$ is that it might include additional clients whose true Euclidean distances from $f$ fall within the range $[r_f, \Gamma \cdot r_f]$. These additional clients contribute negatively to the sum since $r_f^2 - \cost(c,f)$ becomes negative for them, potentially misclassifying a truly paid facility as unpaid.

    To mitigate this issue, when summing over all clients in $B_{\C}^+(f, r_f)$, we only consider those clients whose distances satisfy $\cost(c,f) \leq r_f^2$, ensuring that only valid contributions are included.

    Additionally, as in previous steps, we do not work directly with $r_f$, but rather with its approximation $\hat{r}_f$, which satisfies:
    \begin{align*}
        \frac{r_f}{9\Gamma} \leq \hat{r}_f \leq r_f.
    \end{align*}
    This means that we must evaluate the condition using $\hat{r}_f$ instead of $r_f$. Since we only have a constant-factor approximation of $r_f$, we must adjust the threshold $\kappa$ to be $\Gamma^2$. This ensures that facilities that are actually paid remain paid in our approximation, though some additional facilities that are not fully paid may also be included.

    Given that accessing $B_{\C}^+(f, r_f)$, computing its size, and summing contributions all take $O(1)$ MPC rounds, the overall computation remains efficient.
\end{proof}

\begin{lemma}\label{lemma:fl-stepV}
    Step $V$ of \Cref{alg:fl-overview} can be implemented in $O(1)$ MPC rounds. More specifically, there is an $O(1)$-round MPC algorithm that computes an unweighted graph $H_{base}$ with at most $n^{1+\eps}$ edges (for $\eps>0$ an arbitrarily small constant) such that the graph $H_{base}^2$ satisfies the requirements for graph $H'$ in \Cref{alg:fl-overview} with the constants $C_{H,1}$ and $C_{H,2}$ given by
    \[
    C_{H,1}=4\cdot C_R^4\cdot \zeta^2,\quad
    C_{H,2}=12\cdot\sqrt{\kappa\rho}\cdot C_R^3\cdot \zeta^2.
    \]
    where $\rho$ and $\zeta$ are the constants from \Cref{lem:constant-rad,lemma:sameradius}.
\end{lemma}

\begin{proof}
    Recall that the dependency graph $H = (\S, E_H)$ on the set $\S$ is defined as follows:
    \begin{align*}
        E_H := \left\{\{f,f'\} \in \binom{\S}{2} \mid 
        \exists c \in \C \text{ such that } \kappa\alpha_{c,1} > \max\{\cost(c,f), \cost(c,f')\} \right\}.
    \end{align*}
    We cannot directly construct the graph $H$. We can however construct an an alternative graph $H' = (S, E_{H'})$, ensuring that $E_H \subseteq E_{H'}$. To understand the conditions that $E_{H'}$ must satisfy, let us first understand some basic properties of $H$. Consider some edge $\set{f, r'}$ of $H$. W.l.o.g., assume that $r_f\leq r_{r'}$. First, by \Cref{lem:constant-rad} we know that $\alpha_{c,1}\leq \rho\cdot r_f^2$, where $\rho$ is the constant specified in the statement of \Cref{lem:constant-rad}. Because we also know that $\alpha_{c,1}\geq \kappa\max\{\cost(c,f), \cost(c,f')\}$, the distance between $f$ and $f'$ can therefore be bounded as
    \begin{equation}\label{eq:dependencymindist}
    \dist(f,f') \leq \dist(f,c)+\dist(c,f')\leq 2\cdot\sqrt{\kappa\cdot\alpha_{c,1}} \leq 
    2\cdot\sqrt{\kappa\cdot\rho}\cdot r_f.
    \end{equation}
    From \Cref{lemma:sameradius}, we further know that $r_{f'}\leq \zeta\cdot r_f$ for the constant $\zeta$ that is defined in the statement of \Cref{lemma:sameradius}. The algorithm does not know the exact value of $r_f$ for a facility $f$. In Step I, it however computes an estimate $\hat{r}_f$ such that $r_f/C_R\leq \hat{r}_f\leq r_f$. For $f$ and $f'$, we therefore know that
    \begin{equation}\label{eq:dependencyradii}
        \max\set{\hat{r}_f, \hat{r}_{f'}} \leq C_R\cdot \zeta\cdot
        \min\set{\hat{r}_f, \hat{r}_{f'}}.
    \end{equation}
    If we construct the graph $H'$ such that any two facilities $f,f'\in \S$ for which \eqref{eq:dependencymindist} and \eqref{eq:dependencyradii} hold are connected by an edge in $H'$, then we definitely guarantee that $E_H\subseteq E_{H'}$ (i.e., $H$ is a subgraph of $H'$). We want to achieve this while not adding any edges $\set{f,f'}$ to $H'$ for which $\dist(f,f')=\omega(\min\{r_f,r_{f'}\})$ or for which $r_f$ and $r_{f'}$ are not within a constant factor. Moreover since $H$ and thus also $H'$ might be a dense graph, we cannot explicitly compute it. We will construct a sufficiently sparse graph $H_{base}$ such that $H_{base}^2$ satisfies the conditions of $H'$. 

    As a first step, we define $O(\log n)$ radius classes and assign the facilities in $\S$ to those classes. Each facility $f\in\S$ is assigned to $O(1)$ of the classes. For $i=0,1,\dots,O(\log n)$, we define
    \[
    \S_i := \set{f\in \S \,:\, 2^i \leq \hat{r}_f \leq 2^{i+1}\cdot C_R\cdot \zeta}.
    \]
    Note that by \eqref{eq:dependencyradii}, for any two facilities $f,f'$ that are connected by an edge in $H$, there is at least one radius class $\S_i$ such that $f$ and $f'$ are both in $\S_i$. For each set $\S_i$, we first compute a weighted graph $G_i$ by applying \Cref{lemma:graphapprox}. The lemma guarantees that $G_i$ can be computed in $O(1)$ time in the MPC model, that $G_i$ has at most $|\S_i|^{1+\eps}\leq n^{1+\eps}$ edges, and that any two facilities $f,f'\in \S_i$ are connected by a path of hop-length at most $2$ and total weight at most $\dist(f,f')$ in $G_i$. In addition, for any $f,f'$, the weighted distance in $G_i$ is at least $\dist(f,f')/\Gamma$. Base on $G_i$, we now define an unweighted graph $H_i=(\S_i,E_{H_i})$, where the edge set $E_{H_i}$ is defined as
    \[
    E_{H_i} := \set{\set{f,f'}\in \binom{\S_i}{2}\,:\,
    G_i\text{ contains an edge }\set{f,f'}\text{ of weight }\leq
    2^{i+2}\cdot\sqrt{\kappa\rho}\cdot C_R^2\cdot\zeta}.
    \]
    We next verify that if $f,f'\in \S_i$ and $\set{f,f'}\in E_H$, then $H_i$ contains a path of length at most $2$ between $f$ and $f'$. We thus have to show that $f,f'\in \S_i$ and \eqref{eq:dependencymindist} holds, then $H_i$ contains a path of length at most $2$ between $f$ and $f'$. From \eqref{eq:dependencymindist}, we have
    \[
    \dist(f,f') \leq 2\sqrt{\kappa\rho}\cdot r_f \leq
    2\sqrt{\kappa\rho}C_R\cdot \hat{r}_f \leq
    2\sqrt{\kappa\rho}\cdot C_R\cdot 2^{i+1}\cdot C_R\cdot \zeta.
    \]
    The last inequality follows from $f\in \S_i$. We therefore know that $G_i$ contains a path of hop length $2$ and total weight at most $2^{i+2}\cdot\sqrt{\kappa\rho}\cdot C_R^2\cdot\zeta$. Consequently, $H_i$ contains a path of length at most $2$ between $f$ and $f'$.

    We now define our base graph $H_{base}=(\S, E_{base})$ as $E_{base} := E_{H_0}\cup E_{H_1} \cup \dots$. From the above observation and the fact that for any edge $\set{f,f'}$ of $H$, there is some $i$ for which $f,f'\in \S_i$, $H_{base}$ definitely contains a path of length at most $2$ between $f$ and $f'$. The graph $H':=H_{base}^2$ therefore is a supergraph of $H$ as desired.

    It remains to show that for any edge $\set{f,f'}$ in $H'$, we have
    \[
    \max\set{r_f,r_{f'}}\leq C_{H,1}\cdot \min\set{r_f, r_{f'}}
    \quad\text{and}\quad
    \dist(f,f') \leq C_{H,2}\cdot\min\set{r_f, r_{f'}}.
    \]
    Note that for any edge $\set{f,f'}$ in $H_{base}$, the values of $\hat{r}_f$ and $\hat{r}_{f'}$ are within a factor at most $2^{i+1}\cdot \C_R\cdot \zeta$. The values or $r_f$ and $r_f'$ are then within a factor of at most $2\cdot C_R^2\cdot \zeta$. Consequently, for any edge $\set{f,f'}$ of $H'$, we $r_f$ and $r_f'$ are within a factor
    \[
    C_{H,1} = 4\cdot C_R^4\cdot \zeta^2.
    \]
    Let us now also upper bound the distance between any two facilities $f,f'$ that are connected by an edge in $H'$ (and thus by a path of length at most $2$ in $H_{base}$). We first consider two facilities $f,f'$ that are connected by an edge in $H_{base}$. Assume that this edge was added as part of the graph $H_i$, i.e., both $f$ and $f'$ are in $\S_i$. We know that $f$ and $f'$ are connected by an edge of weight at most $2^{i+2}\cdot\sqrt{\kappa\rho}\cdot C_R^2\cdot\zeta$ in $G_i$. From the properties of $G_i$ (cf.\ \Cref{lemma:graphapprox}), we know that
    \begin{eqnarray*}
        \dist(f,f') & \leq &
        \Gamma\cdot d_{G_i}(f,f')\\
        & \leq & \Gamma\cdot 2^{i+2}\cdot\sqrt{\kappa\rho}\cdot C_R^2\cdot\zeta\\
        & \leq & 4\cdot\sqrt{\kappa\rho}\cdot C_R^2\cdot\zeta\cdot\min\set{\hat{r}_f,\hat{r}_{f'}}\\
        & \leq & 4\cdot\sqrt{\kappa\rho}\cdot C_R^2\cdot\zeta\cdot\min\set{r_f,r_{f'}}.
    \end{eqnarray*}
    Now consider a path $f, f', f''$ of length $2$ in $H_{base}$. Assume that $\hat{r}_f\leq \hat{r}_{f''}$. The above calculation implies that
    \begin{eqnarray*}
        \dist(f,f'') & \leq &
        4\cdot\sqrt{\kappa\rho}\cdot C_R^2\cdot\zeta\cdot\min\set{\hat{r}_f,\hat{r}_{f'}} +
        4\cdot\sqrt{\kappa\rho}\cdot C_R^2\cdot\zeta\cdot\min\set{\hat{r}_{f'},\hat{r}_{f''}}\\
        & \leq &
        4\cdot\sqrt{\kappa\rho}\cdot C_R^2\cdot\zeta\cdot\hat{r}_f +
        4\cdot\sqrt{\kappa\rho}\cdot C_R^2\cdot\zeta\cdot\hat{r}_{f'}\\
        & \leq &
        4\cdot \sqrt{\kappa\rho}\cdot C_R^2\cdot\zeta\cdot 
        \big(1+2\cdot C_R\cdot\zeta\big)\cdot \hat{r}_f\\
        & \leq &
        12\cdot \sqrt{\kappa\rho}\cdot C_R^3\cdot \zeta^2\cdot \min\set{r_f,r_{f''}}.
    \end{eqnarray*}
    The last inequality in particular uses that for all $f$, $\hat{r}_f\leq r_f$ and thus $\hat{r}_f=\min\set{\hat{r}_f, \hat{r}_{f''}}\leq\min\set{r_f,r_{f''}}$.
\end{proof}

\begin{lemma}\label{lemma:fl-stepVI}
    Step VI of \Cref{alg:fl-overview} can be implemented in $\mathcal{O}(\log\log n\cdot\log\log\log n)$ MPC rounds.
\end{lemma}
\begin{proof}
    We start by finding one a facility $f_c\in\S$ for which $\max\big\{r_f^2,\cost(c,f)\big\}=O(\alpha_{c,0})$. By \Cref{lem:paidfacility}, such a facility must exist. In order to find it, every client $c$ searches for the $\S$-node with minimum $\hat{r}_f$-value within an approximate $O(\alpha_{c,0})$ ball of sufficiently large radius. This can be done in $O(1)$ rounds by using the graph $G$ that is provided by \Cref{lemma:graphapprox} and by using \Cref{lemma:basicMins} to probe the $2$-hop neighborhood in $G$ restricted to edges of weight $O(\alpha_{c,0})$. In this way, in $O(1)$ time, every client can find its facility $f_c$.

    We will compute the set of cluster centers $\S_0\subseteq \S$ by using the approximate ruling set algorithm of \Cref{lemma:weighted2rulingset}. However in order to guarantee that a large fraction of the total dual client values is within constant distance of a cluster center, we first need to compute weights for the facilities $f\in \S$. The weight $w(f)$ of a facility $f\in \S$ is defined as $w(f):=\sum_{c\in \C : f_c=f}\alpha_c$, i.e., $w(f)$ is the sum of the dual values of the clients that chose $f$ as their close-by approximately paid facility. We can exactly compute those weights for all $f\in \S$ in $O(1)$ MPC rounds by using \Cref{lemma:groupwiseaggregation}.

    By using \Cref{lemma:fl-stepV}, we can assume that we have an explicitly constructed graph $H_{base}$ that is relatively sparse ($\leq n^{1+\delta}$ edges for an arbitrarily small constant $\delta>0$) and such that $H_{base}^2$ satisfies the criteria for graph $H'$ in \Cref{alg:fl-overview}. The cluster centers $\S_0$ have to form a set of facilities that are at pairwise distance at least $4$ in $H$ and such that all except a $1/\log n$-fraction of the node (facility) weights are within distance $O(1)$ in $H'$. We can enforce pairwise distance at least $4$ in $H$ by requiring pairwise distance at least $4$ in the supergraph $H'$ and thus pairwise distance at least $7$ in $H_{base}$. By using the algorithm of \Cref{lemma:weighted2rulingset} with parameters $t=6$ and $\eps=1/\log n$, we obtain a set that directly satisfies all the properties we need. We can therefore compute $\S_0$ in $O(\log\log n\cdot\log\log\log n)$ rounds in the MPC model.

    In the last part of Step VI, we have to build the clusters by assigning each facility $f\in \S$ to the cluster of the closest facility $f_0\in\S_0$. The cluster center that is at minimum distance from $f$ in $H_{base}$ is also at minimum distance from $f$ in $H'$. We can therefore assign all facilities to their cluster centers by running parallel BFS searches on $H_{base}$ from all cluster centers for at most $O(\log\log\log n)$ steps (we always only forward the first search that reaches a node). This can clearly be done in $O(\log\log\log n)$ MPC rounds.
\end{proof}

\begin{lemma}\label{lemma:fl-stepVII}
    Step VII of \Cref{alg:fl-overview} can be implemented in $\mathcal{O}(1)$ MPC rounds.
\end{lemma}

\begin{proof}
    The goal of Step VII is to open one facility in each cluster that approximately minimizes the total distance to the points in that cluster. Let each point be labeled according to its assigned cluster. Our objective is to compute, for each cluster, a single facility to open that best represents the cluster in terms of proximity to its members.

    We proceed in two phases:

    \textbf{Phase 1: Computing geometric centers of clusters.}
     We apply Lemma~\ref{lemma:groupwiseaggregation} to group points by cluster label and compute the arithmetic mean of the coordinate vectors of the points in each cluster in $O(1)$ MPC rounds. The result is a mean vector $\mu_f$ for each cluster $C_f$.

    \textbf{Phase 2: Selecting the representative facility.}
    For each cluster, we now need to select a facility from that cluster which is closest (in squared Euclidean distance) to the corresponding mean vector $\mu_f$. We again use Lemma~\ref{lemma:groupwiseaggregation} to group all facilities by cluster label and compute the facility minimizing the distance to $\mu_f$ in each group in $O(1)$ MPC rounds.

    Therefore, the full procedure for Step VII completes in $\mathcal{O}(1)$ MPC rounds.
\end{proof}

\subsection{Implementation of the {\it k}-Means Reduction in the MPC Model}
In this section, we describe how to implement the 3-step $k$-means algorithm from \Cref{sec:kmeans-highlevel-alg} in the fully scalable MPC model.

\begin{lemma}\label{lemma:kmeans-stepI}
    Step (I) of the $k$-means algorithm described in \Cref{sec:kmeans-highlevel-alg} can be implemented in $O(\log \log n \cdot \log\log\log n)$ MPC rounds.
\end{lemma}

\begin{proof}
    In this step, the goal is to compute two opening costs $\lambda_1$ and $\lambda_2$ such that the facility location algorithm instantiated with $\lambda_1$ opens at most $k$ facilities (i.e., $k_1 := |\F_1| \leq k$), and with $\lambda_2$ opens at least $k$ facilities (i.e., $k_2 := |\F_2| \geq k$), where $\lambda_2 \leq \lambda_1 \le 2\lambda_2$.

    Following \Cref{eq:pointset}, we know that for sufficiently large values of $\lambda = \mathrm{poly}(n)$, the facility location algorithm opens only a single facility, whereas for $\lambda = 1$, it opens all facilities. Therefore, there must exist some value of $\lambda$ in the interval $[1, \lambda_{\max}]$ at which the number of opened facilities transitions from more than $k$ to fewer than $k$.

    We discretize the interval $[1, \lambda_{\max}]$ into $O(\log n)$ geometrically spaced values of $\lambda$. Since the number of opened facilities is monotonically non-increasing in $\lambda$, we are guaranteed to find two consecutive values $\lambda_i$ and $\lambda_{i+1}$ such that one opens fewer than $k$ facilities and the other opens more than $k$. These two values therefore satisfy the condition that $\lambda_2 \leq \lambda_1\le 2\lambda_2$. If either value corresponds to a solution that opens exactly $k$ facilities, we return that facility location solution as the final $k$-means solution and terminate.

    Otherwise, we return the facility location instances corresponding to $\lambda_1$ and $\lambda_2$ as the output of Step (I). Since all facility location instances for the $O(\log n)$ candidate values of $\lambda$ are executed in parallel, the total number of MPC rounds is determined by the runtime of a single facility location computation. This results in a total round complexity of $O(\log \log n \cdot \log \log \log n)$. The global memory usage increases only by a factor of $O(\log n)$ due to the parallel execution over $O(\log n)$ instances.

\end{proof}

Now, we move towards the second step of the $k$-means algorithm explained in \Cref{sec:kmeans-highlevel-alg} and prove that it can be implemented in $O(1)$ MPC rounds. Specifically, we show that computing the set $\F'_2$ can be performed efficiently in parallel.

\begin{lemma}\label{lemma:kmeans-stepII}
    Step (II) of the $k$-means algorithm described in \Cref{sec:kmeans-highlevel-alg}—i.e., selecting a set $\F_2'$ of size $k_1$ containing a $\delta$-approximate nearest neighbor from $\F_2$ for each facility in $\F_1$—can be implemented in $\mathcal{O}(1)$ MPC rounds.
\end{lemma}
\begin{proof}
    The objective is to compute a set $\F_2' \subseteq \F_2$ of size $k_1$ such that for every facility $f \in \F_1$, there exists a facility $f_2'(f) \in \F_2$ included in $\F_2'$ that serves as a $\delta$-approximate nearest neighbor of $f$. If the total number of selected facilities is less than $k_1$, we complete the set $\F_2'$ by arbitrarily adding facilities from $\F_2 \setminus \F_2'$ until the desired size is reached.

    To compute the approximate nearest neighbors efficiently in the MPC model, we use the graph $G$ constructed as described in \Cref{lemma:graphapprox}. This graph is a sparse spanner, constructed via locality-sensitive hashing (LSH), where an edge $(f, f')$ exists if the distance $\dist(f, f')$ is known and sufficiently small. 
    
    We then apply \Cref{lemma:basicMins} with parameter $\ell = \mathcal{O}(1)$ to each $f \in \F_1$ to compute a facility $f_2'(f) \in \F_2$ such that
    \[
    \dist(f, f_2'(f)) \leq \delta \cdot \min_{f' \in \F_2} \dist(f, f'),
    \]
    where $\delta$ is the approximation factor guaranteed by the graph construction. According to \Cref{lemma:basicMins}, this step can be executed in $\mathcal{O}(1)$ MPC rounds using $\tilde{\mathcal{O}}(n)$ global memory.
    
    Next, using a convergecast over the facilities in $\F_2$, we compute the cardinality of the set $\F_2'$. If $|\F_2'| < k_1$, we need to add $k_1 - |\F_2'|$ additional facilities from $\F_2 \setminus \F_2'$ to ensure that the final set $\F_2'$ has exactly $k_1$ facilities. Since these additional facilities can be selected arbitrarily, we use a greedy selection strategy that can be implemented efficiently in the MPC model.

    To do so, we first sort the facilities in $\F_2$ according to a fixed total order (e.g., their IDs or spatial coordinates) in $\mathcal{O}(1)$ MPC rounds using standard parallel sorting algorithms. The sorted facilities are then evenly assigned to the machines. Each machine has a local memory of $O(n^\sigma)$ for some constant $\sigma \in (0,1)$, and there are at most $n$ facilities in $\F_2$, which guarantees that the total number of machines is $O(n^{1 - \sigma})$. This induces a logical tree over the machines with fan-out $n^\sigma$ and depth $O(1/\sigma) = O(1)$.
    
    In the first phase (bottom-up), we perform a convergecast operation along this tree to compute, for each internal node, the number of unselected facilities stored in its subtree. Since the tree has constant height, this aggregation step requires only $\mathcal{O}(1)$ MPC rounds.
    
    In the second phase (top-down), we traverse the tree in a greedy fashion to select the required number of additional facilities. The root node begins with the total demand of $k_1 - |\F_2'|$ and distributes this demand among its children proportionally to the number of unselected facilities in each subtree (or arbitrarily, if selection is unconstrained). This process is recursively continued by each internal node down to the leaves. When a leaf is reached, the machine locally selects the requested number of facilities from its memory.
    
    Because the tree has constant height and each level of communication can be executed in $\mathcal{O}(1)$ MPC rounds, the entire padding step also completes in $\mathcal{O}(1)$ rounds. At the end of this process, the set $\F_2'$ contains exactly $k_1$ facilities, as required.
\end{proof}

We now consider Step (III) of the $k$-means algorithm described in \Cref{sec:kmeans-highlevel-alg}. In this step, the final set of $k$ centers is constructed by sampling from the facility sets $\F_1$, $\F_2'$, and $\F_2 \setminus \F_2'$, according to a convex combination determined by $a := \frac{k_2 - k}{k_2 - k_1}$ and $b := \frac{k - k_1}{k_2 - k_1} = 1 - a$. We show that this step can be implemented efficiently in the MPC model.

\begin{lemma}\label{lemma:kmeans-stepIII}
    Step (III) of the $k$-means algorithm described in \Cref{sec:kmeans-highlevel-alg} can be implemented in $\mathcal{O}(1)$ MPC rounds.
\end{lemma}

\begin{proof}
    The goal of this step is to sample a final set $Z$ of $k$ facilities that defines the output clustering. This is done in two stages:

    \begin{enumerate}[(A)]
        \item With probability $a := \frac{k_2 - k}{k_2 - k_1}$, the set $\F_1$ is selected (i.e., all facilities in $\F_1$ are added to $Z$). Otherwise, with probability $b := 1 - a$, the set $\F_2'$ is selected and added to $Z$.
        \item Additionally, $k - k_1$ facilities are sampled uniformly at random from $\F_2 \setminus \F_2'$ and added to $Z$.
    \end{enumerate}
    
    Both of these sampling steps can be implemented in constant MPC rounds as follows:
    
    \paragraph{Step (A): Sampling either $\F_1$ or $\F_2'$.}
    Since the sampling decision is based on a single global coin toss (deciding whether to choose $\F_1$ or $\F_2'$), we can broadcast the outcome of this coin toss to all machines in $\mathcal{O}(1)$ rounds. The selected set is then included in $Z$ by every machine that holds elements of it.

    \paragraph{Step (B): Uniform Sampling of $k - k_1$ Facilities from $\F_2 \setminus \F_2'$.}
Unlike Step (II), where additional facilities were added arbitrarily, here we must sample exactly $k - k_1$ facilities uniformly at random from $\F_2 \setminus \F_2'$.

To do this efficiently in MPC, we distribute the elements of $\F_2 \setminus \F_2'$ across machines and organize these machines into a virtual tree of height $\mathcal{O}(1)$ and fan-out $n^{\sigma}$, where $\sigma \in (0,1)$ is determined by the available local memory $O(n^\sigma)$. The sampling process proceeds in two phases:

\begin{itemize}
    \item \textbf{Bottom-Up (Convergecast Phase).}  
    Each machine computes the number of unmarked facilities it stores locally. This count is aggregated up the tree using a convergecast to compute, for each internal node, the total number of facilities in its subtree. This requires only $\mathcal{O}(1)$ rounds because the tree has constant height.

    \item \textbf{Top-Down (Sampling Phase).}  
    Once the subtree sizes are known, we begin the sampling process at the root. The root node samples how many of the $k - k_1$ points should be taken from each of its child subtrees, proportional to the size of each subtree. These target counts are passed recursively down the tree. At the leaves (i.e., the machines storing the facilities), each machine samples the required number of facilities uniformly at random from its local memory.
    \end{itemize}
    
    This two-phase sampling procedure ensures that the final $k - k_1$ facilities are drawn uniformly at random from $\F_2 \setminus \F_2'$ and completes in $\mathcal{O}(1)$ MPC rounds due to the bounded tree height and constant-round communication per level.
    
    \paragraph{Final Client Assignments.}
    Once the full set $Z$ of $k$ centers has been determined, each client $c$ must be connected to a facility $\phi(c) \in Z$ according to the following rules:
    \begin{itemize}
        \item If $f_c^{(2)} \in \F_2'$ and $\F_1$ is chosen in Step (A), then $\phi(c) := f_c^{(1)}$.
        \item If $f_c^{(2)} \in \F_2'$ and $\F_2'$ is chosen in Step (A), then $\phi(c) := f_c^{(2)}$.
        \item If $f_c^{(2)} \in \F_2 \setminus \F_2'$:
        \begin{itemize}
            \item If $f_c^{(2)}$ was among the $k - k_1$ sampled facilities in Step (B), then $\phi(c) := f_c^{(2)}$.
            \item If not, and $\F_1$ was chosen in Step (A), then $\phi(c) := f_c^{(1)}$.
            \item Otherwise, set $\phi(c) := f_2'(f_c^{(1)})$.
        \end{itemize}
    \end{itemize}
    
    All of this logic can be implemented locally using previously computed mappings (from Steps I and II) and a broadcasted sampling outcome. Therefore, this final step also requires only $\mathcal{O}(1)$ MPC rounds.
\end{proof}

\begin{lemma}\label{lem:highprob}
By repeating the randomized steps of the k-means algorithm explained in \Cref{sec:kmeans-highlevel-alg} independently $O(\log n)$ times and selecting the solution with the minimum total connection cost among these repetitions, we can ensure that with high probability, the total cost of the solution is within a constant factor $K$ of the optimal cost.
\end{lemma}

\begin{proof}
From \Cref{thm:kmeansapprox}, we know that the expected total connection cost $\mathbb{E}[\text{COST}]$ of a single run of k-means algorithm satisfies:
\begin{align*}
    \mathbb{E}[\text{COST}] \leq C \cdot \text{OPT},
\end{align*}
where $C$ is the approximation factor, and $\text{OPT}$ is the optimal cost of the $k$-means problem.

Our goal is to strengthen this guarantee from expectation to high probability. To achieve this, we employ probabilistic analysis using Markov's inequality.

\paragraph{Bounding the Probability for a Single Run.}
Using Markov's inequality, the probability that the total cost $\text{COST}$ exceeds $K \cdot \text{OPT}$ in a single run is:
\begin{align*}
    \Pr\left[ \text{COST} > K \cdot \text{OPT} \right] \leq \frac{\mathbb{E}[\text{COST}]}{K \cdot \text{OPT}} \leq \frac{C}{K}.
\end{align*}
By choosing $K = 2C$, we obtain:
\begin{align*}
    \Pr\left( \text{COST} > K \cdot \text{OPT} \right) \leq \frac{1}{2}.
\end{align*}

\paragraph{Repeating the Algorithm.}
We independently repeat the randomized portion of \Cref{sec:kmeans-highlevel-alg} $T = \lambda \log n$ times in parallel, where $\lambda \geq 1$ is a constant to be determined later.

\paragraph{Bounding the Probability Over Multiple Repetitions.}
Since the repetitions are independent, the probability that \emph{all} repetitions result in a total cost exceeding $K \cdot \text{OPT}$ is:
\begin{align*}
    \Pr\left( \bigcap_{t=1}^{T} \left\{ \text{COST}_t > K \cdot \text{OPT} \right\} \right) &= \left( \Pr\left( \text{COST}_t > K \cdot \text{OPT} \right) \right)^{T} \leq \left( \frac{1}{2} \right)^{T} = 2^{-T}.
\end{align*}
By setting $T = \lambda \log n$, we have:
\begin{align*}
    \Pr\left( \text{No run has } \text{COST}_t \leq K \cdot \text{OPT} \right) &\leq n^{-\lambda}.
\end{align*}

\paragraph{Selecting An Approximately Best Solution.}
By selecting the solution with the minimum total cost among the $T$ repetitions, we ensure that with high probability (at least $1 - 1/n^\lambda$), the total cost is at most $K \cdot \text{OPT}$. We cannot compute the exact cost of a given solution. By using our graph approximation of the Euclidean metric space (cf.\ \Cref{lemma:graphapprox,lemma:basicSums}), we can however compute the cost of a given solution up to a constant factor. We can therefore find an approximately best solution among the $T$ repetitions, which still leads to a constant approximation of the $k$-means problem w.h.p.

\paragraph{Conclusion.}
By choosing $\lambda \geq 2$, the failure probability becomes negligible for large $n$. Therefore, by repeating the algorithm $O(\log n)$ times and selecting the best solution, we obtain a solution whose total cost is within a constant factor $K$ of the optimal cost with high probability, completing the proof.
\end{proof}

\subsection{Proof of the Main Theorem}\label{sec:mainproof}
It remains to prove the main theorem that is stated in \Cref{sec:intro}. It follows by \Cref{thm:approx_fl_highlevel,thm:kmeansapprox} that the highlevel algorithms described in \Cref{sec:fl-highlevel,sec:kmeans-highlevel-alg} achieve a constant approximation for the facility location problem and the $k$-means problem. While \Cref{thm:kmeansapprox} only proves that the $k$-means algorithm achieves a constant approximation in expectation, by running the algorithm $O(\log n)$ in parallel (cf.\ \Cref{lem:highprob}), the constant approximation for the $k$-means problem can also be achieved w.h.p. In \Cref{lemma:fl-stepI,lemma:fl-stepII,lemma:fl-stepIII,lemma:fl-stepIV,lemma:fl-stepV,lemma:fl-stepVI,lemma:fl-stepVII}, it is shown that facility location algorithm of \Cref{sec:fl-highlevel} can be implemented in the MPC model in a fully scalable way in $O(\log\log n\cdot\log\log\log n)$ rounds. Similarly, in \Cref{lemma:kmeans-stepI,lemma:kmeans-stepII,lemma:kmeans-stepIII}, it is shown that the $k$-means algorithm of \Cref{sec:kmeans-highlevel-alg} can be implemented in $O(1)$ MPC rounds. Altogether, this proves the claim of the main theorem.\hspace*{\fill}\qed

\newpage
\bibliographystyle{alpha}
\bibliography{arxiv}

\end{document}